\documentclass[sigconf, nonacm]{acmart}
\usepackage{subfigure}
\usepackage{array}
\usepackage{pdfpages}
\usepackage{tabularx}
\usepackage{textcomp}
\usepackage{stfloats}
\usepackage{url}
\usepackage{verbatim}
\usepackage{graphicx}
\usepackage{mathrsfs}

\usepackage{xcolor} 
\usepackage{caption}
\usepackage{float} 
\usepackage{hyperref}
\usepackage{subcaption}
\usepackage{grffile}
\usepackage{algorithm,algorithmicx}
\usepackage{algpseudocode}
\usepackage{CJKutf8}
\usepackage{amsfonts,amsthm,amsmath}
\usepackage{multirow}
\usepackage{makecell}
\usepackage{pifont}
\usepackage{tikz}
\usepackage{mathrsfs}
\usepackage{threeparttable}
\usepackage{balance}

\PassOptionsToPackage{table,xcdraw}{xcolor}
\usepackage{colortbl}
\usepackage{xcolor}

\DeclareFontFamily{U}{rsfs}{\skewchar\font127 }
\DeclareFontShape{U}{rsfs}{m}{n}{%
   <-6> rsfs5
   <6-8> rsfs7
   <8-> rsfs10
}{}
\makeatletter
\newenvironment{breakablealgorithm}
{
		\begin{center}
			\refstepcounter{algorithm}
			\hrule height.8pt depth0pt \kern2pt
			\renewcommand{\caption}[2][\relax]{
				{\raggedright\textbf{\ALG@name~\thealgorithm} ##2\par}%
				\ifx\relax##1\relax 
				\addcontentsline{loa}{algorithm}{\protect\numberline{\thealgorithm}##2}%
				\else 
				\addcontentsline{loa}{algorithm}{\protect\numberline{\thealgorithm}##1}%
				\fi
				\kern2pt\hrule\kern2pt
			}
		}{
		\kern2pt\hrule\relax
	\end{center}
}
\makeatother
\DeclareFontFamily{U}{rsfs}{\skewchar\font127 }
\DeclareFontShape{U}{rsfs}{m}{n}{%
   <-6.5> rsfs5
   <6.5-8> rsfs7
   <8-> rsfs10
}{}
\newcommand\vldbdoi{XX.XX/XXX.XX}

\newcommand\vldbavailabilityurl{URL_TO_YOUR_ARTIFACTS}

\newcommand*\emptycirc[1][1ex]{\tikz\draw (0,0) circle (#1);} 
\newcommand*\halfcirc[1][1ex]{%
	\begin{tikzpicture}
	\draw[fill] (0,0)-- (90:#1) arc (90:270:#1) -- cycle ;
	\draw (0,0) circle (#1);
	\end{tikzpicture}}
\newcommand*\fullcirc[1][1ex]{\tikz\fill (0,0) circle (#1);} 

\newtheorem{theorem}{Theorem}
\newtheorem{lemma}{Lemma}
\newtheorem{definition}{Definition}
\newtheorem{problem}{Problem}
\floatname{algorithm}{Algorithm} 

\DeclareFontFamily{U}{rsfs}{\skewchar\font127 }
\DeclareFontShape{U}{rsfs}{m}{n}{%
   <-6> rsfs5
   <6-8> rsfs7
   <8-> rsfs10
}{}

\begin{document}
\begin{sloppypar}

\title{EVA-S3PC: Efficient, Verifiable, Accurate Secure Matrix Multiplication Protocol Assembly and Its Application in Regression}

\author{Shizhao Peng}
\orcid{0000-0001-6333-5703}
\affiliation{%
  \institution{Beihang University}
  \streetaddress{37 Xueyuan Road}
  \city{Beijing}
  \country{China}
  \postcode{100191}
}
\email{by1806167@buaa.edu.cn}

\author{Tianrui Liu}
\orcid{0000-0002-6268-7018}
\affiliation{%
  \institution{Beihang University}
  \streetaddress{37 Xueyuan Road}
  \city{Beijing}
  \country{China}
  \postcode{100191}
}
\email{liu_tianrui@buaa.edu.cn}

\author{Tianle Tao}
\orcid{0009-0005-1218-1923}
\affiliation{%
  \institution{Beihang University}
  \streetaddress{37 Xueyuan Road}
  \city{Beijing}
  \country{China}
  \postcode{100191}
}
\email{taotianle@buaa.edu.cn}

\author{Derun Zhao}
\orcid{0009-0008-2294-9723}
\affiliation{%
  \institution{Beihang University}
  \streetaddress{37 Xueyuan Road}
  \city{Beijing}
  \country{China}
  \postcode{100191}
}
\email{derunz@buaa.edu.cn}

\author{Hao Sheng}
\orcid{0000-0002-2811-8962}
\affiliation{%
  \institution{Beihang University}
  \streetaddress{37 Xueyuan Road}
  \city{Beijing}
  \country{China}
  \postcode{100191}
}
\email{shenghao@buaa.edu.cn}

\author{Haogang Zhu*}
\orcid{0000-0003-1771-2752}
\affiliation{%
  \institution{Beihang University}
  \streetaddress{37 Xueyuan Road}
  \city{Beijing}
  \country{China}
  \postcode{100191}
}
\email{haogangzhu@buaa.edu.cn}

\begin{abstract}
Efficient multi-party secure matrix multiplication is crucial for privacy-preserving machine learning, but existing mixed-protocol frameworks often face challenges in balancing security, efficiency, and accuracy. This paper presents an efficient, verifiable and accurate secure three-party computing (EVA-S3PC) framework that addresses these challenges with elementary 2-party and 3-party matrix operations based on data obfuscation techniques. We propose basic protocols for secure matrix multiplication, inversion, and hybrid multiplication, ensuring privacy and result verifiability. Experimental results demonstrate that EVA-S3PC achieves up to 14 significant decimal digits of precision in Float64 calculations, while reducing communication overhead by up to $54.8\%$ compared to state of art methods. Furthermore, 3-party regression models trained using EVA-S3PC on vertically partitioned data achieve accuracy nearly identical to plaintext training, which illustrates its potential in scalable, efficient, and accurate solution for secure collaborative modeling across domains.
\end{abstract}

\maketitle



\section{Introduction}
The rapid advancement of digital technologies such as AGI (Artificial General Intelligence), IoT (Internet of Things), and cloud computing makes data a fundamental production factor in digital economy. Government organizations and businesses nowadays use centralized cloud services \cite{6424959} for data collaboration among multiple entities, which produced promising applications in healthcare, finance, and governance\cite{10.1145/3158363}, but this routing in general suffers from the problem of privacy. \autoref{fig:S3PC Modeling Problem} illustrates a ideal scenario of three financial institutions (P1: Bank with corporate financial data, P2: Rating Agency with credit rating, P3: Insurance holding a ten-year default record as labels) aiming to conduct joint regression analysis based on heterogeneously distributed data without sharing raw data directly. While deep learning as a service (DLaaS) \cite{sekar2023deep} can explore the value of their respective data in financial risk control scenarios by pooling data on a central service, most real life application would request the analysis to be carried out without exchanging any raw data to avoid privacy breaches.

\begin{figure}[tp]
  \centering
     \quad
  \includegraphics[width=1.0\hsize]{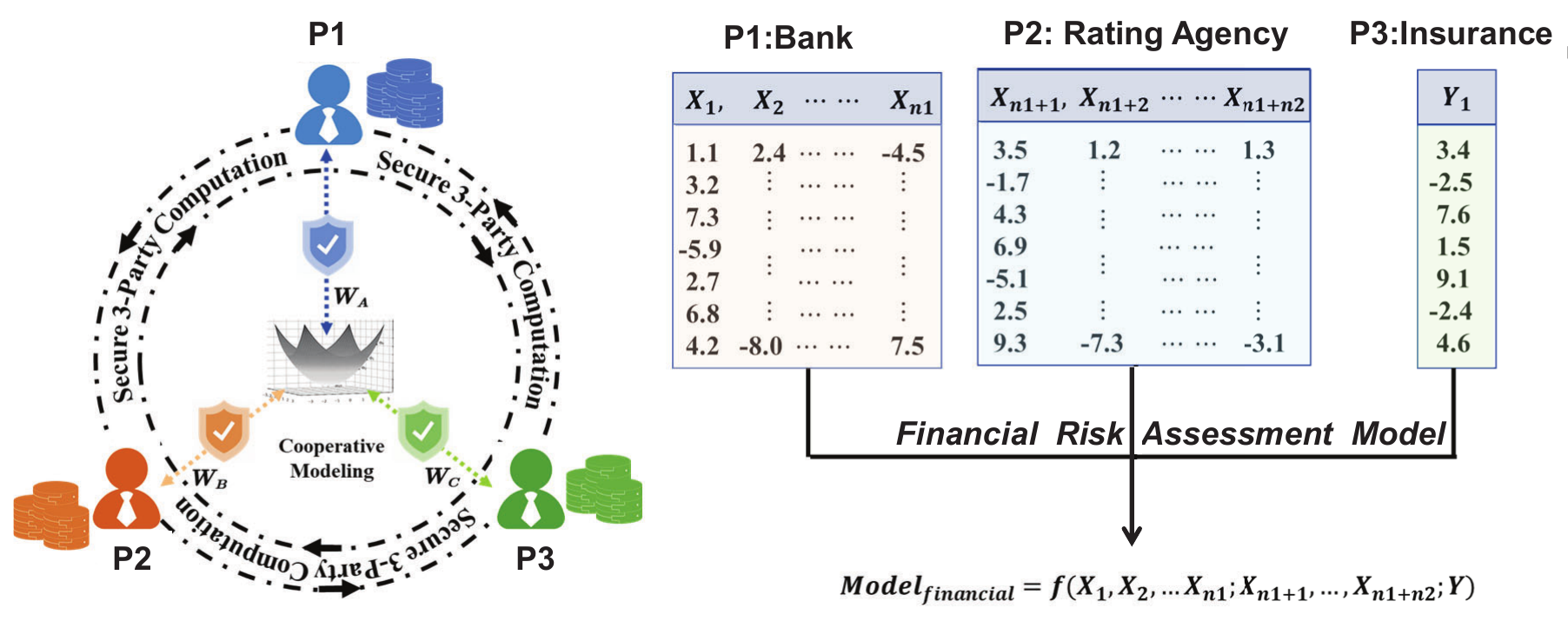}
     \quad
  \vspace{-0.3cm}
  \caption{Secure Three-Party Cooperative Modeling Problem}
  \label{fig:S3PC Modeling Problem}
  \Description[<short description>]{<long description>}
\end{figure}

In recent years, frequent data security incidents and heightened public awareness of personal privacy have led governments worldwide to implement regulations (e.g., the well-known GDPR \cite{voigt2017eu}) to protect data ownership and user rights. These regulations have made efficient data sharing across organizations more challenging, creating isolated "data islands" where vast amounts of data remain untouched. Therefore, achieving decentralized computation and collaborative modeling without compromising privacy has become a major research challenge for academia and industry. Current mainstream research directions include \textit{Secure Multi-party Computation (SMPC)} \cite{zhou2024secure,mohanta2020multi}, \textit{Homomorphic Encryption (HE)} \cite{munjal2023systematic}, \textit{Differential Privacy(DP)} \cite{zhao2022survey}, and \textit{privacy-preserving machine learning (PPML)} frameworks based on various cryptographic primitives. A critical component in all of these approaches is secure matrix multiplication, which is fundamental to many machine learning algorithms. Examples include calculating gain coefficients in decision trees \cite{abspoel2021secure}, computing Euclidean distances in K-nearest neighbor classification \cite{haque2020privacy}, and solving coefficients in linear regression \cite{aono2017input}.

\begin{table*}[htbp]
  \centering
  \caption{Comparison of various mixed-protocols SMPC frameworks related to 2PC and 3PC} 
    \resizebox{\linewidth}{!}{
    \begin{threeparttable}
    \begin{tabular}{cccccccccccccc}
    \toprule
    \multirow{2}[3]{*}{} & \multirow{2}[3]{*}{\textbf{Framework}} & \multicolumn{1}{c}{\multirow{2}[3]{*}{\textbf{Dev. Language}}} & \multicolumn{1}{c}{\multirow{2}[3]{*}{\textbf{Techniques Used}}} & \multicolumn{2}{c}{\textbf{Operator Supported}} & \multicolumn{1}{c}{\multirow{2}[3]{*}{\textbf{Security}}} & \multicolumn{3}{c}{\textbf{Theoretical Evaluation}} &       & \multicolumn{3}{c}{\textbf{Pratical Performance}} \\
\cmidrule{5-6}\cmidrule{8-10}\cmidrule{12-14}          &       &       &       & \multicolumn{1}{c}{\textbf{Mat.Mul}} & \textbf{Mat.Inv} &       & \textbf{Complexity} & \multicolumn{1}{c}{\textbf{Verifiablity}} & \multicolumn{1}{c}{\textbf{ Extensibility}} &       & \textbf{Float64} & \textbf{Comm.} & \multicolumn{1}{c}{\textbf{Big Matrix}} \\
    \midrule
    \multirow{8}[1]{*}{\textbf{2PC}} & \multicolumn{1}{c}{\textbf{SecureML\cite{mohassel2017secureml}}} & \textbf{C++} & \textbf{HE,GC,SS} & \textbf{\ding{52}} & \textbf{\ding{56}} & \fullcirc \:or \halfcirc & \textbf{High} & \textbf{\ding{56}} & \halfcirc &       & \halfcirc & \textbf{High} & \textbf{\ding{56}} \\
          & \textbf{Crypten\cite{knott2021crypten}} & \textbf{python} & \textbf{SS} & \textbf{\ding{52}} & \textbf{\ding{52}} & \halfcirc & \textbf{Medium} & \textbf{\ding{56}} & \fullcirc &       & \emptycirc & \textbf{Low} & \textbf{\ding{52}} \\
          & \multicolumn{1}{c}{\textbf{Delphi\cite{mishra2020delphi}}} & \textbf{C++} & \textbf{HE,GC,SS} & \textbf{\ding{52}} & \textbf{\ding{52}} & \halfcirc & \textbf{High} & \textbf{\ding{56}} & \halfcirc &       & \halfcirc & \textbf{High} & \textbf{\ding{56}} \\
          & \multicolumn{1}{c}{\textbf{Motion\cite{braun2022motion}}} & \textbf{C++} & \textbf{OT,GC,SS} & \textbf{\ding{52}} & \textbf{\ding{56}} & \halfcirc & \textbf{Medium} & \textbf{\ding{56}} & \halfcirc &       & \halfcirc & \textbf{High} & \textbf{\ding{56}} \\
          & \textbf{Pencil\cite{liu2024pencil}} & \textbf{C++} & \textbf{HE,DP} & \textbf{\ding{52}} & \textbf{\ding{56}} & \fullcirc \:or \halfcirc & \textbf{High} & \textbf{\ding{56}} & \halfcirc &       & \halfcirc & \textbf{High} & \textbf{\ding{56}} \\
          & \textbf{LibOTe\cite{libOTe}} & \textbf{C++} & \textbf{OT} & \textbf{\ding{52}} & \textbf{\ding{52}} & \halfcirc & \textbf{Medium} & \textbf{\ding{56}} & \halfcirc &       & \fullcirc & \textbf{High} & \textbf{\ding{56}} \\
          & \multicolumn{1}{c}{\textbf{Chameleon\cite{riazi2018chameleon}}} & \textbf{C++} & \textbf{GC,SS} & \textbf{\ding{52}} & \textbf{\ding{56}} & \halfcirc & \textbf{Medium} & \textbf{\ding{56}} & \halfcirc &       & \emptycirc & \textbf{Medium} & \textbf{\ding{52}} \\
          & \multicolumn{1}{c}{\textbf{Du's work\cite{du2002practical}}} & \textbf{N/A} & \textbf{DD} & \textbf{\ding{52}} & \textbf{\ding{52}} & \halfcirc & \textbf{Low} & \textbf{\ding{56}} & \fullcirc &       & \fullcirc & \textbf{Low} & \textbf{\ding{52}} \\
    \midrule
    \multirow{8}[1]{*}{\textbf{3PC}} & \textbf{ABY3\cite{mohassel2018aby3}} & \textbf{C++} & \textbf{GC,SS} & \textbf{\ding{52}} & \textbf{\ding{56}} & \halfcirc & \textbf{Medium} & \textbf{\ding{56}} & \halfcirc &       & \halfcirc & \textbf{Medium} & \textbf{\ding{52}} \\
          & \multicolumn{1}{c}{\textbf{Kumar's Work\cite{kumar2017privacy}}} & \textbf{Matlab} & \textbf{DD} & \textbf{\ding{52}} & \textbf{\ding{52}} & \halfcirc & \textbf{Low} & \textbf{\ding{52}} & \emptycirc &       & \fullcirc & \textbf{Low} & \textbf{\ding{52}} \\
          & \textbf{Daalen's Work\cite{van2023privacy}} & \textbf{N/A} & \textbf{DD} & \textbf{\ding{52}} & \textbf{\ding{56}} & \halfcirc & \textbf{Low} & \textbf{\ding{56}} & \emptycirc &       & \fullcirc & \textbf{Low} & \textbf{\ding{52}} \\
          & \textbf{MP-SPDZ\cite{keller2020mp}} & \textbf{C++} & \textbf{HE,OT,SS} & \textbf{\ding{52}} & \textbf{\ding{56}} & \fullcirc \:or \halfcirc & \textbf{High} & \textbf{\ding{56}} & \halfcirc &       & \halfcirc & \textbf{High} & \textbf{\ding{56}} \\
          & \textbf{Tenseal\cite{benaissa2021tenseal}} & \textbf{python} & \textbf{HE} & \textbf{\ding{52}} & \textbf{\ding{56}} & \fullcirc \:or \halfcirc & \textbf{High} & \textbf{\ding{56}} & \fullcirc &       & \halfcirc & \textbf{High} & \textbf{\ding{56}} \\
          & \textbf{FATE\cite{FATE}} & \textbf{python} & \textbf{HE,SS} & \textbf{\ding{52}} & \textbf{\ding{56}} & \fullcirc \:or \halfcirc & \textbf{High} & \textbf{\ding{56}} & \fullcirc &       & \halfcirc & \textbf{High} & \textbf{\ding{56}} \\
          & \textbf{SecretFlow\cite{ma2023secretflow}} & \textbf{python} & \textbf{HE,GC,SS} & \textbf{\ding{52}} & \textbf{\ding{52}} & \fullcirc \:or \halfcirc & \textbf{High} & \textbf{\ding{56}} & \fullcirc &       & \halfcirc & \textbf{High} & \textbf{\ding{56}} \\
          \rowcolor{gray!30}
          & \textbf{EVA-S3PC} & \textbf{python} & \textbf{DD} & \textbf{\ding{52}} & \textbf{\ding{52}} & \halfcirc & \textbf{ Low} & \textbf{\ding{52}} & \fullcirc &       & \fullcirc & \textbf{Low} & \textbf{\ding{52}} \\
    \bottomrule
    \end{tabular}
        \begin{tablenotes}
            \item \textbf{Note:} {\textbf{Dev. Language} indicates the main development language for each framework. \textbf{Mat.Mul} and \textbf{Mat.Inv} indicates multiplication and inversion operator for matrix. \textbf{\ding{52}} and \textbf{\ding{56}}} indicates the framework supports or not supported the feature. For \textbf{Security}, \:\fullcirc\:/ \halfcirc\: denotes the framework can against malicious adversary or semi-honest adversary. \:\fullcirc\:/ \halfcirc\:/ \emptycirc\: in \textbf{Extensibility} refers to the difficulty level of coupling the framework with other SMPC protocols based on existing languages and compiler. For Float64, these symbols indicate the precision degree(high, medium or low) under this data type. \textbf{Comm.} here indicates the communication overhead.   
        \end{tablenotes}   
    \end{threeparttable}}%
  \label{tab:SMPC_Comparison}%
\end{table*}%

We organize representative frameworks supporting 2-party or 3-party matrix multiplication operators in Table \ref{tab:SMPC_Comparison}, and evaluate 16 cutting-edge frameworks based on development language, technique type, security level, theoretical assessment metrics (Float64 computational precision, computational complexity, communication overhead), and practical performance indicators (support for result verification, capability for large matrix calculations, and extensibility with various security protocols).
For frameworks incorporating HE primitives, such as \cite{keller2020mp,benaissa2021tenseal}, their security is intrinsically high due to the NP-hard nature of the underlying computational problem. However, HE involves extensive calculations over long ciphertext sequences and can only approximate most non-linear operators, resulting in high computational complexity and low precision, making it less suitable for precise computation of large matrices. In contrast, frameworks like Chameleon and ABY3 \cite{riazi2018chameleon,mohassel2018aby3} integrate Garbled Circuits (GC) for efficient Boolean operations and Secret Sharing (SS) for arithmetic operations, balancing efficiency in logic and arithmetic computations. Nevertheless, additive secret sharing (ASS) and replicated secret sharing (RSS) require conversions between various garbled circuits during large-scale matrix multiplications, introducing substantial communication overhead \cite{knott2021crypten}. Additionally, MSB truncation protocols \cite{wagh2019securenn} tend to accumulate errors when computing floating-point numbers outside the ring $Z_{2}^p$, limiting their applicability in non-linear and hybrid multiplication types (e.g., $\mathsection{\ref{S2PI}\sim \ref{S3PHM}}$). Frameworks relying on DP, while effectively concealing sensitive information through noise, can accumulate errors in matrix multiplications, potentially affecting computational accuracy in practical applications. Although OT-based frameworks \cite{rabin2005exchange,yadav2022survey} ensure high precision and secure data transfer during matrix multiplication, communication costs grow linearly with matrix size, restricting efficiency and scalability in complex computational scenarios.
Evidently, while these frameworks generally ensure security in semi-honest three-party settings, few solutions balance precision, complexity, and communication overhead while also supporting result verification, extensibility, and scalability for large matrices in practical applications.

\noindent Data disguising is a commonly used linear space random perturbation method, including linear transformation disguising, Z+V aggregation disguising, and polynomial mapping disguising \cite{du2002practical}. Based on real-number domains $\mathbb{R}$, it preserves matrix homogeneous by disguising original data, thus ensuring privacy protection for matrix input and output with minimal interaction and flexible asynchronous computation. This approach avoids the computational complexity of HE ciphertext and the significant communication overhead associated with GC circuit transformations, while also maintaining high precision and timeliness.
Existing research in this area primarily stems from the CS (Commodity Server) model built on Beaver’s \cite{beaver1997commodity,beaver1998server} multiplication triples, introducing a Third Trust Party (TTP) not involved in the actual computation to ensure process security. However, prior work \cite{du2001privacy,atallah2001secure} often remains at the theoretical level, focusing on 2-party vector or array operations and linear equation solutions without rigorous security proofs or practical performance evaluations (running time, precision, overhead).
To further apply data disguising techniques to various scientific computing problems and more complex three-party practical applications, this paper presents a secure 3-party computation framework under a semi-honest setting with verifiable results. The main contributions of this work are as follows:
\begin{itemize}
    \item A secure 3-party computing (S3PC) framework EVA-S3PC under the semi-honest adversary model is proposed with five elementary protocols: Secure 2-Party Multiplication (S2PM), Secure 3-Party Multiplication (S3PM), Secure 2-Party Inversion (S2PI), Secure 2-Party Hybrid Multiplication (S2PHM), Secure 3-Party Hybrid Multiplication Protocol(S3PHM), using data disguising technique. Rigorous security proofs are provided  based on computation indistinguishability theory, and correctness is also proved in semi-honest environment.
    \item A secure and efficient result validation protocol using Monte Carlo method is proposed to detect abnormality in the result produced by computation made of aforementioned elementary protocols.
    \item A typical linear regression model with data features and labels spited on 3 nodes was constructed using the elementary protocols under the framework, including the algorithm for Secure 3-Party Linear Regression Training (S3PLRT) and Secure 3-Party Linear Regression Prediction (S3PLRP) respectively.
    \item Theoretical and practical computational complexity, communication overhead, computing precision and prediction accuracy are analyzed and compared with representative SMPC models.
\end{itemize}
\noindent\textbf{Organizations:}
The remainder of this paper proceeds as follows. Section \ref{Related Work} presents recent advancements in secure multi-party computation for matrix multiplication, inversion, and regression analysis. Section \ref{Preliminaries} introduces the framework of S3PC and review some essential preliminaries. In Section \ref{Proposed-Work}, we describe the proposed elementary protocols S2PM, S3PM, S2PI, S2PHM, S3PHM with security proof. Section \ref{Applications} constructs S3PLRT and S3PLRP using building blocks from Section \ref{Proposed-Work}. Section \ref{Theoretical} provides theoretical analysis of computational and communication complexity, followed by experimental comparison of performance and precision with various secure computing schemes in Sections \ref{Experiments}. Finally, some conclusions are drawn in Section \ref{Conclusion}.
\section{Related Work}\label{Related Work}
\subsection{Secure Matrix Multiplication}
S2PM and S3PM are the foundational linear computation protocols within our framework, from which all other sub-protocols can be derived. Existing research on secure matrix multiplication primarily utilizes SMPC and data disguising techniques. Frameworks such as Sharemind \cite{du2001privacy,bogdanov2008sharemind} achieve secure matrix multiplication by decomposing matrices into vector dot products $\alpha_i \cdot \beta_i$. Each entry $\alpha_i \cdot \beta_i$ is computed locally, and re-sharing in the ring $\mathbb{Z}_{2^{32}}$ is completed through six Du-Atallah protocols. Other studies, including \cite{braun2022motion,wagh2019securenn,tan2021cryptgpu}, reduce the linear complexity of matrix multiplication by precomputing random triples $\langle a \rangle^s, \langle b \rangle^s, \langle c \rangle^s$ (where $S$ denotes additive secret sharing over $\mathbb{Z}_{2^l}$). 
In \cite{miller2021simple,furukawa2017high,guo2020efficient}, optimizations to the underlying ZeroShare protocol employ AES as a PRNG in ECB mode to perform tensor multiplication, represented as \textbf{$\langle z \rangle = \langle x \cdot y \rangle$}. The DeepSecure and XOR-GC frameworks \cite{10.1145/3195970.3196023,kolesnikov2008improved} utilize custom libraries and standard logic synthesis tools to parallelize matrix multiplication using GC in logic gate operations, enhancing computational efficiency. ABY3\cite{mohassel2018aby3} combines GC and SS methods, introducing a technique for rapid conversion between arithmetic, binary, and Yao's 3PC representations, achieving up to a $50\times$  reduction in communication for matrix multiplications.
Other approaches, including \cite{benaissa2021tenseal}, use CKKS encryption for vector-matrix multiplication, expanding ciphertext slots by replicating input vectors to accommodate matrix operations.
LibOTe \cite{libOTe} implements a highly efficient 1-out-of-n OT by adjusting the Diffie-Hellman key-exchange and optimizing matrix linear operations with a $64 \times 64 \rightarrow 128$-bit serial multiplier in $\mathbb{F}_{2^{255}-19}$ for precision. The SPDZ framework and its upgrades \cite{keller2018overdrive,baum2019using} enhance efficiency and provable security by integrating hidden key generation protocols for BGV public keys, combining the strengths of HE, OT, and SS. Frameworks like SecretFlow, Secure ML, Chameleon, and Delphi \cite{mishra2020delphi,riazi2018chameleon,ma2023secretflow,7958569} integrate sequential interactive GMW, fixed-point ASS, precomputed OT, and an optimized STP-based vector dot product protocol for matrix multiplication, achieving significant improvements in communication overhead and efficiency. In earlier work, Atallah \cite{atallah2010securely} proposed secure vector multiplication methods using data disguising for statistical analysis and computational geometry, while Du and Dallen \cite{van2023privacy} introduced 2-party and n-party diagonal matrix multiplication within the CS model for multivariate data mining. Clifton and Zhan \cite{du2002practical,vaidya2002privacy} further optimized these algorithms to improve time complexity and extend applicability in n-party computing scenarios.
\subsection{Secure Inverse and Regression Analysis}
\noindent Nikolaenko \cite{nikolaenko2013privacy} proposes a server-based privacy-preserving linear regression protocol for horizontally partitioned data, combining linearly homomorphic encryption (LHE) and GC. 
For vertically partitioned data, Giacomelli and Gasc{\'o}n \cite{gascon2016privacy,giacomelli2018privacy} utilize Yao's circuit protocol and LHE, incurring high overhead and limited non-linear computation precision. 
In contrast, ABY3 \cite{mohassel2018aby3} reduces regression communication complexity using delayed re-sharing techniques. Building on this, Mohassel further improves linear regression accuracy with an approximate fixed-point multiplication method that avoids Boolean truncation \cite{mohassel2017secureml}.
Gilad \cite{gilad2019secure} presents the first three-server linear regression model, yet heavy GC use limits scalability due to high communication costs.
Liu \cite{liu2024pencil} combines DP with HE and SS to support linear regression across vertically and horizontally partitioned models, protecting model parameter privacy.
Rathee and Tan \cite{tan2021cryptgpu,rathee2020cryptflow2} leverage GPU acceleration and fixed-point arithmetic over shares, using 2-out-of-3 replicated secret sharing to support secure regression across three-party servers.
Ma \cite{ma2023secretflow} introduces the first SPU-based, MPC-enabled PPML framework, with compiler optimizations that substantially enhance training efficiency and usability in secure regression.

\section{Framework and PRELIMINARIES}\label{Preliminaries}
\noindent In this section, we introduce a efficient, verifiable and accurate S3PC (EVA-S3PC) framework (shown in \autoref{fig:PPS3PC-Framework}) with the security model and data disguising technique used.
\subsection{EVA-S3PC Framework}
\noindent A S3PC call from a client requests the computation of a function $f(A,B,C)$ on data $A$, $B$, $C$ supplied by isolated data owners without disclosing any information about individual data to anyone else except for the owner. In this paper, we assume that the operands are organized in the form of matrices. Computation proceeds by invoking the basic protocols in the Secure Protocol Library (SPL). Each basic protocol consists of Commodity-Server (CS) pre-processing, Online Computing and Result Verification.
CS, initially introduced by Beaver\cite{beaver1997commodity} as a semi-trusted third party, is widely utilized to address various S2PC problems. It simply generates offline random matrices $(R_a,r_a),(R_b,r_b)$ and $(R_c,r_c)$ to be used by the subsequent online computing stage to disguise original data, and is prohibited from conspiring with any other participants. Its simplicity feature makes it easy to be built in real world applications.
\begin{figure}[ht]
  \centering
  \includegraphics[width=1.0\hsize]{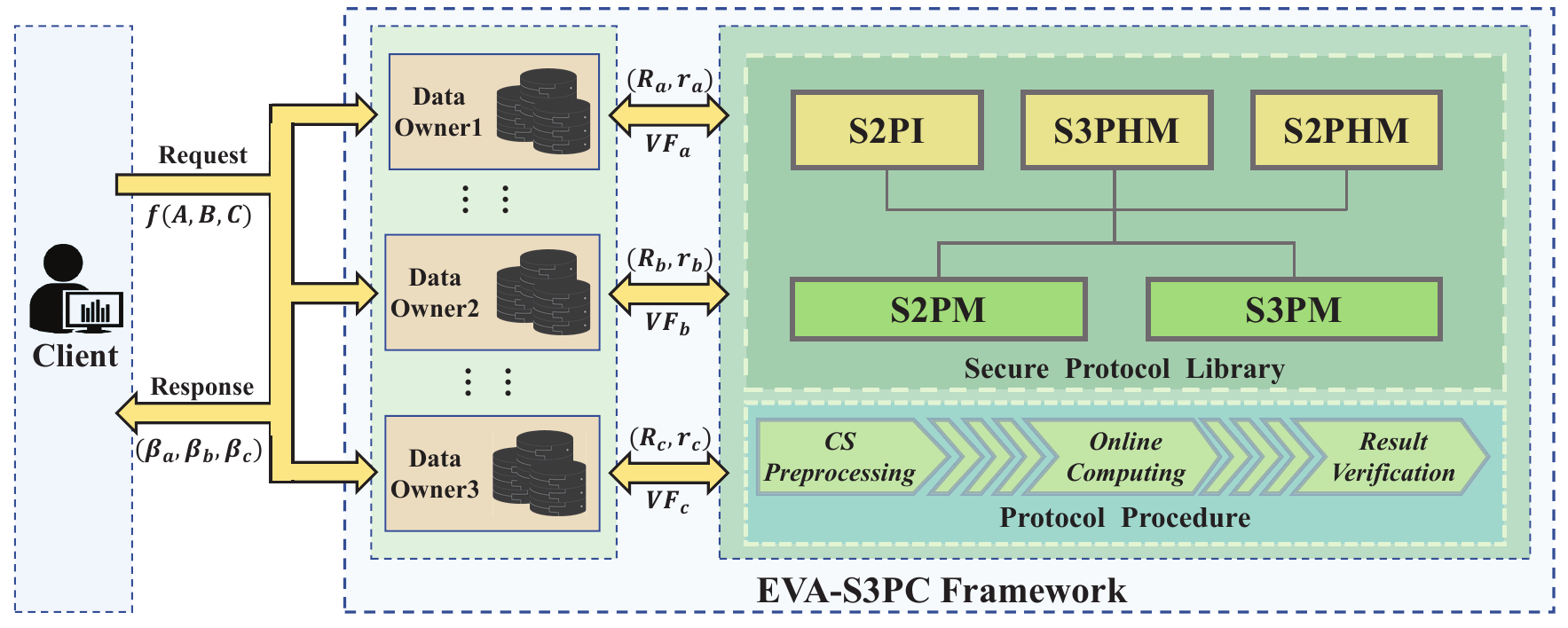}
  \vspace{-0.6cm}
  \caption{Framework of EVA-S3PC}
  \label{fig:PPS3PC-Framework}
  \Description[<short description>]{<long description>}
\end{figure}
\noindent The online computing stage following CS pre-processing computes $f(A,B,C)$ by sequentially executing the fundamental protocols (S2PM, S3PM, S2PI, S2PHM, S3PHM) among data owners who can only access partial outputs $(\beta_a,\beta_b,\beta_c)$ with which no information about the input and output data can be inferred. 
Because the request client normally does not participate in the calculation of $f(A,B,C)$, there ought to be a mechanism to check the reliability of the output to ensure that its computation follows the exact protocol. In the result verification stage, participants collaborate to produces check matrices ($VF_a$, $VF_b$, $VF_c$) respectively and can independently use them to check the reliability of the output result.
S3PC protocols design should satisfy the following properties for the requirement of security, correctness, verifiability, precision and efficiency:
\begin{itemize}
\item \textbf{Security}
Input data ($A,B,C$) cannot be inferred by any participant node, except for the data owner itself, by using their local input, output and intermediate results produced by the protocol.
\item \textbf{Correctness}.
The request client can aggregate the outputs of all participants to produce correct result if the protocol is strictly followed by them.
\item \textbf{Verifiability}.
Each participants can independently and robustly verify the correctness of the result using intermediate results produced by the protocol with negligible probability of error.  
\item \textbf{Precision}.
Computation of Float64 double-precision matrix of various dimension and distribution should have at least $8\sim 10$ decimal bit accurate significant.
\item \textbf{Efficiency}.
The protocol should have minimum computational and communication complexity.
\end{itemize}
\subsection{Security Model}
\noindent For the definition of security, we follow a semi-honest model of S2PC using the criteria of computational indistinguishability between the view of ideal-world and the simulated views of real-world on a finite filed, and extended it to the scenario of a real number filed.

\begin{definition}[Semi-honest adversaries model \cite{evans2018pragmatic}]\label{def1}
    In a semi-honest adversary model, it is hypothesized that all participants follow the exact protocol during computation but may use input and intermediate results of their own to infer others’ original data.
\end{definition}

\begin{definition}[Computational Indistinguishability \cite{goldreich2004foundations}]\label{def2}
     A probability ensemble $X = \{X(a, n)\}_{a\in\{0,1\}^{*};n\in \mathbb{N}}$ is an infinite sequence of random variables indexed by $a\in \{0, 1\}^{*}$ and $n\in \mathbb{N}$. In the context of secure computation, the $a$ represents the parties’ inputs and $n$ denotes the security parameter. Two probability ensembles $X = \{X(a, n)\}_{a\in\{0,1\}^{*};n\in \mathbb{N}}$ and $Y = \{Y(a, n)\}_{a\in \{0,1\}^{*};n\in \mathbb{N}}$ are said to be computationally indistinguishable, denoted by $X\overset{c}{\equiv}Y$, if for every non-uniform polynomial-time algorithm D there exists a negligible function $\mu(\cdot)$ such that for every $a\in \{0, 1\}^{*}$ and every $n\in \mathbb{N}$,
\begin{align}
    |Pr[D(X(a,n))=1]-Pr[D(Y(a,n))=1]|\leq \mu(n)  
\end{align}
\end{definition}

\begin{definition}[Privacy in Semi-honest 2-Party Computation \cite{lindell2017simulate}]\label{def3}
    Let $f:\{0,1\}^{*}\times \{0,1\}^{*}\mapsto \{0,1\}^{*}\times\{0,1\}^{*}$ be a functionality, where $f_1(x,y)$ (resp.,$f_2(x,y)$) denotes the first(resp.,second) element of $f(x,y)$ and $\pi$ be a two-party protocol for computing $f$. The view of the first(resp., second) party during an execution of $\pi$ on $(x,y)$, denoted $VIEW^{\pi}_{1}(x, y)$ (resp., $VIEW^{\pi}_{2}(x, y)$), is $(x,r^1,m^1_1,\cdots,m^1_t)$ (resp.,$(y,r^2,m^2_1,\cdots ,m^2_t )$), where $r^1$ (resp.,$r^2$) represents the outcome of the first (resp.,second) party’s internal coin tosses, and $m^1_i$(resp.,$m^2_i$) represents the $i^{th}$ message it has received. The output of the first (resp.,second) party during an execution of $\pi$ on $(x,y)$, denoted $OUTPUT^{\pi}_1(x,y)$ (resp.,$OUTPUT^{\pi}_2(x,y)$), is implicit in the party’s view of the execution.
We say that $\pi$ privately computes $f(x,y)$ if there exist polynomial time algorithms, denoted $S_1$ and $S_2$ such that:
\begin{align}
 &\{(S_1(x,f_1(x, y)),f_2(x,y))\}_{x,y\in \{0,1\}^{*}}\nonumber \\
 &\overset{c}{\equiv} \{(VIEW^{\pi}_1(x, y), OUTPUT^{\pi}_2(x, y))\}_{x,y\in \{0,1\}^{*}} \nonumber
 \end{align}  
 \begin{align}
 &\{(f_1(x, y), S_2(y, f_2(x, y)))\}_{x,y\in \{0,1\}^{*}}\nonumber \\
 &\overset{c}{\equiv} \{(OUTPUT^{\pi}_1(x, y), VIEW^{\pi}_2(x, y))\}_{x,y\in \{0,1\}^{*}}
 \end{align}
 where $\overset{c}{\equiv}$ denotes computational indistinguishability and $|x|=|y|$. We stress that above $VIEW^{\pi}_{1}(x,y)$ and $VIEW^{\pi}_{2}(x,y)$, $OUTPUT^{\pi}_1(x,y)$ and $OUTPUT^{\pi}_2(x, y)$ are related random variables, defined as a function of the same random execution.
\end{definition}

\begin{definition}[Privacy in Semi-honest 3-party computation]\label{def4}
Let $f=(f_1,f_2,f_3)$ be a functionality. We say that $\pi$ privately computes $f(x,y,z)$ if there exist polynomial time algorithms, denoted $S_1$, $S_2$ and $S_3$ such that:
\begin{align}
 & \{(S_1(x,f_1(x, y, z)),f_2(x,y,z),f_3(x,y,z))\}_{x,y,z\in \{0,1\}^{*}}\nonumber \\
 & \overset{c}{\equiv} \{(VIEW^{\pi}_1(x, y, z), OUTPUT^{\pi}_2(x, y, z),\nonumber \\
 & OUTPUT^{\pi}_3(x, y, z)\}_{x,y,z\in \{0,1\}^{*}} \nonumber
 \end{align}  
 \begin{align}
 & \{(f_1(x, y, z),S_2(y,f_2(x,y,z)),f_3(x,y,z))\}_{x,y,z\in \{0,1\}^{*}}\nonumber \\
 & \overset{c}{\equiv} \{(OUTPUT^{\pi}_1(x, y, z), VIEW^{\pi}_2(x, y, z),\nonumber \\
 & OUTPUT^{\pi}_3(x, y, z)\}_{x,y,z\in \{0,1\}^{*}} \nonumber
 \end{align}  
 \begin{align}
 & \{(f_1(x, y, z),f_2(x,y,z),S_3(z,f_3(x,y,z)))\}_{x,y,z\in \{0,1\}^*{}}\nonumber \\
 & \overset{c}{\equiv} \{(OUTPUT^{\pi}_1(x, y, z), OUTPUT^{\pi}_2(x, y, z)\nonumber \\
 & VIEW^{\pi}_3(x, y, z)\}_{x,y,z\in \{0,1\}^{*}} 
 \end{align}
where, again, $\overset{c}{\equiv}$ denotes computational indistinguishability and $|x|=|y|=|z|$. $VIEW^{\pi}_{1}(x,y,z)$. $VIEW^{\pi}_{2}(x,y,z)$ and $VIEW^{\pi}_{3}(x,y,z)$, $OUTPUT^{\pi}_1(x,y,z)$, $OUTPUT^{\pi}_2(x, y, z)$ and $OUTPUT^{\pi}_3(x, y, z)$ are related random variables, defined as a function of the same random execution.
\end{definition}

\noindent This definition is for the general case of the real-ideal security paradigm defined in a formal language and for deterministic functions, as long as they can ensure that the messages $\{S_i(x,f_i(x)), (i\in 1,2...)\}$ generated by the simulator in the ideal-world are distinguishable form $\{view_i^\pi(x), (i\in 1,2...)\}$ in the real-world, then it can be shown that a protocol privately computes $f$ in a finite field. Furthermore, a heuristic model defined on a real number field is introduced as follows \cite{du2004privacy}:
\begin{definition}[Security Model in field of real number]\label{def5}
    All inputs in this model are in the real number field $\mathbb{R}$. Let $I_A$ and $I_B$ represent Alice's and Bob's private inputs, and $O_A$ and $O_B$ represent Alice's and Bob's outputs, respectively. Let $\pi$ denote the two-party computation involving Alice and Bob, where $(O_A, O_B)=\pi(I_A, I_B)$. Protocol $\pi$ is considered secure against dishonest Bob if there is an infinite number of $(I^{*}_A, O^{*}_A)$ pairs in $(\mathbb{R},\mathbb{R})$ such that $(O^{*}_A,O_B)=\pi(I^{*}_A,I_B)$. A protocol $\pi$ is considered secure against dishonest Alice if there is an infinite number of $(I^{*}_B, O^{*}_B)$ pairs in $(\mathbb{R},\mathbb{R})$ such that $(O_A, O^{*}_B)=\pi(I_A, I^{*}_B)$.
\end{definition}

\noindent A protocol is considered secure in the field of real numbers if, for any input/output combination $(I,O)$ from one party, there are an infinite number of alternative inputs in $\mathbb{R}$ from the second party that will result in $O$ from the first party's perspective given its own input $I$. From the adversary's point of view, this infinite number of the other party’s input/output represents a kind of stochastic indistinguishability in real number field, which is similar to computational indistinguishability in the real-ideal paradigm. Moreover, a simulator in the ideal world is limited to merely accessing the corrupted parties' input and output. In other words, the protocol $\pi$ is said to securely compute $f$ in the field of real numbers if, and only if, computational indistinguishability is achieved with any inputs from non-adversaries over a real number field, and the final outputs generated by the simulator are constant and independent from all inputs except for the adversaries.
\section{PROPOSED WORK}\label{Proposed-Work}
\noindent EVA-S3PC consists of sub-protocols that can be hierarchically constructed from elementary sub-protocols  including S2PM and S3PM, and derived sub-protocols including S2PI, S2PHM and S3PHM. We will then introduce each sub-protocols, and privacy exposure state analyses and provide formal security proofs under the semi-honest model.
\subsection{S2PM}\label{S2PM}
\noindent The problem of S2PM is defined as: 
\begin{problem}[Secure 2-Party Matrix Multiplication]\label{Problem-S2PM}
    Alice has an $n\times s$ matrix $A$ and Bob has an $s\times m$ matrix $B$. They want to conduct the multiplication, such that Alice gets $V_a$ and Bob gets $V_b$, where $V_a+V_b=A\times B$.
\end{problem}
\subsubsection{Description of S2PM}
The proposed S2PM includes three phase: CS pre-processing phase  in Algorithm \ref{alg:S2PM-Preprocessing}, online computation phase in Algorithm \ref{alg:S2PM-Computing}, and result verification phase in Algorithm \ref{alg:S2PM-Verification}.

\noindent \textbf{Pre-processing Phase.}
In Algorithm \ref{alg:S2PM-Preprocessing}, CS generates a set of random private matrices $(R_a,r_a)$ for Alice and $(R_b,r_b)$ for Bob to disguise their input matrices $A$ and $B$. Moreover, for the purpose of result verification to be given in Algorithm \ref{alg:S2PM-Verification}, the standard matrix $S_t = R_a \cdot R_b$ is also sent to both Alice and Bob. This raises the problem of leaking $R_a$ (resp., $R_b$) to Bob (resp., Alice) when it is a non-singular matrix. Therefore, to protect the privacy of $R_b$, the matrix $R_a$ must satisfy the constraint $rank(R_a) < s$ according to lemma 1 ( resp., $rank(R_b) < s$ for the protection of $R_a$). 
\renewcommand{\thealgorithm}{1} 
    \begin{breakablealgorithm}
        \caption{S2PM CS Pre-processing Phase}
        \label{alg:S2PM-Preprocessing}
        \begin{algorithmic}[1] 
            \Require The dimension $(n,s)$ of Alice's private matrix $A$ and the dimension $(s,m)$ of Bob's private matrix $B$.
            \Ensure Alice gets private matrices $(R_a,r_a,S_t)$ and Bob gets private matrices $(R_b,r_b,S_t)$. 
            \State Generate two random matrices $R_a \in \mathbb{R}^{n \times s}$ and $R_b \in \mathbb{R}^{s \times m}$, such that $rank(R_a)=min(n,s)-1, rank(R_b)=min(s,m)-1$;
            \State Generate two random matrices $r_a, r_b$ with constraint $r_a+r_b = R_a\cdot R_b$ and let $S_t=R_a \cdot R_b$, where $r_a,r_b \in \mathbb{R}^{n \times m}$;
            \State Send $(R_a, r_a)$ and $(R_b, r_b)$ to Alice and Bob, respectively, denoted as $(R_a,r_a,S_t) \Rightarrow Alice$ and $(R_b,r_b,S_t) \Rightarrow Bob$.
            \State \Return
        \end{algorithmic}
    \end{breakablealgorithm}   
Also note that the CS, as a semi-honest third party, will not participate in any computation or communication after delieverying $R_a, R_b, r_a, r_b, S_t$ in the pre-processing stage. These random matrices are completely independent of $A$ and $B$ except for the size parameters of them. The protection of size information is feasible but is out of the scope of this paper.

\noindent\textbf{Online Phase.}
The online phase following CS pre-processing is made up of consecutive matrix calculation in the order as shown in Algorithm \ref{alg:S2PM-Computing}. 
Note that the final product of $A\times B$ is disguised by $V_a$ and $V_b$ to prevent Alice or Bob from knowing the it.
The correctness of the results can be easily verified with $V_a+V_b = [(\hat{A}\times B+(r_a-V_b))+r_a-(R_a\times \hat{B})]+V_b = [A\times B-V_b+(r_a+r_b-R_a\times R_b)]+V_b = A\times B$.
\renewcommand{\thealgorithm}{2} 
    \begin{breakablealgorithm}
        \caption{S2PM Online Computing Phase}
        \label{alg:S2PM-Computing}
        \begin{algorithmic}[1] 
            \Require Alice holds a private matrix  $A \in \mathbb{R}^ {n \times s}$,  Bob holds a private matrix $B \in \mathbb{R}^{s\times m}$.
            \Ensure Alice gets private matrices $(V_a,VF_a)$ and Bob gets private matrices $(V_b,VF_b)$, where $V_a, V_b, VF_a, VF_b \in\mathbb{R}^ {n \times m}$.
            \State Upon receiving the random matrices $(R_a, r_a, S_t)$, Alice computes $\hat{A} = A + R_a$ and sends $\hat{A} \Rightarrow \text{Bob}$
            \State Upon receiving the random matrices  $(R_b, r_b, S_t)$ , Bob computes $\hat{B} = B + R_b$ and sends  $\hat{B} \Rightarrow \text{Alice}$;
            \State  Bob generates a random matrix $V_b \in R^{n\times m}$, computes  $VF_b=V_b-\hat{A}\cdot B$, $T = r_b-VF_b$, then sends  $(VF_b,T)\Rightarrow \text{Alice}$;
            \State  Alice locally computes the matrix $V_a = T + r_a - (R_a \cdot \hat{B})$, $VF_a=V_a+R_a\cdot \hat{B}$, and sends $VF_a\Rightarrow \text{Bob}$;
            \State \Return $(V_a,VF_a),(V_b,VF_b)$ 
        \end{algorithmic}
    \end{breakablealgorithm}
\noindent\textbf{Verification Phase.}
A practical result verification solution is proposed based on Freivalds' work\cite{motwani1995randomized,freivalds1979fast}. The nature of the algorithm is a True-biased Monte-Carol verification process with one-side error $\sigma=1/2$, which corresponds to the steps $2\sim 6$ of Algorithm \ref{alg:S2PM-Verification}.
\renewcommand{\thealgorithm}{3} 
    \begin{breakablealgorithm}
        \caption{S2PM Result Verification Phase}
        \label{alg:S2PM-Verification}
        \begin{algorithmic}[1]
            \Require Alice holds verified matrices  $(VF_a,S_t) \in \mathbb{R}^ {n \times m}$ and Bob holds verified matrices $(VF_b,S_t) \in \mathbb{R}^ {n \times m}$.
            \Ensure Accept  or Reject $V_a$ and $V_b$.
            \For{$i=1:l$}
            \State Alice generates an vector $\hat{\delta_a} \in \mathbb{R}^{m\times 1}$ whose elements are all randomly composed of 0/1;
            \State  Alice then computes $E_r=(VF_a+VF_b-S_t)\times \hat{\delta_a}$;  
            \If{$E_r\neq (0,0,\cdots,0)^T$}
            \State \Return Rejected; 
            \EndIf
            \EndFor
            \State  Bob repeats the same verification procedure as Alice;
            \State \Return Accepted
        \end{algorithmic}
    \end{breakablealgorithm}

\begin{lemma}\label{lemma1}
    Given the linear system $A \cdot X=B$ and the augmented matrix $(A|B)$, n represents the number of rows in $X$. If $rank(A)=rank(A|B)<n$, then the system has a infinite solutions\cite{bellman1997introduction}.
\end{lemma}
\begin{theorem}\label{theorem1}
    The S2PM result verification phase satisfies robust abnormality detection.
\end{theorem}

\begin{proof}
    When $V_a$ or $V_b$ are computed correctly such that $V_a + V_b = A \cdot B$, $E_r$ in step 3 is always \textbf{0} regardless of $\hat{\delta_a}$. Because according to Algorithm \ref{alg:S2PM-Computing}:
\begin{flalign}
     VF_a + VF_b - S_t &= V_a + V_b + R_a \cdot \hat{B} - \hat{A} \cdot B - R_a \cdot R_b \nonumber \\
                       &= (V_a + V_b - A \cdot B) + R_a \cdot R_b - R_a \cdot R_b \nonumber \\
                       &= V_a + V_b - A \cdot B &
\end{flalign}

\noindent When the computation of $V_a$ or $V_b$ is abnormal such that $V_a + V_b \neq A \cdot B$. Let $ProA_l$ denote the probability that Alice fails to detect the error after $l$ rounds of Algorithm \ref{alg:S2PM-Verification}. We will show that $ProA_l$ is upper-bounded by the exponential of $l$ and is negligible with sufficiently large $l$. 
Let $H = VF_a + VF_b - S_t$, and $E_r = H \times \hat{\delta_a} = (p_1, \cdots, p_n)^T$. There should be at least one non-zero element in $H$, say $h_{ij} \neq 0$ and its corresponding $p_i$ in $E_r$, with the relationship $p_i = h_{ij} \delta_j + w$, where $w = \sum_{k=1}^m h_{ik} \delta_k - h_{ij} \delta_j$. The marginal probability $Pr(p_i = 0)$ can be calculated as:
\begin{align}
    & Pr(p_i = 0) = Pr(p_i = 0 | w = 0)Pr(w = 0) \; + \nonumber \\
    & Pr(p_i = 0 | w \neq 0)Pr(w \neq 0)
\end{align}
Substituting the posteriors
\begin{align}
Pr(p_i = 0 | w = 0) = Pr(\delta_j = 0) = 1/2 \nonumber \\
Pr(p_i = 0 | w \neq 0) \leq Pr(\delta_j = 1) = 1/2    
\end{align}
gives
\begin{equation}
    Pr(p_i = 0) \leq (1/2) Pr(w = 0) + (1/2) Pr(w \neq 0)    
\end{equation}
Because $Pr(w \neq 0) = 1 - Pr(w = 0)$, we have
\begin{equation}
    Pr(p_i = 0) \leq 1/2
\end{equation}
The probability of single check failure $ProA_1$ ($l=1$) satisfies
\begin{equation}
    ProA_1 = Pr(E_r = (0, \cdots, 0)^T) \leq Pr(p_i = 0) \leq 1/2
\end{equation}
And when such random check is repeated $l$ times, $ProA_l$ satisfies
\begin{equation}
    ProA_l = ProA_1^l \leq \frac{1}{2^l}
\end{equation}
Since the verification is independently carried out by each participant, and fails when all of them fail. The probability of failing a 2-party verification:
\begin{equation}
    {Pf}_{S2PM} = (ProA_l)^2 \leq \frac{1}{4^l}
\end{equation}
The proof is now completed.
\end{proof}

\noindent In practice, $l$ is set to $20$ as a balance between the verification cost and a negligible upper bound probability of $(\frac{1}{4})^{20} \approx 9.09\times 10^{-13}$ for failing to detect any abnormality in the result.

\subsubsection{Security of S2PM}
\noindent An intuitive analysis of the S2PM protocol's security concerning potential compromise of data privacy followed by a formal proof of security under a semi-honest model in real number field will be provided.\\
\textbf{Data Privacy Compromise Analysis of S2PM.}
\begin{table}[tp]
  \centering
  \caption{Data Privacy Compromise Analysis of S2PM within the Semi-Honest Model.}
  \resizebox{\linewidth}{!}{
    \begin{tabular}{cccc}
    \toprule
    \multirow{2}[4]{*}{\textbf{Participants}} & \multicolumn{3}{c}{\textbf{Honest-but-curious Model with Two Party}} \\
\cmidrule{2-4}    \multicolumn{1}{c}{} & \multicolumn{1}{c}{\textbf{Private Data }} & \multicolumn{1}{c}{\textbf{Privacy Inference Analysis}} & \multicolumn{1}{c}{\textbf{Security constraints.}} \\
    \midrule
    \textbf{Alice} & $R_a,r_a,\hat{A},\hat{B},T,VF_a, V_a, S_t$ & $\psi(m_A)\Rightarrow B|m_A=\{S_t,R_a,\hat{B}\}$ &  $rank(R_a)<s$\\
    \midrule
    \textbf{Bob} &$R_b,r_b,\hat{A},\hat{B},T,VF_b, V_b, S_t$ & $\psi(m_B)\Rightarrow A|m_B=\{ S_t,R_b,\hat{A}\}$ &  $rank(R_b)<s$\\
    \bottomrule
    \end{tabular}}
  \label{tab:S2PM_Security}
\end{table}%
The private data including input, output and intermediate result of each participant are listed in Table 4. For each party in the computation, $m_A$ or $m_B$ is the maximum set of key messages it can make use of to infer and disclose the data of others. 
For instance, the key-message set $m_A = \{S_t, R_a, \hat{B}\}$ held by Alice obviously compromises the privacy of Bob's original data $B$, because Alice might deduce $R_b$ from the relationship $S_t = R_a \cdot R_b$ and subsequently infer $B$ using $\hat{B} = B + R_b$. To prevent this from happening, security constraint $rank(R_a) < s$ must be applied so that $S_t = R_a \cdot X$ has infinite number of solutions, effectively preventing Bob’s private data $B$. Similar security analysis equally applies to Bob.\\
\textbf{Security Proof of S2PM}
Following the definition of 2-party secure computation in Section 2.2, let $f=(f_1,f_2)$ be a probabilistic polynomial-time function and let $\pi$ be a 2-party protocol for computing $f$.To prove the security of computation on a real number field, it is necessary to construct two simulators, $S_1$ and $S_2$ in the ideal-world, such that the following relations hold simultaneously:
\begin{align}
    S_1(x, f_1(x,y)) \overset{c}{\equiv} view_1^{\pi}(x,y) \nonumber \\
    S_2(y, f_2(x,y)) \overset{c}{\equiv} view_2^{\pi}(x,y)\label{12}
\end{align}
together with two additional constraints to ensure the constancy of one party's outputs with arbitrary inputs from another party:
\begin{align}
    f_1(x,y) \equiv f_1(x,y^*) \nonumber \\
    f_2(x,y) \equiv f_2(x^*,y)\label{13}
\end{align}
so that the input of the attacked party cannot be deduced using the output from the other one. The security proof under the semi-honest (also known as 'honest-but-curious') adversary model is given below.
\begin{theorem}\label{theorem2}
    The S2PM protocol, denoted by $f(A,B)=A\cdot B$, is secure in the honest-but-curious model.
\end{theorem}
\begin{proof}
    We will prove this theorem by illustrating two simulators $S_1$ and $S_2$. Without the loss of generality, we assume that all data are defined on a real number field.\\     
    \noindent\textbf{(Privacy against Alice):}
    Firstly, we present $S_1$ for simulating ${view}_1^\pi (A,B)$ such that $S_1(A,V_a)$ is indistinguishable from ${view}_1^\pi (A,B)$. $S_1$ receives $(A,V_a)$ as input and output of Alice and proceeds as follows:
    \begin{enumerate}
        \item $S_1$ computes $\hat{A} = A + R_a$ and then chooses two pairs of random matrices $({B}', {V_b}'),({R_b}', {r_b}')$ by solving $A\cdot B' = V_a + {V_b}'$ and $R_a\cdot {R_b}' = r_a + {r_b}'$;
        \item $S_1$ performs the following calculations $\hat{B}' = {B}' + {R_b}'$, $T' = \hat{A}\cdot{B}' + {r_b}' - {V_b}'$, ${VF_b}' = {V_b}' - \hat{A}\cdot {B}'$, ${VF_a}' = {V_a}' + R_a \cdot \hat{B}'$;
        \item Finally, $S_1$ outputs the result ${V_a}' = T' + r_a - R_a \cdot \hat{B}'$.
    \end{enumerate}
    Here, $S_1$ simulates the message list of Alice as $S_1 (A,V_a )=\{A,R_a,r_a,\hat{B}',T',{VF_b}',{V_a}'\}$, and the view of Alice in the real-world is ${view}_1^\pi (A,B)=\{A, R_a, r_a, \hat{B}, T, VF_b, V_a\}$. 
    Due to the choices of $B', {R_b}', {r_b}', {V_b}'$ are random and simulatable (denoted as $\{B', {R_b}', {r_b}', {V_b}'\} \overset{c}{\equiv} \{B, R_b, r_b, V_b\}$), it follows that $\{\hat{B}', T', {VF_b}' \} \overset{c}{\equiv} \{\hat{B}, T, VF_b \}$. 
    Besides, $V_a'$ computed by these indistinguishable intermediate variables can be expanded as $V_a' = T' + r_a - R_a \cdot \hat{B}' = \hat{A}\cdot{B}' + {r_b}' - {V_b}'+ r_a - R_a \cdot\hat{B}' = A\cdot B' + R_a \cdot B' + {r_b}' - {V_b}' + r_a - R_a\cdot B' - R_a \cdot{R_b}'=(A\cdot B' - {V_b}') + (r_a + {r_b}' - R_a \cdot{R_b}')$. 
    Replacing $A\cdot B'$ and $R_a \cdot {R_b}'$ in step (1) gives ${V_a}' = V_a$, and therefore ${V_a}' \overset{c}{\equiv} V_a$ and $S_1 (A,V_a) \overset{c}{\equiv} {view}_1^{\pi} (A,B)\nonumber$.    
    Moreover, for the proof of output constancy, because Bob's input $B'$ is a random matrix, and Alice's output $V_a'=f_1(A,B')$ is identically equal to $V_a=f_1(A,B)$, we have $f_1(A,B')\equiv f_1(A,B)$ for any arbitrary input from Bob.

    \noindent\textbf{(Privacy against Bob):}
    $S_2$ simulates ${view}_2^\pi (A,B)$ such that $S_2 (B,V_b)$ is indistinguishable from ${view}_2^\pi (A,B)$. It receives $(B,V_b)$ as input and output of Bob and proceeds as follows:
    \begin{enumerate}
        \item $S_2$ computes $\hat{A}' = A' + {R_a}'$ and then choose two pairs of random matrices $(A', {R_a}')$ by solving $A'\cdot B = {V_a}' + V_b$ and ${R_a}' \cdot R_b = {r_a}' + r_b$;
        \item $S_2$ performs the following calculations $\hat{B} = B + R_b$, $T' = \hat{A}'\cdot B + r_b - V_b$, ${VF_b}' = V_b - \hat{A}' \cdot B$, ${VF_a}' = {V_a}' + {R_a}'\cdot \hat{B}$;
        \item Finally, $S_2$ outputs the result $V_b$;
    \end{enumerate}
    
    \noindent Here, $S_2$ simulates the message list of Bob as $S_2(B,V_b)=\{B,R_b,r_b,\hat{A}', {VF_a}',V_b\}$. Note that in the real world the view of Bob is ${view}_2^\pi(A,B)=\{B, R_b, r_b,  \hat{A}, VF_a, V_b\}$. 
    Due to the choices of $A', {R_a}', {r_a}', {V_a}'$ are random and simulatable ($\{A', {R_a}', {r_a}', {V_a}'\} \overset{c}{\equiv} \{A, R_a, r_a, V_a\}$), it follow that $\{\hat{A}', {VF_a}' \} \overset{c}{\equiv} \{\hat{A}, VF_a \}$, and therefore $S_2 (B,V_b) \overset{c}{\equiv} {view}_2^{\pi} (A,B)$.
    Meanwhile, Alice's input $A'$ is a random matrix in the field of real numbers, and Bob's output $V_b'=f_2(A',B)$ is identically equal to $V_b=f_2(A,B)$, so $f_2(A',B)\equiv f_2(A,B)$ for any arbitrary input from Alice.
    Consequently, both simulators $S_1$ and $S_2$ fulfill the criteria of computational indistinguishability and output constancy, as outlined in equations (\ref{12}) and (\ref{13}). Protocol $f(A,B)=A\cdot B$ is secure in the honest-but-curious model.
\end{proof}

\subsection{S3PM}\label{S3PM}
\noindent The problem of S3PM is defined as:
\begin{problem}[Secure 3-Party Matrix Multiplication Problem]\label{problem-S3PM}
    Alice has an $n\times s$ matrix $A$, Bob has an $s\times t$ matrix $B,$ and Carol has an $t\times m$ matrix $C$. They want to conduct the multiplication  $f(A,B,C)=A\cdot B\cdot C$, in the way that Alice outputs $V_a$, Bob outputs $V_b$, Carol output $V_c$ and $V_a+V_b+V_c=A\cdot B\cdot C$.
\end{problem}
\subsubsection{Description of S3PM}
\noindent The proposed S3PM includes three similar phases as in S2PM, and is demonstrated in Algorithms \ref{alg:S3PM-Preprocessing}$\sim$\ref{alg:S3PM-Verification}.

\noindent \textbf{Preprocessing Phase.}
CS generates a set of random private matrices $(R_a,r_a)$ for Alice, $(R_b,r_b)$ for Bob and $(R_c,r_c)$ to disguise their input matrices $A, B, C$. Moreover, for the purpose of result verification in Algorithm \ref{alg:S3PM-Verification}, the standard matrix $S_t = R_a \cdot R_b \cdot R_c$ is also sent to both Alice, Bob and Carol. The difference from Algorithm \ref{alg:S2PM-Preprocessing} is that the private matrices $R_a,R_b,R_c$ can be arbitrarily random without any rank constraint, because equation $R_a \cdot R_b \cdot R_c = S_t$ is unsolvable when Alice, Bob, or Carol has only one private random matrix.
\renewcommand{\thealgorithm}{4} 
    \begin{breakablealgorithm}
        \caption{S3PM CS Preprocessing Phase}
        \label{alg:S3PM-Preprocessing}
        \begin{algorithmic}[1] 
            \Require The dimension $(n,s)$ of Alice's private matrix $A$, the dimension $(s,t)$ of Bob's private matrix $B$, and the dimension $(t,m)$ of Carol's private matrix $C$.
            \Ensure Alice gets private matrices $(R_a,r_a,S_t)$, Bob gets private matrices $(R_b,r_b,S_t)$,and Carol gets private matrices $(R_c,r_c,S_t)$. 
            \State Generate three groups of random matrices: $(R_a, r_a)$, $(R_b, r_b)$, $(R_c, r_c)$ with constraint $r_a + r_b + r_c = R_a \cdot R_b \cdot R_c$, and let $S_t=R_a \cdot R_b \cdot R_c$,  where $R_a \in \mathbb{R}^{n \times s}$, $R_b \in \mathbb{R}^{s \times t}$, $R_c \in \mathbb{R}^{t \times m}$, and $r_a, r_b, r_c, S_t \in \mathbb{R}^{n \times m}$;
            \State Send $(R_a,r_a,S_t) \Rightarrow Alice$,  $(R_b,r_b,S_t) \Rightarrow Bob$, $(R_c,r_c,S_t) \Rightarrow Carol$.
            \State \Return $(R_a, r_a, S_t), (R_b, r_b, S_t), (R_c, r_c, S_t)$ 
        \end{algorithmic}
    \end{breakablealgorithm}
\noindent\textbf{Online Phase.}
The online phase following CS pre-processing is made up of consecutive matrix calculation in the order as shown in Algorithm \ref{alg:S3PM-Computing}.
The output is also disguised into $V_a$, $V_b$ and $V_c$ to prevent Alice, Bob or Carol from knowing it.
\renewcommand{\thealgorithm}{5} 
    \begin{breakablealgorithm}
        \caption{S3PM Online Computing Phase}
        \label{alg:S3PM-Computing}
        \begin{algorithmic}[1]
            \Require Alice holds a private matrix  $A \in \mathbb{R}^ {n \times s}$, Bob holds a private matrix $B \in \mathbb{R}^{s\times t}$, and Carol holds a private matrix $C \in \mathbb{R}^{t\times m}$; 
            \Ensure Alice outputs private matrices $(V_a,VF_a)$, Bob outputs private matrices $(V_b,VF_b)$ and Carol gets private matrices $(V_c,VF_c)$, where $V_a,V_b,V_c,VF_a,VF_b,VF_c \in\mathbb{R}^ {n \times m}$;    
            \State Upon receiving random matrices $(R_a, r_a, S_t)$, Alice computes $\hat{A} = A + R_a$ and sends $\hat{A} \Rightarrow Bob$;
            \State Upon receiving random matrices $(R_c, r_c, S_t)$, Carol computes $\hat{C} = C + R_c$ and sends  $\hat{C} \Rightarrow Bob$;
            \State Upon receiving random matrices $(R_b, r_b, S_t)$, Bob computes $\hat{B} = B + R_b$, $\psi_1=\hat{A} \cdot \hat{B}$, $\gamma_1=\hat{A} \cdot R_b$,  $\psi_2=\hat{B} \cdot \hat{C}$, $\gamma_2=R_b \cdot \hat{C}$, $M_b=\hat{A} \cdot R_b \cdot \hat{C}$ and then sends $(\psi_1,\gamma_1) \Rightarrow Carol$,$(\psi_2,\gamma_2) \Rightarrow Alice$;
            \State Alice computes $S_a = R_a \cdot \gamma_2$, $M_a = A \cdot \psi_2$;
            \State Carol computes $S_c = \gamma_1 \cdot R_c$, $M_c = \psi_1 \cdot R_c$;
            \State Bob performs a full-rank decomposition (FRD) on matrix $\hat{B}$ with Gauss elimination, splitting it into a column full-rank (CFR) matrix $B_1\in \mathbb{R}^{s \times r}$ and a row full-rank (RFR) matrix $B_2 \in \mathbb{R}^{r \times t}$ with $\hat{B}=B_1\cdot B_2$, and sends  ${B_1} \Rightarrow Alice, {B_2} \Rightarrow Carol$;
            \State Alice generates a random matrix $V_a \in \mathbb{R}^{n \times m}$, and  computes $VF_a=M_a+S_a-V_a$, $T_a = VF_a-r_a$ and $t_1 = R_a \cdot B_1$, and sends $(T_a,VF_a) \Rightarrow Bob$, $(t_1,VF_a) \Rightarrow Carol$;
            \State Bob generates a random matrix $V_b \in \mathbb{R}^{n \times m}$ and computes $VF_b=-M_b-V_b$, $T_b =T_a+VF_b-r_b$, and sends  $VF_b \Rightarrow Alice$ and $(T_b,VF_b) \Rightarrow Carol$;
            \State Carol computes $t_2=B_2 \cdot R_c$, $S_b = t_1 \cdot t_2$,  $V_c =T_b - M_c + S_b + S_c - r_c$and $VF_c=-M_c+S_c+S_b-V_c$, and sends $VF_c \Rightarrow Alice$, $VF_c \Rightarrow Bob$
            \State \Return $(V_a,VF_a),(V_b,VF_b),(V_c,VF_c)$ 
        \end{algorithmic}
    \end{breakablealgorithm}
We can easily verified the correctness of the algorithm by following $V_a + V_b + V_c = V_a + V_b + (T_b - M_c + S_b + S_c - r_c) = V_a + V_b + [((M_a + S_a - V_a - r_a)- M_b - V_b - r_b) - M_c + S_b + S_c - r_c]  = (M_a - M_b - M_c) + (S_a + S_b + S_c) - (r_a + r_b + r_c)$. Replacing the variables $\{M_a, M_b, M_c, S_a, S_b, S_c\}$ in step 3$\sim$5 gives $V_a + V_b + V_c = (A\cdot B\cdot C + A\cdot R_b\cdot R_c + A\cdot R_b\cdot C + A\cdot B\cdot R_c) - (A\cdot R_b\cdot C + A\cdot R_b\cdot R_c + R_a\cdot R_b\cdot C + R_a\cdot R_b\cdot R_c)-(A\cdot B\cdot R_c + R_a\cdot R_b\cdot R_c + R_a\cdot B\cdot R_c + A\cdot R_b\cdot R_c) + (R_a\cdot R_b\cdot C + R_a\cdot R_b\cdot R_c) + (R_a\cdot R_b\cdot R_c + R_a\cdot B\cdot R_c) + (R_a\cdot R_b\cdot R_c + A\cdot R_b\cdot R_c)- R_a\cdot R_b\cdot R_c = A\times B\times C$.

\noindent \textbf{Verification Phase.}
Similar to the verification for \textit{S2PM}, we present a brief overview and proof for the result verification phase of \textit{S3PM}.
\renewcommand{\thealgorithm}{6} 
    \begin{breakablealgorithm}
        \caption{S3PM Result Verification Phase}
        \label{alg:S3PM-Verification}
        \begin{algorithmic}[1] 
            \Require Alice, Bob and Carol hold verified matrices $(VF_a,S_t) \in \mathbb{R}^ {n \times m}$, $(VF_b,S_t) \in \mathbb{R}^ {n \times m}$ , $(VF_c,S_t) \in \mathbb{R}^ {n \times m}$ respectively.
            \Ensure Accept or reject  $V_a$, $V_b$ and $V_c$ as correct result.
            \For{$i=1:l$}
            \State Alice generates an vector $\hat{\delta_a} \in \mathbb{R}^{m\times 1}$ whose elements are all randomly composed of 0/1;
            \State  Alice then computes $E_r=(VF_a+VF_b+VF_c-S_t)\times \hat{\delta_a}$;  
            \If{$E_r\neq (0,0,\cdots,0)^T$}
            \State \Return Rejected;
            \EndIf
            \EndFor
            \State  Bob and Carol repeat the same verification procedure as Alice;
            \State \Return Accepted;
        \end{algorithmic}
    \end{breakablealgorithm}
\begin{theorem}\label{theorem3}
    The S3PM result verification phase satisfies robust abnormal detection.
\end{theorem}

\begin{proof}
    This proof is similar to that of theorem \ref{theorem1}. When $V_a$, $V_b$ or $V_c$ are computed correctly such that $V_a + V_b +V_c= A \cdot B \cdot C$, $E_r$ in step 3 is always \textbf{0} regardless of $\hat{\delta_a}$. Because $VF_a + VF_b + VF_C - S_t = A\cdot B \cdot C -(V_a + V_b +V_c)$ according to Algorithm \ref{alg:S3PM-Computing}.

\noindent When the computation of $V_a$, $V_b$, or $V_c$ is abnormal, following the same process outlined in equations $(6)\sim (11)$, we can calculate the joint probability of failing a 3-party verification ${Pf}_{S3PM} = (ProA_l)^3 \leq (\frac{1}{8})^l$. The proof is now completed.
\end{proof}

\noindent Practically, we set $l=20$ so the failure probability does not exceed $(\frac{1}{8})^{20} \approx 8.67\times 10^{-19}$.
\subsubsection{Security of S3PM}
\noindent Similar to S2PM, a analysis of S3PM protocol's security concerning potential compromise of data privacy under a semi-honest model in the real number field will be provided.

\noindent \textbf{Data Privacy Compromise Analysis of S3PM.}
The private data include each party's inputs, outputs, intermediate results and are listed in Table \ref{tab:S3PM_Security}. For instance, we can intuitively observe that any combination of Alice's private data cannot yield valid information about $R_b, R_c, \hat{B}, \hat{C}$ (denoted as $m_A = \{\varnothing\}$), so the privacy of other participants is not compromised without further constraint. Similar security analysis equally applies to Bob and Carol.
\begin{table}[htbp]
  \centering
  \caption{Private Data Compromise Analysis of S3PM within the Semi-Honest Model.}
  \resizebox{\linewidth}{!}{
    \begin{tabular}{cccc}
    \toprule
    \multirow{2}[4]{*}{\textbf{Participants}} & \multicolumn{3}{c}{\textbf{Honest-but-curious Model with Three Party}} \\
\cmidrule{2-4}    \multicolumn{1}{c}{} & \multicolumn{1}{c}{\textbf{Private Data }} & \multicolumn{1}{c}{\textbf{Privacy Inference Analysis}} & \multicolumn{1}{c}{\textbf{Security constraints.}} \\
    \midrule
    \textbf{Alice} & $R_a,r_a,\hat{A},\hat{B}\hat{C},R_b\hat{C},B_1,T_a,VF_a,V_a,S_t$ & $\psi(m_A)\nRightarrow B, \psi(m_A)\nRightarrow C|m_A=\{\varnothing\}$ &  \textbf{/}\\
    \midrule
    \textbf{Bob} &$R_b,r_b,\hat{A},\hat{B},\hat{C},B_1,B_2,T_b,VF_b,V_b,S_t$ & $\psi(m_B)\nRightarrow A, \psi(m_B)\nRightarrow C|m_B=\{\varnothing\}$  &  \textbf{/}\\
    \midrule
    \textbf{Carol} & $R_c,r_c,\hat{C},\hat{A}\hat{B},\hat{A}R_b,B_2,R_aB_1,VF_c,V_c,S_t$ &$\psi(m_C)\nRightarrow A, \psi(m_C)\nRightarrow B|m_C=\{\varnothing\}$  &  \textbf{/}\\
    \bottomrule
    \end{tabular}}%
  \label{tab:S3PM_Security}%
\end{table}%

\noindent \textbf{Security Proof of S3PM.}
Following the definition of three-party secure computing in Section 2.2, let $f=(f_1,f_2,f_3)$ be a probabilistic polynomial-time function and let $\pi$ be a 3-party protocol for computing $f$.
To prove the security of computation based on a real number field, it is necessary to construct three simulators, $S_1$, $S_2$ and $S_3$ in the ideal-world, such that the following relations hold simultaneously:
\begin{align}\label{14}
    S_1(x, f_1(x,y,z)) \overset{c}{\equiv} view_1^{\pi}(x,y,z)\nonumber\\
    S_2(y, f_2(x,y,z)) \overset{c}{\equiv} view_2^{\pi}(x,y,z)\nonumber\\
    S_3(z, f_3(x,y,z)) \overset{c}{\equiv} view_3^{\pi}(x,y,z)
\end{align}
together with three additional constraints to ensure the constancy of one party's outputs with arbitrary inputs from another:
\begin{align}\label{15}
    f_1(x,y,z) \equiv f_1(x,y^*,z^*)\nonumber\\
    f_2(x,y,z) \equiv f_2(x^*,y,z^*)\nonumber\\
    f_3(x,y,z) \equiv f_3(x^*,y^*,z)
\end{align}
so that the input of the attacked party cannot be deduced using the output from another one. The security proof under the honest-but-curious adversary model is given below.
\begin{theorem}
    The protocol S3PM, denoted by $f(A,B,C)=A\cdot B\cdot C$, is secure in the honest-but-curious model.
\end{theorem}
\begin{proof}
    The construction of simulators of the proof is similar to that of theorem \ref{theorem2}. Without loss of generality, we assume all data are defined on real number field and illustrate three simulators $S_1$, $S_2$ and $S_3$.    
    \noindent\textbf{(Privacy against Alice):}
    $S_1$ simulates ${view}_1^\pi (A,B,C)$ and we prove that $S_1(A,V_a)$ is indistinguishable from ${view}_1^\pi (A,B,C)$. $S_1$ receives $(A,V_a)$ (input and output of Alice) as input and proceeds as follows:
    \begin{enumerate}
        \item $S_1$ computes $\hat{A} = A + R_a$, and then chooses two pairs of random matrices $(B', {V_b}',C',{V_c}')$ and $({R_b}',{r_b}',{R_c}',{r_c}')$ by solving $A\cdot B'\cdot C' = V_a + {V_b}'+{V_c}'$ and $R_a\cdot {R_b}'\cdot {R_c}' = r_a + {r_b}' + {r_c}'$;
        \item $S_1$ performs the following calculations $\hat{C}' = {C}' + {R_c}'$, $\hat{B}' = {B}' + {R_b}'$, ${\psi_1}'=\hat{A}\cdot {\hat{B}}'$, ${\gamma_1}'=\hat{A}\cdot {R_b}'$, ${\psi_2}'={\hat{B}}'\cdot {\hat{C}}'$, ${\gamma_2}'={R_b}'\cdot {\hat{C}}'$, ${M_b}'=\hat{A}\cdot {R_b}'\cdot {\hat{C}}'$,        ${S_a}' = R_a \cdot {R_b}'\cdot {\hat{C}}'$, ${M_a}' = A\cdot {\hat{B}}'\cdot {\hat{C}}'$, ${S_c}' = \hat{A}\cdot {R_b}'\cdot {R_c}'$, ${M_c}' = \hat{A}\cdot {\hat{B}}'\cdot {R_c}'$;
        \item $S_1$ decomposes matrix ${\hat{B}}'$ into a CFR matrix ${B_1}'$ and RFR matrix ${B_2}'$ such that ${\hat{B}}'={B_1}'\cdot {B_2}'$;
        \item $S_1$ computes ${VF_a}'= {M_a}'+{S_a}'-{V_a}$, ${T_a}' = {VF_a}'-{r_a}'$, ${t_1}' = {R_a}'\cdot {B_1}'$, ${VF_b}'=-{M_b}'-{V_b}'$, ${T_b}' ={T_a}'-{VF_b}'-{r_b}'$, ${t_2}'={B_2}'\cdot {R_c}', {S_b}' =  R_a\cdot {\hat{B}}'\cdot {R_c}'$, ${V_c}' ={T_b}' - {M_c}' + {S_b}' + {S_c}' - {r_c}'$ and ${VF_c}'=-{M_c}'+{S_c}'+{S_b}'-{V_c}'$;
        \item Finally, $S_1$ outputs $V_a$;
    \end{enumerate}
    \noindent$S_1$ simulates the message list of Alice as $S_1 (A,V_a)=\{A,R_a,r_a,{\psi_2}',{r_2}',{B_1}',{VF_b}',{VF_c}',V_a\}$, and the view of Alice in the real-world is ${view}_1^\pi (A,B,C) = \{A,R_a,r_a,\psi_2,r_2,B_1,VF_b,VF_c,V_a\}$. Because the choices of $\{B', C', {R_b}', {R_c}', {r_b}',{r_c}', {V_b}', {V_c}'\}$ are random and simulatable ($\{B', {R_b}', C', {R_c}', {V_b}', {V_c}'\} \overset{c}{\equiv} \{B, R_b, C, R_c, V_b,V_c\}$), $\{{\hat{B}}', {\hat{C}}'\} \overset{c}{\equiv} \{\hat{B},\hat{C}\}$ can be derived. Substituting these random variable in step (2) leads to $\{{\psi_2}', {r_2}', {M_a}', {M_b}', {M_c}', {S_a}', {S_b}', {S_c}'\} \overset{c}{\equiv} \{{\psi_2}, {r_2}, {M_a}, {M_b}, {M_c}, {S_a}, {S_b}, {S_c}\}$. Besides, there are infinite factorization options ${B_1}'$ and ${B_2}'$ satisfying full-rank decomposition of ${\hat{B}}'$ (where $\hat{B}'=B_1'\cdot B_2'=(B_1'\cdot P)\cdot (P^{-1}\cdot B_2')$ and P can be any non-singular matrix) so $\{{B_1}', {B_2}'\} \overset{c}{\equiv} \{B_1, B_2\}$.    
    Note that ${V_c}'$ can be expanded as ${V_c}' ={T_b}' - {M_c}' + {S_b}' + {S_c}' - {r_c}'= ({T_a}'-{VF_b}'-{r_b}') - {M_c}' + {S_b}' + {S_c}' - {r_c}'= {VF_a}'- {VF_b}'- {M_c}' + {S_b}' + {S_c}'-{r_a}'- {r_b}' - {r_c}'$. Also because ${VF_a}'$ and ${VF_b}'$ are deduced from the indistinguishable intermediate variables $\{{M_a}',{M_b}',{S_a}',{V_b}'\}$ in step (2), $\{{VF_a}',{VF_b}'\} \overset{c}{\equiv} \{{VF_a},{VF_b}\}$. Finally, ${VF_c}' \overset{c}{\equiv} {VF_c}$ can be derived from $\{{VF_a}', {VF_b}', {M_c}', {S_b}', {S_c}', {r_a}', {r_b}', {r_c}'\} \overset{c}{\equiv} \{{VF_a}, {VF_b}, {M_c}, {S_b}, {S_c}, {r_a}, {r_b}, {r_c}\}$, and therefore  $S_1 (A,V_a) \overset{c}{\equiv} {view}_1^\pi (A,B,C)$.
    Moreover, to prove output constancy, since the inputs $B'$ and $C'$ from Bob and Carol are both random matrices, and Alice's output $V_a'=f_1(A,B',C')$ is identically equal to $V_a=f_1(A,B,C)$, it follows that $f_1(A,B',C') \equiv f_1(A,B,C)$ for any arbitrary input from others' inputs.

    \noindent\textbf{(Privacy against Bob or Carol):}
    Similarly, we can construct simulators $S_2$ and $S_3$ and prove
    \begin{align}\label{22}
    &S_2(y, f_2(x,y,z)) \overset{c}{\equiv} view_2^{\pi}(x,y,z)\nonumber\\
    &S_3(z, f_3(x,y,z)) \overset{c}{\equiv} view_3^{\pi}(x,y,z)\nonumber\\
    &f_2(x,y,z) \equiv f_2(x^*,y,z^*)\nonumber\\
    &f_3(x,y,z) \equiv f_3(x^*,y^*,z)
    \end{align} 
    
    \noindent Consequently, simulators $S_1$, $S_2$ and $S_3$ fulfill the criteria of computational indistinguishability and output constancy, as outlined in equations (\ref{14}) and (\ref{15}). Protocol $f(A,B,C)=A\cdot B\cdot C$ is secure in the honest-but-curious model.
\end{proof}
    
\subsection{S2PI}\label{S2PI}
\noindent The problem of S2PI is defined as:
\begin{problem}[Secure 2-Party Matrix Inverse Problem]
    Alice has an $n\times n$ matrix $A$ and Bob has an $n\times n$ matrix $B$. They want to conduct the inverse computation $f(A,B)=(A + B)^{-1}$ in the way that Alice outputs $V_a$ and Bob outputs $V_b$, where $V_a + V_b = (A + B)^{-1}$.
\end{problem}
\subsubsection{Description of S2PI}
\noindent The proposed S2PI is made of three consecutive call of S2PM as illustrated in Algorithm \ref{alg:S2PI}. The first two calls compute $P\cdot A\cdot Q$ and $P\cdot B\cdot Q$ in steps (1)$\sim$(4) with the private random matrices $P$ and $Q$ of Alice and Bob respectively. The third one in steps (5)$\sim$(7) removes $Q^{-1}$ and $P^{-1}$ in the inverse $T^{-1}=P^{-1}\cdot (A+B)^{-1}\cdot Q^{-1}$ and output $U_a$ and $U_b$ such that $U_a + U_b = (A + B)^{-1}$. 
S2PM is used as a sub-protocol in S2PI and strictly follows the preprocessing, online computing, verification phases. For the analysis of the check failure probability of S2PI ${Pf}_{S2PI}$, note that S2PI detects the failure when any one of the three S2PMs detects failure, so ${Pf}_{S2PI} = ({Pf}_{S2PM})^{3} \leq (\frac{1}{4^l})^3 \approx 7.52\times 10^{-37} (l = 20)$.
\renewcommand{\thealgorithm}{7} 
    \begin{breakablealgorithm}
        \caption{S2PI}
        \label{alg:S2PI}
        \begin{algorithmic}[1]
            \Require Alice holds a private matrix  $A \in \mathbb{R}^ {n \times n}$, Bob holds a private matrix $B \in \mathbb{R}^{n\times n}$
            \Ensure Alice and Bob output $U_a \in\mathbb{R}^ {n \times n} $ and $U_b \in\mathbb{R}^ {n \times n}$ respectively.
            \State Alice generates a random non-singular matrix $P \in \mathbb{R}^{n \times n}$ and computes the matrix $I_A = P\cdot A$, where $I_A \in \mathbb{R}^{n\times n}$;
            \State Bob generates a random non-singular matrix $Q \in \mathbb{R}^{n \times n}$ and computes the matrix $I_B = B\cdot Q$, where $I_B \in \mathbb{R}^{n\times n}$;
            \State Alice and Bob jointly compute $I_A\times Q=V_{a1} + V_{b1} = (P\cdot A)\cdot Q$ with \textbf{S2PM} protocol, where $V_{a1}, V_{b1} \in \mathbb{R}^{n \times n}$;
            \State Alice and Bob jointly compute $P\times I_B=V_{a2} + V_{b2} = P\cdot (B\cdot Q)$ with \textbf{S2PM} protocol, where $V_{a2}, V_{b2} \in \mathbb{R}^{n \times n}$;
            \State Alice computes $V_a = V_{a1} + V_{a2}$ and sends $V_a \Rightarrow \text{Bob}$;
            \State Bob computes $V_b = V_{b1} + V_{b2}$, $T = (V_a + V_b)$, and $I_B^* = Q\cdot T^{-1}$, where $T,I_B^* \in \mathbb{R}^{n \times n}$;
            \State Alice and Bob jointly compute $I_B^*\times P=U_a + U_b = Q\cdot T^{-1} \cdot P$ with \textbf{S2PM} protocol, where $U_a, U_b \in \mathbb{R}^{n \times n}$;
            \State \Return $U_a$,$U_b$.
        \end{algorithmic}
    \end{breakablealgorithm}

\noindent Furthermore, the correctness of results can be easily verified with $U_a+U_b=Q\cdot T^{-1}\cdot P=Q\cdot (V_a+V_b)^{-1}\cdot P=Q\cdot [(V_{a1}+V_{a2})+(V_{b1}+V_{b2})]^{-1}\cdot P=Q\cdot [(V_{a1}+V_{b1})+(V_{a2}+V_{b2})]^{-1}\cdot P=Q\cdot [(P\cdot A)\cdot Q+P\cdot (B\cdot Q)]^{-1}\cdot P=Q\cdot [Q^{-1}\cdot (A+B)^{-1}\cdot P^{-1}]\cdot P=(A+B)^{-1}$.

\subsubsection{Security of S2PI}
\noindent We prove the security of S2PI using a universal composability framework (UC)\cite{10.5555/874063.875553} that is supported by an important property that a protocol is secure when all its sub-protocols, running either in parallel or sequentially, are secure. Therefore, to simplify the proof of S2PI security, the following Lemma is used\cite{bogdanov2008sharemind}. 

\begin{lemma}\label{lemma2}
    A protocol is perfectly simulatable if all the sub-protocols are perfectly simulatable.
\end{lemma}

\begin{theorem}\label{theorem5}
    The S2PI protocol $f(A,B)=(A+B)^{-1}$ is secure in the honest-but-curious model.
\end{theorem}

\begin{proof}
    The protocol $f(A,B)=(A+B)^{-1}$ is implemented in a universal hybrid model in which the first step with two calls of S2PM running in parallel computes $T=P\cdot(A+B)\cdot Q=(P\cdot A)\cdot Q+P\cdot(B\cdot Q)=S2PM_1(P\cdot A,Q)+S2PM_2(P,B\cdot Q)$ and the second step invokes the third call of S2PM sequentially for the computation of $(A+B)^{-1}=S2PM_3({Q\cdot T}^{-1},P)$.
    According to Lemma \ref{lemma2}, the security of S2PI can be reduced to the security of the sequential composition of the above two steps respectively.
    Given that the outputs $(V_{a1,}V_{a2}), (V_{b1},V_{b2})$ from the first step with two calls of S2PM are individual to Alice and Bob respectively and the view in real-world is computationally indistinguishable from the simulated view in the ideal-world according to Theorem 4, the summation $T=V_a + V_b=(V_{a1}+V_{a2})+(V_{b1}+V_{b2})$ can be easily proven to be computational indistinguishable. The S2PM in second step can also be perfectly simulated in the ideal-world and is indistinguishable from the view in the real-world. Therefore, the S2PI protocol $f(A,B)=(A+B)^{-1}$ is secure in the honest-but curious model.
\end{proof}

\subsection{S2PHM}\label{S2PHM}
\noindent The problem of S2PHM is defined as:
\begin{problem}[Secure 2-Party Matrix Hybrid Multiplication Problem]
    Alice has private matrices $(A_1,A_2)$ and Bob has private matrices $(B_1,B_2)$, where $(A_1,B_1) \in \mathbb{R}^{n\times s}$,$(A_2,B_2) \in \mathbb{R}^{s\times m}$. They want to conduct the hybrid multiplication $f[(A_1, A_2), (B_1, B_2)] = (A_1 + B_1) \cdot (A_2+B_2)$ in which Alice gets $V_a$ and Bob gets $V_b$ such that $V_a + V_b =  (A_1 + B_1) \cdot (A_2+B_2)$.
\end{problem}
\subsubsection{Description of S2PHM}
\noindent S2PHM protocol is made of two calls of S2PM protocol as shown in Algorithm \ref{alg:S2PHM}.  
The probability of failing verification ${Pf}_{S2PHM} = {{Pf}_{S2PM}}^2 \leq (\frac{1}{4^l})^2 \approx 8.27\times 10^{-25w} ($l = 20$)$. And then, the correctness of results can be easily verified with $V_a+V_b=(V_{a0}+V_{a1}+V_{a2})+(V_{b0}+V_{b1}+V_{b2})=(V_{a0}+V_{b0})+(V_{a1}+V_{b1})+(V_{a2}+V_{b2})=A_1\cdot A_2+B_1\cdot B_2+A_1\cdot B_2+B_1\cdot A_2=(A_1+B_1)\cdot(A_2+B_2)$.
\renewcommand{\thealgorithm}{8} %
    \begin{breakablealgorithm}
        \caption{S2PHM}%
        \label{alg:S2PHM}
        \begin{algorithmic}[1] 
            \Require Alice holds private matrices ($A_1,A_2$), Bob holds private matrices ($B_1,B_2$), where $(A_1,B_1) \in \mathbb{R}^{n\times s}$ ,$(A_2,B_2) \in \mathbb{R}^{s\times m}$
            \Ensure Alice and Bob output $V_a, V_b \in\mathbb{R}^ {n \times m}$, respectively. 
            \State Alice computes $V_{a0} = A_1 \times A_2 \in \mathbb{R}^{n \times m}$;
            \State Bob computes $V_{b0} = B_1 \times B_2 \in \mathbb{R}^{n \times m}$;
            \State Alice and Bob jointly compute $V_{a1} + V_{b1} = A_1 \times B_2$ with \textbf{S2PM}, where $V_{a1}, V_{b1} \in \mathbb{R}^{n \times m}$;
            \State Alice and Bob jointly compute $V_{a2} + V_{b2} = B_1 \times A_2$ with \textbf{S2PM}, where $V_{a2}$, $V_{b2} \in \mathbb{R}^{n \times m}$;
            \State Alice sums $V_a = V_{a0} + V_{a1} + V_{a2}$, where $V_a \in \mathbb{R}^{n \times m}$;
            \State Bob sums $V_b = V_{b0} + V_{b1} + V_{b2}$, where $V_b \in \mathbb{R}^{n \times m}$;
            \State \Return $V_a,V_b$ 
        \end{algorithmic}
    \end{breakablealgorithm}
\noindent The following theorem about the security of S2PHM can be intuitively proved using the universal composability property in a similar way as S2PI.

\begin{theorem}\label{theorem7}
    The protocol S2PHM, which can be represented by $f[(A_1, A_2), (B_1, B_2)] = (A_1 + B_1) \cdot (A_2+B_2)$, is secure in the honest-but-curious model.
\end{theorem}

\subsection{S3PHM}\label{S3PHM}
\noindent The problem of S3PHM is defined as:
\begin{problem}[Secure 3-Party Matrix Hybrid Multiplication Problem]
    Alice has $n\times s$ matrices $(A_1,A_2)$, Bob has $s\times t$ matrices $(B_1,B_2)$ and Carol has a $t\times m$ matrix C. They want to conduct the multiplication computation $f[(A_1,A_2),(B_1,B_2),C]=(A_1 + B_1)\cdot (A_2+B_2)\cdot C$ in which Alice gets $V_a$, Bob gets $V_b$ and Carol gets $V_c$ such that $V_a + V_b +V_c = (A_1 + B_1)\cdot (A_2 + B_2)\cdot C$.
\end{problem}
\subsubsection{Description of S3PHM}
\noindent S3PHM protocol can be achieved invoking twice S2PM protocols and trice run in parallel, and input from each sub-protocol is totally individual without any interactive.  
All sub-protocols strictly follow the same procedures as S2PM  or S3PM, and an upper bound probability of jointly verification can be intuitively illustrate as ${Pf}_{S3PHM}={{Pf}_{S2PM}}^2\cdot {{Pf}_{S3PM}}^2 \leq (\frac{1}{4^l})^2 \times (\frac{1}{8^l})^2 \approx 6.22\times 10^{-61}$ ($l = 20$) in the end.
\renewcommand{\thealgorithm}{9} %
    \begin{breakablealgorithm}
        \caption{S3PHM}%
        \label{alg:S3PHM}
        \begin{algorithmic}[1] %
            \Require Alice, Bob and Carol hold private matrices ($A_1,A_2$), ($B_1,B_2$), C where $(A_1,B_1) \in \mathbb{R}^{n\times s}$ ,$(A_2,B_2) \in \mathbb{R}^{s\times t}$ and $C\in \mathbb{R}^{t\times m}$, respectively.
            \Ensure Alice, Bob and Carol get $V_a, V_b, V_c \in\mathbb{R}^ {n \times m}$, respectively. 
            \State Alice computes $A^* = A_1 \times A_2 \in \mathbb{R}^{n \times t}$.
            \State Bob computes $B^* = B_1 \times B_2 \in \mathbb{R}^{n \times t}$;
            \State Alice and Carol jointly compute $V_{a0} + V_{c0} = A^*\times C$ with \textbf{S2PM}, where $V_{a0}, V_{c0} \in \mathbb{R}^{n \times m}$;
            \State Bob and Carol jointly compute $V_{b3} + V_{c3} = B^*\times C$ with \textbf{S2PM}, where $V_{b3}, V_{c3} \in \mathbb{R}^{n \times m}$;
            \State Alice, Bob, and Carol jointly compute \textbf{S3PM} $V_{a1} + V_{b1} + V_{c1} = A_1 \times B_2 \times C$, where $V_{a1}, V_{b1}, V_{c1} \in \mathbb{R}^{n \times m}$;
            \State Alice, Bob, and Carol jointly compute \textbf{S3PM} $V_{a2} + V_{b2} + V_{c2} = B_1 \times A_2 \times C$, where $V_{a2}, V_{b2}, V_{c2} \in \mathbb{R}^{n \times m}$;
            \State  Alice sums $V_a = V_{a0} + V_{a1} + V_{a2}$, where $V_a \in \mathbb{R}^{n \times m}$;
            \State Bob sums $V_b = V_{b1} + V_{b2} + V_{b3}$, where $V_b \in \mathbb{R}^{n \times m}$;
            \State Carol sums $V_c = V_{c0} + V_{c1} + V_{c2} + V_{c3}$, where $V_c \in \mathbb{R}^{n \times m}$.
            \State \Return $V_a,V_b,V_c$ 
        \end{algorithmic}
    \end{breakablealgorithm}
 
 And then, the correctness of results can be easily verified with $V_a+V_b+V_c=(V_{a0}+V_{a1}+V_{a2})+(V_{b1}+V_{b2}+V_{b3})+(V_{c0}+V_{c1}+V_{c2}+V_{c3})=(V_{a0}+V_{c0})+(V_{a1}+V_{b1}+V_{c1})+(V_{a2}+V_{b2}+V_{c2})+(V_{b3}+V_{c3})=A_1\cdot A_2\cdot C+A_1\cdot B_2\cdot C+B_1\cdot A_2\cdot C+B_1\cdot B_2\cdot C=(A_1+B_1)\cdot (A_2+B_2)\cdot C$.
 
The following theorem about the security of S3PHM can be intuitively proved using the universal composability property in a similar way as S2PI.

\begin{theorem}\label{theorem8}
    The protocol S3PHM, which can be represented by $f[(A_1,A_2),(B_1,B_2),C]=(A_1 + B_1)\cdot (A_2+B_2)\cdot C$, is secure in the honest-but-curious model.
\end{theorem}

\section{APPLICATIONS IN LINEAR REGRESSION}\label{Applications}
In this section, we will use the basic protocols above to build a secure 3-party LR (S3PLR) algorithm applied on distributed databases for joint multivariate analysis. The algorithm consists of two parts: the secure 3-party LR training (S3PLRT) and the secure 3-party LR prediction (S3PLRP).

\subsection{Heterogeneous Distributed Database}
Firstly, let $M=\{X:Y\}$ be a dataset consists of $n$ samples with $m$ features, denoted as $X_i=[x_{i1},x_{i2},\cdots ,x_{im}](1\leq i\leq n)$ and labels $Y=[y_1,y_2,\cdots,y_n]$. We presume a generic scenario in which the columns of dataset $M$ is vertically divided into $t$ databases $\{M_1:M_2:\cdots :M_t\}$ that belong to owners $P_1,P_2,\cdots ,P_t$ where $M_i\in\mathbb{R}^{n \times m_i},m+1=\Sigma_{i=1}^{t} {m_i}(1\leq i\leq t)$. These databases form a heterogeneous distributed
database when the following two constraints are satisfied (1) $ \forall i \neq j, M_i\cap M_j=\varnothing (1\leq i,j \leq t)$ and (2) $M_1\cup M_2\cup \cdots M_t = M$. Notably, the heterogeneous database can also denoted as $M=\{M_1:\textbf{0}\} + \cdots + \{\textbf{0}:M_i:\textbf{0}\}+\cdots +\{\textbf{0}:M_t\}=X_1+X_2+\cdots X_t$, where $X_i\in \mathbb{R}^{n\times (m+1)}(1\leq i\leq t)$.

\subsection{S3PLR}\label{S3PLR}
In a traditional linear regression $Y=X \cdot \beta$, the least square estimator of the weight parameter is $\beta=(X^T X)^{-1}(X^{T})Y$. With a test dataset $X^*$, the prediction is made as $\hat{Y}=X^*\cdot {\beta}$.
\begin{figure}[ht]
  \centering
  \includegraphics[width=1.0\hsize]{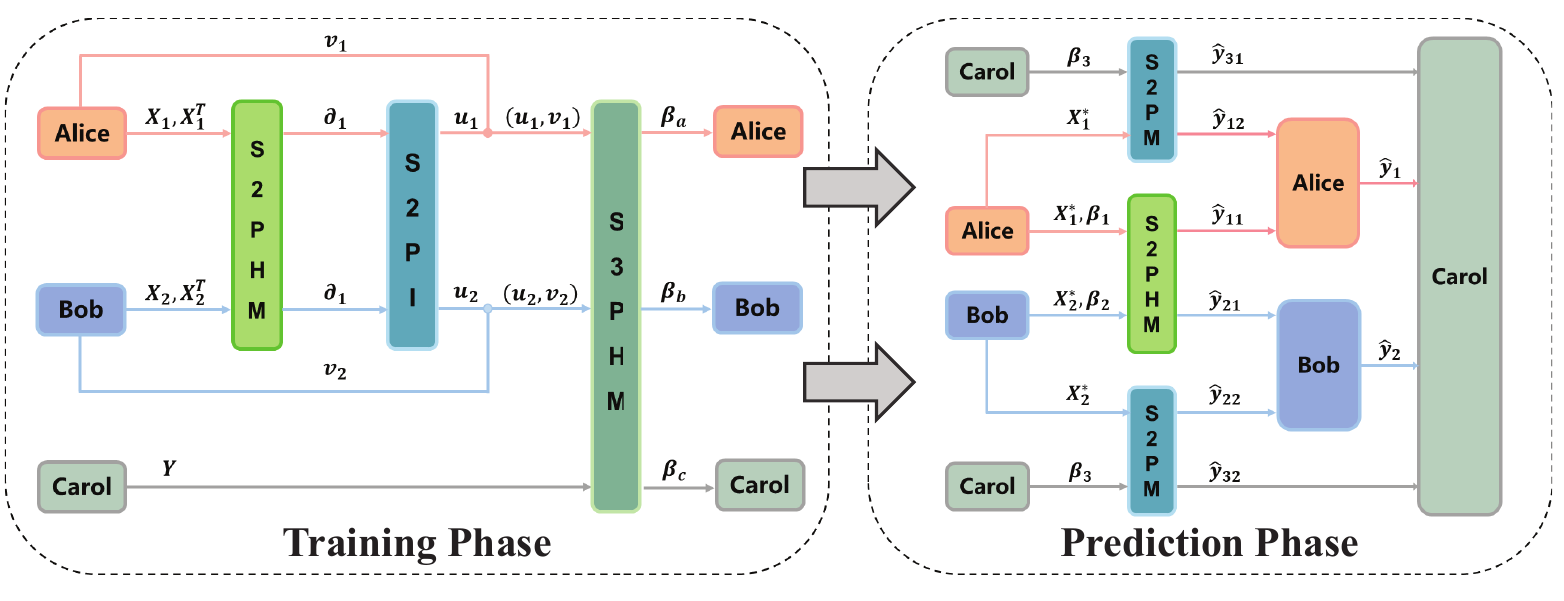}
  \vspace{-0.3cm}
  \caption{S3PLR Training and Prediction Based on Basic Secure Protocol Library}
  \label{fig:S3PLR}
\Description[<short description>]{<long description>}
\end{figure}
\subsubsection{Description of S3PLR}
\subsubsection{Definition of S3PLR Problem} 
In a 3-party heterogeneous database, the training and test datasets are divided into three parties $M=\{A:B:Y\}$ and $M^*=\{A^*:B^*:Y^*\}$. Using the convention above $X=(A:B)=X_1+X_2$ and $X^*=(A^*:B^*)=X_1^*+X_2^*$, the training and prediction problem of S3PLR is defined as: 
\begin{problem}[Secure 3-Party Linear Regression Training Problem]    
    Alice has a $n_1\times m$ training set $X_1$, Bob has a $n_1\times m$ training set $X_2$, and Carol has a $n_1\times 1$ label $Y$. They want to find out a solution $\beta \in \mathbb{R}^{m \times 1}$ satisfying $\beta=f(X_1,X_2,Y)=[(X_1+X_2)^T(X_1+X_2)]^{-1}\cdot (X_1^T+X_2^T)\cdot Y$, in which Alice gets $\beta_1$, Bob gets $\beta_2$, and Carol gets $\beta_3$ such that $\beta_1+\beta_2+\beta_3=\beta$.
\end{problem}
\begin{problem}[Secure 3-Party Linear Regression Prediction Problem]  
    Alice has a $n_2\times m$ test set $X_1^*$ with regression coefficient $\beta_1$, Bob has a $n_2\times m$ test set $X_2^*$ with regression coefficient $\beta_2$, and Carol has a regression coefficient $\beta_3$. They want to infer the prediction value of $\hat{Y}$ as $\hat{Y}=f((X_1^*,\beta_1),(X_2^*,\beta_2),\beta_3)=X^*\beta=(X_1^*+X_2^*)(\beta_1+\beta_2+\beta_3)$, where Alice gets $\hat{Y_1}$, Bob gets $\hat{Y_2}$, and Carol gets $\hat{Y_3}$ such that  $\hat{Y}=\hat{Y_1}+\hat{Y_2}+\hat{Y_3}$.
\end{problem}
\noindent The implementation of training and prediction of S3PLR is shown in figure \ref{fig:S3PLR}. The training stage can be straightforwardly reduced to a S2PHM problem, a S2PI problem and a S3PHM problem proposed in section \ref{Proposed-Work}. In the test stage, the solution $\hat{Y}=X^*\beta=(X_1^*+X_2^*)(\beta_1+\beta_2+\beta_3)=(X_1^*+X_2^*)(\beta_1+\beta_2)+(X_1^*\cdot \beta_3)+(X_2^*\cdot \beta_3)$ can also be reduced to combination of S2PM problem and S2PHM problem.

\noindent \textbf{Training Phase.}
The proposed training phase of S3PLR (S3PLRT) is made of three consecutive call of S2PHM, S2PI, and S3PHM as illustrated in Algorithm \ref{alg:S3PLRT}. The probability of failing the verification ${Pf}_{S3PLRT} = {{Pf}_{S2PM}}\cdot{{Pf}_{S2PI}}\cdot{{Pf}_{S3PHM}} \leq (\frac{1}{4^l})^7\times (\frac{1}{8^l})^2 \approx 3.87\times 10^{-121} ($l = 20$)$. Moreover, the correctness can be verified with $\beta_1+\beta_2+\beta_3=(u_1+u_2)(v_1+v_2)Y=(\partial_1+\partial_2)^{-1}(v_1+v_2)Y=[({X_1}^T+{X_2}^T)(X_1+X_2)]^{-1}({X_1}^T+{X_2}^T)Y=(X^TX)^{-1}X^TY=\beta$.

\renewcommand{\thealgorithm}{7} 
    \begin{breakablealgorithm}
        \caption{S3PLRT}
        \label{alg:S3PLRT}
        \begin{algorithmic}[1] 
            \Require Alice has a training set $X_1 \in \mathbb{R}^ {n_1 \times m}$ , Bob has a training set $X_2 \in \mathbb{R}^{n_1 \times m}$ , Carol has the label $Y\in \mathbb{R}^{n_1\times1}$ 
            \Ensure Alice,Bob and Carol get the model parameters $\beta_a, \beta_b, \beta_c \in\mathbb{R}^ {m \times 1}$, respectively.
            \State Alice and Bob execute \textbf{S2PHM} with their private matrices $X_1$ and $X_2 $ and get matrices $\partial_1, \partial_2 \in \mathbb{R}^{m \times m}$ respectively such that $\partial_1 + \partial_2 = (X_1^T + X_2^T) \cdot (X_1 + X_2)$.
            \State Alice and Bob use $\partial_1, \partial_2$ as input for a \textbf{S2PI} and produce the matrices $u_1, u_2 \in \mathbb{R}^{m \times m}$ at the two parties respectively, such that $u_1 + u_2 = (\partial_1 + \partial_2)^{-1}$;
            \State Alice computes $v_1=X_1^T$ and Bob computes $v_2=X_2^T$.
            \State Alice, Bob, and Carol collaboratively execute \textbf{S3PHM} with their respective matrix pairs $(u_1, v_1)$, $(u_2, v_2)$, and $Y$. The result is randomly split into matrices $\beta_1, \beta_2, \beta_3 \in \mathbb{R}^{m \times 1}$ at the three parties respectively, such that $\beta_1 + \beta_2 + \beta_3=(u_1 + u_2) \cdot (v_1 + v_2) \cdot Y$.
            \State \Return $\beta_1,\beta_2,\beta_3$ 
        \end{algorithmic}
    \end{breakablealgorithm}

\noindent \textbf{Prediction Phase.}
The proposed prediction phase of S3PLR (S3PLRP) is made of three parallel calls of a S2PHM and two S2PM as illustrated in Algorithm \ref{alg:S3PLRP}. The probability of failing a verification ${Pf}_{S3PLRP} = {{Pf}_{S2PM}}^2\cdot{{Pf}_{S2PHM}} \leq (\frac{1}{4^l})^4 \approx 6.84\times 10^{-49} ($l = 20$)$. Moreover, the correctness can be verified with $\hat{Y_1}+\hat{Y_2}+\hat{Y_3}=\widehat{y}_{11}+\widehat{y}_{12}+\widehat{y}_{21}+\widehat{y}_{22}+\widehat{y}_{31}+\widehat{y}_{32}=(\widehat{y}_{11}+\widehat{y}_{21})+(\widehat{y}_{12}+\widehat{y}_{31})+(\widehat{y}_{22}+\widehat{y}_{32})= (X_1^* + X_2^*) \cdot (\beta_1 + \beta_2)+ X_1^* \times \beta_3+X_2^* \times \beta_3=(X_1^*+X_2^*)(\beta_1+\beta_2+\beta_3)$.
\renewcommand{\thealgorithm}{8} 
    \begin{breakablealgorithm}
        \caption{S3PLRP}
        \label{alg:S3PLRP}
        \begin{algorithmic}[1] 
            \Require Alice has a test set  $X_1^* \in \mathbb{R}^{n_2\times m}$ with  model parameter fragment $\beta_1$, Bob has a test set $X_2^* \in \mathbb{R}^{n_2\times m}$with model parameter fragment $\beta_2$, Carol has model parameter fragment $\beta_3$, where $\beta_1,\beta_2,\beta_3 \in \mathbb{R}^{m\times 1}$
            \Ensure Alice, Bob and Carol get the model prediction  $\hat{Y_1},\hat{Y_2},\hat{Y_3} \in\mathbb{R}^ {n_2 \times 1}$,  respectively.  
            \State Alice and Bob execute \textbf{S2PHM} with their private matrices $(X_1^*,\beta_1)$ and $(X_2^*,\beta_2)$ and produce $\widehat{y}_{11}, \widehat{y}_{21} \in \mathbb{R}^{n_2 \times 1}$ at the two parties respectively, such that $\widehat{y}_{11} + \widehat{y}_{21} = (X_1^* + X_2^*) \cdot (\beta_1 + \beta_2)$;
            \State Alice and Carol execute \textbf{S2PM} with their private matrices $X_1^*$ and $\beta_3$ and produce $\widehat{y}_{12}, \widehat{y}_{31} \in \mathbb{R}^{n_2 \times 1}$ at the two parties respectively, such that $\widehat{y}_{12} + \widehat{y}_{31} = X_1^* \times \beta_3$;
            \State Bob and Carol execute \textbf{S2PM} with their private matrices $X_2^*$ and $\beta_3$ and produce $\widehat{y}_{22}, \widehat{y}_{32} \in \mathbb{R}^{n_2 \times 1}$ at the two parties respectively, such that $\widehat{y}_{22} + \widehat{y}_{32} = X_2^* \times \beta_3$;
            \State Alice computes $\hat{Y_1}=\widehat{y}_{11}+\widehat{y}_{12}$, Bob computes $\hat{Y_2}=\widehat{y}_{21}+\widehat{y}_{22}$ and Carol computes $\hat{Y_3}=\widehat{y}_{31}+\widehat{y}_{32}$;
            \State \Return $\hat{Y_1},\hat{Y_2},\hat{Y_3}$. 
        \end{algorithmic}
    \end{breakablealgorithm}
    
The following theorems \ref{theorem9} and \ref{theorem10} about the security of S3PLRP and S3PLRT can be intuitively proved using the universal composability property.
\begin{theorem}\label{theorem9}
    The protocol S3PLRT, which can be represented by $f((X_1^*,\beta_1),(X_2^*,\beta_2),\beta_3)=(X_1^*+X_2^*)(\beta_1+\beta_2+\beta_3)$, is secure in the honest-but-curious model.
\end{theorem}
\begin{theorem}\label{theorem10}
    The protocol S3PLRP, which can be represented by $=f((X_1^*,\beta_1),(X_2^*,\beta_2),\beta_3)=X^*\beta=(X_1^*+X_2^*)(\beta_1+\beta_2+\beta_3)$, is secure in the honest-but-curious model.
\end{theorem}

\section{THEORETICAL ANALYSIS}\label{Theoretical}
In this section, theoretical analysis about the computational complexity and communication cost of the protocols will be provided. The comparison is carried out on S2PM, S3PM, S2PI, S2PHM, S3PHM, S3PLRT and S3PLRP implemented with the method described in last section and with Du's work [15] respectively. Since Du's work (including S2PM, S2PI, and S2PLR with disclosed label) is limited to matrix operation between two parties, the three-party computation (S3PM, S3PHM, and S3PLR) is extended by consecutively invoking S2PM. The detailed comparison is summarized in Table \ref{tab:Complexity Analysis}.

\subsection{Computation Complexity}
The computation complexity is compared at preprocessing, online computing and verification stages for each protocol respectively.
The S2PM implementation is identical to Du's work except for the newly introduced verification module, which essentially is a multiplication between the $n\times m$ standard matrix $S_t$ and the $m\times 1$ vector $\hat{\delta_a}$ with the complexity of $\mathcal{O}(nm)$. As Du's scheme does not include a verification phase, the complexity of it is denoted as '$\backslash$' in Table \ref{tab:Complexity Analysis}.
For S3PM, a baseline protocol is constructed with Du's S2PM by conducting $\{f(A,B)=V_a+V_b=A\cdot B|A\in \mathbf{R}^{n\times s},B\in \mathbf{R}^{s\times t}\}$ and then $\{f(A,B,C)=(V_a+V_b)C=V_a\cdot C + V_b\cdot C|V_a,V_b\in \mathbf{R}^{n\times t},C\in \mathbf{R}^{t\times m}\}$, and therefore has a complexity of $\mathcal{O}(nst+ntm)$, in comparison to the complexity of $\mathcal{O}(ntm+nst+stm+nsm)$ for the S3PM protocol we proposed in this study. Because the protocols for other matrix operations are built on the combination of S2PM and S3PM, it is straightforward to analysis their complexity by following the corresponding protocols. 
Table \ref{tab:Complexity Analysis} summaries the complexity analysis. For easy of analysis, in the S3PLRT and S3PLRP, it is assumed that sample sizes satisfy $n_1=n_2=n$. Although the verification phase introduces extra cost in complexity, the total complexity of the protocols remain the cubic order of input dimensions that is identical to Du's work.

\begin{table}[htbp]
  \centering 
  \caption{Comparison of The Computation Complexity Between Our Proposed Scheme and Du's Scheme} 
    \resizebox{\linewidth}{!}{
    \begin{tabular}{ccccccccc}
    \toprule
    \multicolumn{2}{c}{\multirow{2}[4]{*}{\textbf{Protocol}}} & \multicolumn{5}{c}{\textbf{Phase Complexity}} & \multicolumn{2}{c}{\multirow{2}[4]{*}{\textbf{Total Complexity}}} \\
    \cmidrule{3-7}    
    \multicolumn{2}{c}{} & \multicolumn{2}{c}{\textbf{Preprocessing}} & \multicolumn{2}{c}{\textbf{Online Computing}} & \textbf{Verification} & \multicolumn{2}{c}{} \\
    \midrule
    \multirow{2}[1]{*}{\textbf{S2PM}} & \textbf{Our Scheme} & \multicolumn{2}{c}{$\mathcal{O}(nsm)$} & \multicolumn{2}{c}{$\mathcal{O}(nsm)$} & $\mathcal{O}(nm)$ & \multicolumn{2}{c}{$\mathcal{O}(nsm)$} \\
          & \textbf{Du's Scheme} & \multicolumn{2}{c}{$\mathcal{O}(nsm)$} & \multicolumn{2}{c}{$\mathcal{O}(nsm)$} & $\backslash$ & \multicolumn{2}{c}{$\mathcal{O}(nsm)$} \\
    \cmidrule{2-9}
    \multirow{2}[0]{*}{\textbf{S3PM}} & \textbf{Our Scheme} & \multicolumn{2}{c}{$\mathcal{O}(nst+ntm)$} & \multicolumn{2}{c}{$\mathcal{O}(ntm+nst+stm+nsm)$} & $\mathcal{O}(nm)$ & \multicolumn{2}{c}{$\mathcal{O}(ntm+nst+stm+nsm)$} \\
          & \textbf{Du's Scheme} & \multicolumn{2}{c}{$\mathcal{O}(nst+ntm)$} & \multicolumn{2}{c}{$\mathcal{O}(nst+ntm)$} & $\backslash$ & \multicolumn{2}{c}{$\mathcal{O}(nst+ntm)$} \\
    \cmidrule{2-9}
    \multirow{2}[0]{*}{\textbf{S2PI}} & \textbf{Our Scheme} & \multicolumn{2}{c}{$\mathcal{O}(n^3)$} & \multicolumn{2}{c}{$\mathcal{O}(n^3)$} & $\mathcal{O}(n^2)$ & \multicolumn{2}{c}{$\mathcal{O}(n^3)$} \\
          & \textbf{Du's Scheme} & \multicolumn{2}{c}{$\mathcal{O}(n^3)$} & \multicolumn{2}{c}{$\mathcal{O}(n^3)$} & $\backslash$ & \multicolumn{2}{c}{$\mathcal{O}(n^3)$} \\
    \cmidrule{2-9}
    \multirow{2}[0]{*}{\textbf{S2PHM}} & \textbf{Our Scheme} & \multicolumn{2}{c}{$\mathcal{O}(nsm)$} & \multicolumn{2}{c}{$\mathcal{O}(nsm)$} & $\mathcal{O}(nm)$ & \multicolumn{2}{c}{$\mathcal{O}(nsm)$} \\
          & \textbf{Du's Scheme} & \multicolumn{2}{c}{$\mathcal{O}(nsm)$} & \multicolumn{2}{c}{$\mathcal{O}(nsm)$} & $\backslash$ & \multicolumn{2}{c}{$\mathcal{O}(nsm)$} \\
    \cmidrule{2-9}
    \multirow{2}[0]{*}{\textbf{S3PHM}} & \textbf{Our Scheme} & \multicolumn{2}{c}{$\mathcal{O}(ntm+nst)$} & \multicolumn{2}{c}{$\mathcal{O}(ntm+nst+stm+nsm)$} & $\mathcal{O}(nm)$ & \multicolumn{2}{c}{$\mathcal{O}(ntm+nst+stm+nsm)$} \\
          & \textbf{Du's Scheme} & \multicolumn{2}{c}{$\mathcal{O}(ntm+nst)$} & \multicolumn{2}{c}{$\mathcal{O}(ntm+nst)$} & $\backslash$ & \multicolumn{2}{c}{$\mathcal{O}(ntm+nst)$} \\
    \cmidrule{2-9}
    \multirow{2}[0]{*}{\textbf{S3PLRT}} & \textbf{Our Scheme} & \multicolumn{2}{c}{$\mathcal{O}(m^3+m^2 n)$} & \multicolumn{2}{c}{$\mathcal{O}(m^3+m^2 n)$} & $\mathcal{O}(m^2)$ & \multicolumn{2}{c}{$\mathcal{O}(m^3+m^2 n)$} \\
          & \textbf{Du's Scheme} & \multicolumn{2}{c}{$\mathcal{O}(m^3+m^2 n)$} & \multicolumn{2}{c}{$\mathcal{O}(m^3+m^2 n)$} & $\backslash$ & \multicolumn{2}{c}{$\mathcal{O}(m^3+m^2 n)$} \\
    \cmidrule{2-9}
    \multirow{2}[1]{*}{\textbf{S3PLRP}} & \textbf{Our Scheme} & \multicolumn{2}{c}{$\mathcal{O}(nm)$} & \multicolumn{2}{c}{$\mathcal{O}(nm)$} & $\mathcal{O}(n)$ & \multicolumn{2}{c}{$\mathcal{O}(nm)$} \\
          & \textbf{Du's Scheme} & \multicolumn{2}{c}{$\mathcal{O}(nm)$} & \multicolumn{2}{c}{$\mathcal{O}(nm)$} & $\backslash$ & \multicolumn{2}{c}{$\mathcal{O}(nm)$} \\
    \bottomrule
    \end{tabular}}
  \label{tab:Complexity Analysis}
\end{table}%

\subsection{Communication Cost}
Compared with the local computation, network communication is the main bottleneck of interactive protocols, so they are in favor of method with lower commutations cost and smaller rounds of interaction. For ease of analysis, we assume that all matrices are square with dimension of $n\times n$ and each element is encoded in length $\ell$.The comparison is summerised in Table \ref{tab:Communication Analysis}.
For S2PM, the computation of $VF_a,VF_b,VF_c,S_t$ for verification introduces slightly more communication cost and rounds compared with the original Du's work. For S3PM, because the baseline method uses 3 calls of S2PM to accomplish the computation, the communication cost and rounds are higher compared with the proposed S3PM protocol. Intuitively, in other protocols made of S2PM and S3PM, the more S3PM they use, the lower communication cost and rounds they will have. Therefore, S3PHM, S3PLRT demonstrates lower communication cost, while S2PHM and S3PLRP are identical because they are purely based on S2PM. For S2PI, due to the optimization of separately stored $P$ and $Q$, the communication cost is reduced to $75\%$ of baseline with 4 calls of S2PM.

\begin{table}[htbp]
  \centering 
  \caption{Comparison of Communication Cost and Rounds Between Our Proposed Scheme and Du's Scheme} 
    \resizebox{\linewidth}{!}{
   \begin{tabular}{ccccccccccc}
    \toprule
    \multicolumn{2}{c}{\multirow{2}[4]{*}{\textbf{Protocol}}} & \multicolumn{4}{c}{\textbf{Communication Cost [\#bits]}} &       & \multicolumn{4}{c}{\textbf{\#Rounds}} \\
\cmidrule{3-6}\cmidrule{8-11}    \multicolumn{2}{c}{} & \multicolumn{2}{c}{\textbf{Our Scheme}} & \multicolumn{2}{c}{\textbf{Du's Scheme}} &       & \multicolumn{2}{c}{\textbf{Our Scheme}} & \multicolumn{2}{c}{\textbf{Du's Scheme}} \\
    \midrule
    \multicolumn{2}{c}{\textbf{S2PM}} & \multicolumn{2}{c}{$(11n^2)\ell$} & \multicolumn{2}{c}{$(7n^2)\ell$} 
    &       & \multicolumn{2}{c}{6} & \multicolumn{2}{c}{5} \\
    \multicolumn{2}{c}{\textbf{S3PM}} & \multicolumn{2}{c}{$(26n^2)\ell$} & \multicolumn{2}{c}{$(33n^2)\ell$} &       & \multicolumn{2}{c}{15} & \multicolumn{2}{c}{18} \\
    \multicolumn{2}{c}{\textbf{S2PI}} & \multicolumn{2}{c}{$(34n^2)\ell$} & \multicolumn{2}{c}{$(45n^2)\ell$} &       & \multicolumn{2}{c}{19} & \multicolumn{2}{c}{25} \\
    \multicolumn{2}{c}{\textbf{S2PHM}} & \multicolumn{2}{c}{$(22n^2)\ell$} & \multicolumn{2}{c}{$(22n^2)\ell$} &       & \multicolumn{2}{c}{12} & \multicolumn{2}{c}{12} \\
    \multicolumn{2}{c}{\textbf{S3PHM}} & \multicolumn{2}{c}{$(74n^2)\ell$} & \multicolumn{2}{c}{$(88n^2)\ell$} &       & \multicolumn{2}{c}{42} & \multicolumn{2}{c}{48} \\
    \multicolumn{2}{c}{\textbf{S3PLRT}} & \multicolumn{2}{c}{$(76n^2+54n)\ell$} & \multicolumn{2}{c}{$(101n^2+54n)\ell$} &       & \multicolumn{2}{c}{73} & \multicolumn{2}{c}{85} \\
    \multicolumn{2}{c}{\textbf{S3PLRP}} & \multicolumn{2}{c}{$(8n^2+36n)\ell$} & \multicolumn{2}{c}{$(8n^2+36n)\ell$} &       & \multicolumn{2}{c}{24} & \multicolumn{2}{c}{24} \\
    \bottomrule
    \end{tabular}}
  \label{tab:Communication Analysis}
\end{table}%

\section{PERFORMANCE EVALUATION}\label{Experiments}
\noindent This section evaluates the performance of five sub-protocols (S2PM, S3PM, S2PI, S2PHM, S3PHM) and their application in three-party linear regression.

\noindent\textbf{Experimental Setup:}
All EVA-S3PC protocols were implemented in Python as separate modules. Performance experiments for the 3-party computation were conducted on a cloud machine with 32 vCPUs (Intel Xeon Platinum 8358), 96 GB RAM, and Ubuntu 20.04 LTS.
To mitigate inconsistencies in communication protocols (gRPC vs. Socket) that can amplify performance differences between WAN and LAN, and to avoid excessive communication overhead masking protocol computation, we evaluated the EVA-S3PC framework in a LAN environment with 10.1 Gbit/s bandwidth and 0.1 ms latency, simulating optimal communication conditions.
The server infrastructure simulates five computing nodes via Docker. One acts as a semi-honest third party (CS server) sending random matrices during preparation and is not involved in further computation. Another serves as the requester, receiving only the final result. This exploratory study assumes an ideal semi-honest environment, excluding collusion.

\subsection{Evaluation of Sub-Protocols}\label{settings}
\noindent We benchmark the running time, communication overhead, and precision of five different sub-protocols for elementary operators in EVA-S3PC, and compare them with four previous SMPC schemes: LibOTe (OT) \cite{libOTe}, ABY3 (GC) \cite{mohassel2018aby3}, CryptGPU (SS) \cite{tan2021cryptgpu}, Tenseal (HE) \cite{benaissa2021tenseal}, and in a semi-honest environment. Since Du's work is purely theoretical without experimental implementation, we implemented their framework, named S2PM-Based, in Python for comparison in our experiments.

\noindent\textbf{Choice of Parameters:}
Although the proposed protocols are theoretically over the infinite field $\mathbb{R}$, real-world computers handle data in finite fields with limited precision length. Prior work \cite{wagh2019securenn} used truncation protocols in $\mathbb{Z}_{2^l}$ for fixed-point numbers to preserve precision, but this limits the range of floating-point computations. For fair comparison, input matrices for running time testing in 2-party and 3-party protocols are Float 64 (ring size $l=64$), with elements sampled from $x \in [10^{-4}, 10^{4}]$ and represented by a 15-bit significant figure "$1.a_1a_2\cdots a_{15}\times 10^\delta (\delta\in \mathbb{Z}=[-4,4]])$". Precision tests are divided into 6 ranges, from $[E0, E0]$ to $[E-10, E+10]$, with exponents incremented by $\Delta=2$.

\subsubsection{Efficiency and Overhead of Sub-protocols}
\noindent \autoref{tab:Runtime of various SMPC} and \autoref{fig: Overhead of Different SMPC Framework} show the average running time (including both computation and communication time) of five sub-protocols implemented by six representative schemes, along with the communication overhead (in KB) for performing 1000 repeated computations with matrix dimension of $N=10, 20, 30, 40, 50$. \autoref{fig: Allocation of S2PM/S3PM} illustrates the overhead of different modules in the S2PM and S3PM protocols, which are the foundational protocols used in EVA-S3PC. \autoref{fig: Overhead and Big Matrix} demonstrates the lightweight nature of the verification module and the efficiency of our EVA-S3PC protocols when applied to large matrix computations (dimension $N\geq1000$).
\begin{table}[htbp]
  \centering
  \caption{Comparison of total running time using various frameworks implementing five sub-protocols with matrix dimension of $N=10, 20, 30, 40, 50$. The speedup ratio of EVA-S3PC is relative to the S2PM based protocols which is in general better than other existing methods. $N/A$ indicates that HE does not support the S2PI protocol.}
    \resizebox{\linewidth}{!}{
    \begin{tabular}{cccccc>{\columncolor{gray!15}}c>{\columncolor{gray!30}}cc}
    \toprule
\multirow{2}[4]{*}{\textbf{Protocol}} & \multirow{2}[4]{*}{\textbf{Dimension}} & \multicolumn{6}{c}{\textbf{Running time (s)}}         & \multicolumn{1}{c}{\multirow{2}[4]{*}{\textbf{Speedup Ratio}}} \\
\cmidrule{3-8}          & \textbf{} & \textbf{HE} & \textbf{OT} & \textbf{GC} & {\textbf{SS}} & {\textbf{S2PM-Based}} & \textbf{EVA-S3PC} \\
    \midrule
    \multicolumn{1}{c}{\multirow{5}[2]{*}{\textbf{S2PM}}} & \textbf{10} & 3.11E+00 & 2.66E+00 & 3.75E-02 & 2.22E-03 & 3.47E-04 & \textbf{3.57E-04} & -2.84\% \\
          & \textbf{20} & 2.73E+01 & 1.04E+01 & 5.85E-02 & 2.84E-03 & 4.14E-04 & \textbf{4.26E-04} & -2.92\% \\
          & \textbf{30} & 9.10E+01 & 2.09E+01 & 7.95E-02 & 4.99E-03 & 5.35E-04 & \textbf{5.51E-04} & -2.88\% \\
          & \textbf{40} & 2.12E+02 & 3.85E+01 & 1.21E-01 & 8.44E-03 & 7.19E-04 & \textbf{7.38E-04} & -2.65\% \\
          & \textbf{50} & 4.09E+02 & 5.78E+01 & 1.88E-01 & 1.03E-02 & 8.70E-04 & \textbf{8.94E-04} & -2.81\% \\
    \midrule
    \multirow{5}[2]{*}{\textbf{S3PM}} & \textbf{10} & 4.33E+00 & 3.36E+00 & 5.23E-02 & 7.29E-03 & 1.01E-03 & \textbf{6.88E-04} & 31.62\% \\
          & \textbf{20} & 3.11E+01 & 1.41E+01 & 6.90E-02 & 8.17E-03 & 1.25E-03 & \textbf{8.95E-04} & 28.46\% \\
          & \textbf{30} & 9.95E+01 & 3.42E+01 & 1.04E-01 & 1.17E-02 & 1.83E-03 & \textbf{1.26E-03} & 30.86\% \\
          & \textbf{40} & 2.15E+02 & 6.43E+01 & 1.52E-01 & 1.91E-02 & 2.47E-03 & \textbf{1.69E-03} & 31.54\% \\
          & \textbf{50} & 4.56E+03 & 9.58E+01 & 2.45E-01 & 2.37E-02 & 3.05E-03 & \textbf{2.15E-03} & 29.43\% \\
    \midrule
    \multirow{5}[2]{*}{\textbf{S2PI}} & \textbf{10} & N\textbackslash{}A & 1.16E+01 & 1.15E-01 & 1.24E-02 & 7.65E-03 & \textbf{5.50E-03} & 28.07\% \\
          & \textbf{20} & N\textbackslash{}A & 4.85E+01 & 1.32E-01 & 2.23E-02 & 1.19E-02 & \textbf{8.14E-03} & 31.89\% \\
          & \textbf{30} & N\textbackslash{}A & 1.32E+02 & 1.45E-01 & 4.09E-02 & 1.63E-02 & \textbf{1.14E-02} & 29.85\% \\
          & \textbf{40} & N\textbackslash{}A & 2.17E+02 & 1.56E-01 & 5.51E-02 & 1.92E-02 & \textbf{1.30E-02} & 32.21\% \\
          & \textbf{50} & N\textbackslash{}A & 3.27E+02 & 1.62E-01 & 6.34E-02 & 2.47E-02 & \textbf{1.70E-02} & 31.12\% \\
    \midrule
    \multirow{5}[2]{*}{\textbf{S2PHM}} & \textbf{10} & 7.25E+01 & 6.30E+01 & 1.95E-01 & 7.31E-02 & 1.35E-03 & \textbf{1.42E-03} & -5.23\% \\
          & \textbf{20} & 8.13E+01 & 8.85E+01 & 2.35E-01 & 1.71E-01 & 1.51E-03 & \textbf{1.60E-03} & -5.93\% \\
          & \textbf{30} & 9.99E+01 & 1.14E+02 & 3.40E-01 & 2.90E-01 & 1.71E-03 & \textbf{1.81E-03} & -6.21\% \\
          & \textbf{40} & 2.11E+02 & 2.07E+02 & 5.06E-01 & 4.92E-01 & 2.05E-03 & \textbf{2.16E-03} & -5.07\% \\
          & \textbf{50} & 4.09E+02 & 3.36E+02 & 7.51E-01 & 7.44E-01 & 2.48E-03 & \textbf{2.60E-03} & -4.51\% \\
    \midrule
    \multirow{5}[2]{*}{\textbf{S3PHM}} & \textbf{10} & 7.95E+01 & 3.81E+01 & 3.28E-01 & 1.37E-01 & 5.95E-03 & \textbf{4.18E-03} & 29.63\% \\
          & \textbf{20} & 2.97E+02 & 9.37E+01 & 3.85E-01 & 2.44E-01 & 6.80E-03 & \textbf{4.87E-03} & 28.49\% \\
          & \textbf{30} & 5.57E+02 & 2.50E+02 & 5.23E-01 & 4.37E-01 & 7.04E-03 & \textbf{4.86E-03} & 31.01\% \\
          & \textbf{40} & 6.01E+02 & 4.92E+02 & 7.31E-01 & 6.82E-01 & 7.63E-03 & \textbf{5.40E-03} & 29.18\% \\
          & \textbf{50} & 8.89E+02 & 7.13E+02 & 1.13E+00 & 1.02E+00 & 9.00E-03 & \textbf{6.25E-03} & 30.51\% \\
    \bottomrule
    \end{tabular}}
  \label{tab:Runtime of various SMPC} 
\end{table}

\noindent\textbf{Comparison to Prior Work.}
As shown in \autoref{fig: Overhead of Different SMPC Framework}, as matrix dimensions increase, the communication overhead of HE and OT schemes significantly exceeds other protocols, while EVA-S3PM scheme incurs the smallest overhead (8\%-25\% of HE), with SS, GC, and S2PM-based method being slightly higher. \autoref{tab:Runtime of various SMPC} shows notable running time differences, with EVA-S3PC consistently outperforming others, achieving at least 28\% efficiency gains over S2PM based S3PM, S2PI, and S3PHM. In S2PM and S2PHM, our performance is slightly lower (2.65\%-6.21\%) due to the additional overhead from the security verification stage in the protocol.
HE performs poorly in all operators due to high communication costs and complexity and has difficulty in forming non-linear operations like matrix inversion S2PI. OT is less efficient in multiplication due to repeated mask generation and communication, and GC is limited by gate encryption/decryption in matrix multiplication. SS performs better than HE, OT and GC, thanks to its optimized local computation.
Both EVA-S3PC and Du’s work use random matrix-based disguising with natural advantages in linear operations by avoiding secret keys. As a result, EVA-S3PC outperforms others protocols in both communication overhead and running time.
\begin{figure}[hbtp]
    \centering
    \hspace{-0.5cm}
    \subfigure[S2PM Comm. Overhead]{\includegraphics[width=0.54\linewidth]{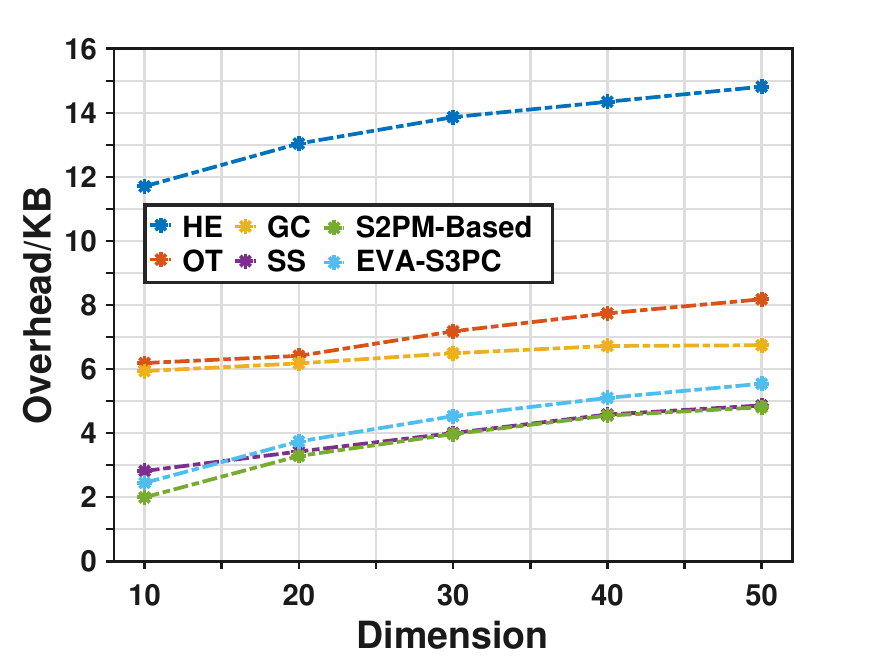}
    \label{fig: S2PM Comm.Overhead}}
    \hspace{-0.6cm}
    \subfigure[S3PM Comm. Overhead]{\includegraphics[width=0.54\linewidth]{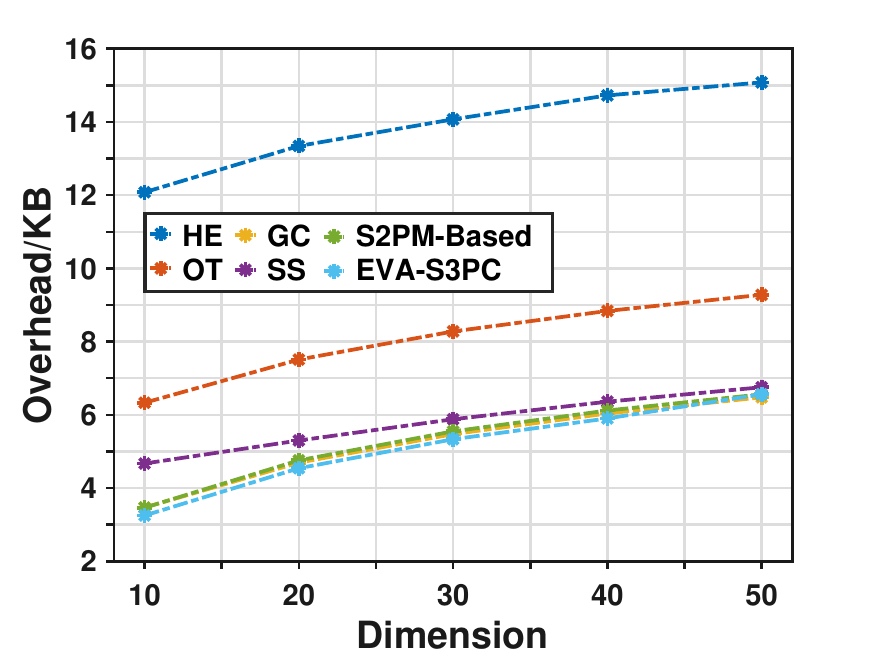}
    \label{fig: S3PM Comm.Overhead}}
    \hspace{-0.7cm}
    \vspace{-0.3cm}

    \centering
    \hspace{-0.5cm}
    \subfigure[S2PI Comm. Overhead]{\includegraphics[width=0.54\linewidth]{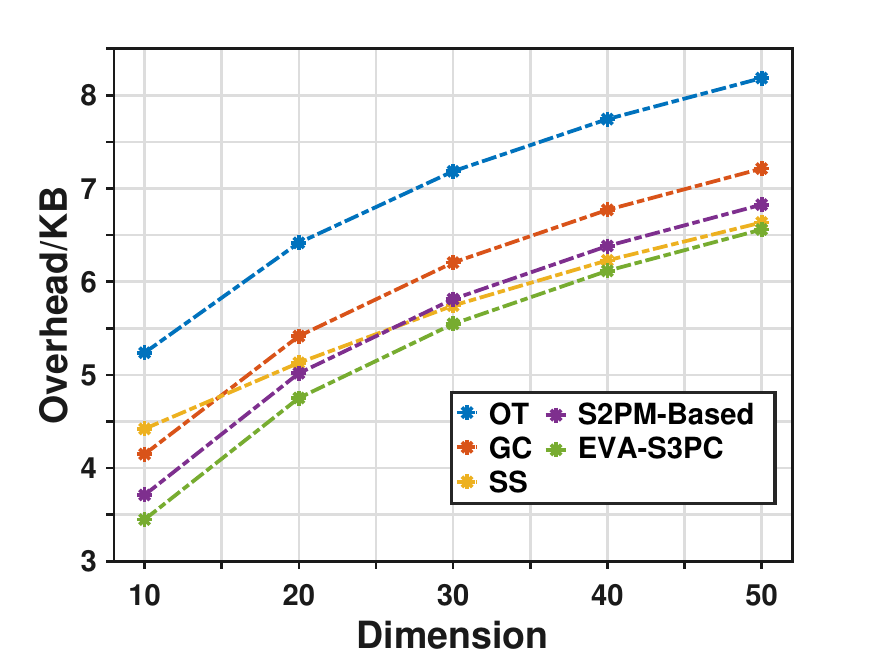}
    \label{fig: S2PI Comm.Overhead}}
    \hspace{-0.6cm}
    \subfigure[S3PHM Comm. Overhead]{\includegraphics[width=0.54\linewidth]{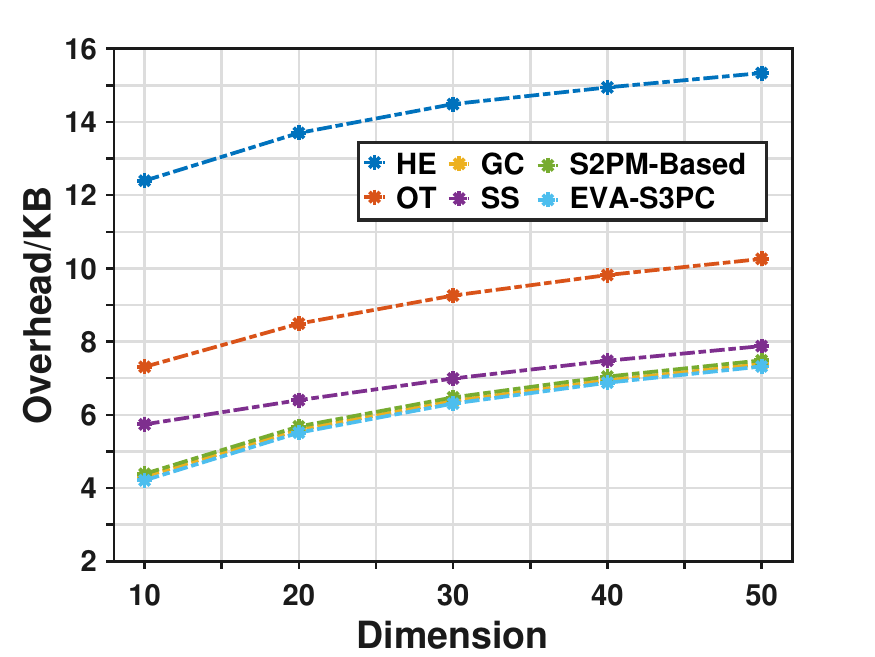}
    \label{fig: S3PHM Comm.Overhead}}
    \hspace{-0.7cm}
    \caption{Communication overhead (/KB) for classical SMPC frameworks in S2PM(a), S3PM(b), S2PI(c), and S3PHM(d) with matrix dimensions $N$ from 10 to 50 and numerical distribution $\delta$ in $[E-4, E+4]$ in $\mathbb{R}$. The HE method is excluded from S2PI due to lack of support for matrix inversion. As S2PHM involves two executions of S2PM, its communication overhead is similar and therefore not shown.}
    \label{fig: Overhead of Different SMPC Framework}
    \Description[<short description>]{<long description>}
\end{figure}

\begin{figure}[!htp]
    \centering
    \hspace{-0.5cm}
    \subfigure[S2PM Stage Time]{\includegraphics[width=0.54\linewidth]{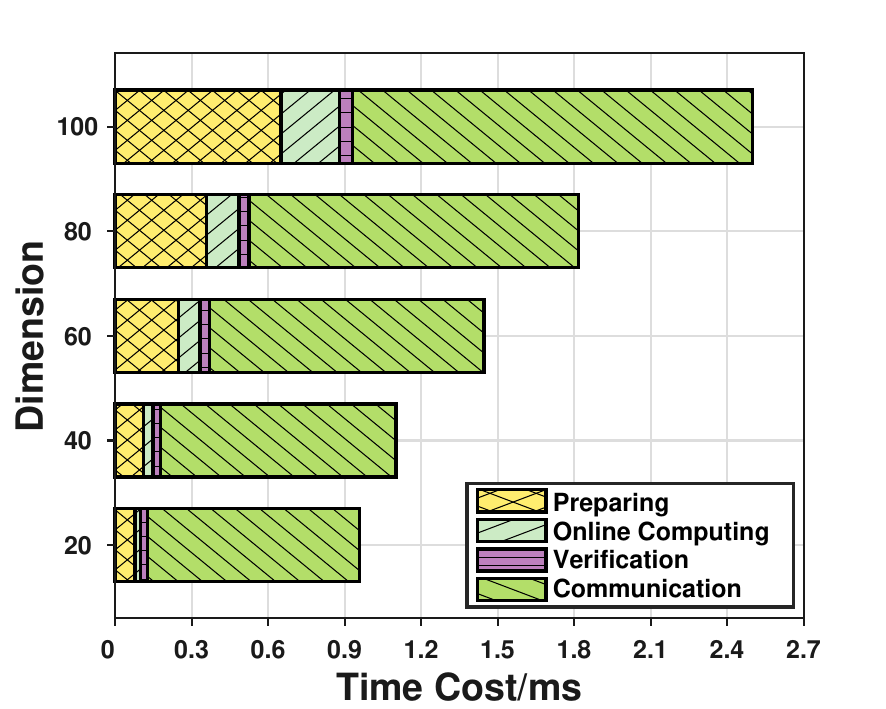}
    \label{fig: Allocation of S2PM}}
    \hspace{-0.6cm}
    \subfigure[S3PM Stage Time]{\includegraphics[width=0.54\linewidth]{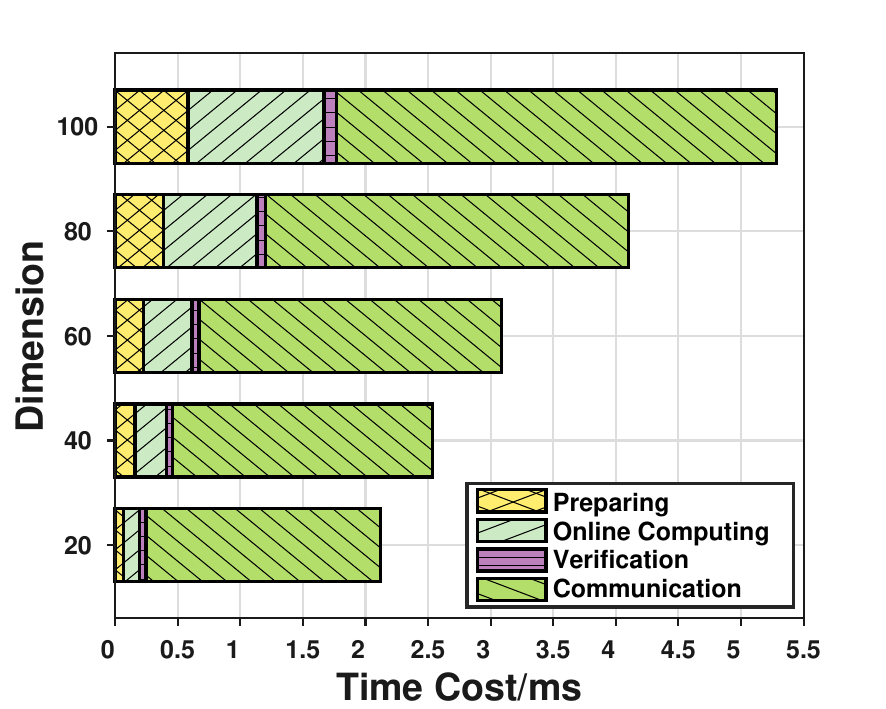}
    \label{fig: Allocation of S3PM}}
    \hspace{-0.7cm}
    \caption{Runtime allocation of different stage in S2PM (a) and S3PM (b), including preparation, online computation, verification, and communication costs, tested with parameters $N=10, 20, 30, 40, 50$.}
    \label{fig: Allocation of S2PM/S3PM}
    \Description[<short description>]{<long description>}
\end{figure}
\noindent\textbf{Protocol stage analysis and big matrix performance of EVA-S3PC.}
\autoref{fig: Allocation of S2PM/S3PM} illustrates the time consumed by computation in different stage and overall communication of S2PM and S3PM with various matrix dimensions. As matrix size doubles, the proportion of time spent in the Preparing and Online phases increases by $8\times$ to $10\times$, while communication overhead grows at a steadier rate of $1.8\times$ to $1.9\times$. This suggests that our protocol effectively controls communication overhead that is the common bottleneck for other schemes, making it a superior option to deal with large matrices. The growth in Preparing and Online Computing phases also suggests potential optimization opportunities for these stages in future work. 

\noindent \autoref{fig: Overhead and Big Matrix}(a) uses the more complex S3PHM protocol to show that as matrix scales from $N=10$ to $N=2000$, with verification rounds of $L=20, 50, 80$, the proportion of verification overhead gradually decreases. Notably, with $L=20$, the probability of a check failure is $Prob_f\leq \frac{1}{2^{200}} \approx 6.22\times 10^{-61}$, and the verification overhead ratio drops below 10\%. This demonstrates the lightweight nature of the verification. In practice, $L=20$ is sufficient to balance verification strength and its overhead for most applications. 
\autoref{fig: Overhead and Big Matrix}(b) shows that for large matrix with dimension between $1000\sim5000$, the running time of the five protocols using EVA-S3PC is well under 10s and demonstrates nearly linear growth. This indicates that EVA-S3PC scales efficiently for large matrix. Also the effective combination of S2PM and S3PM further reduces unnecessary communication, significantly improving overall efficiency.
\begin{figure}[!htb]
    \centering
    \hspace{-0.1cm}
   \subfigure[Proportion of Verification Overhead]{\includegraphics[width=0.49\linewidth]{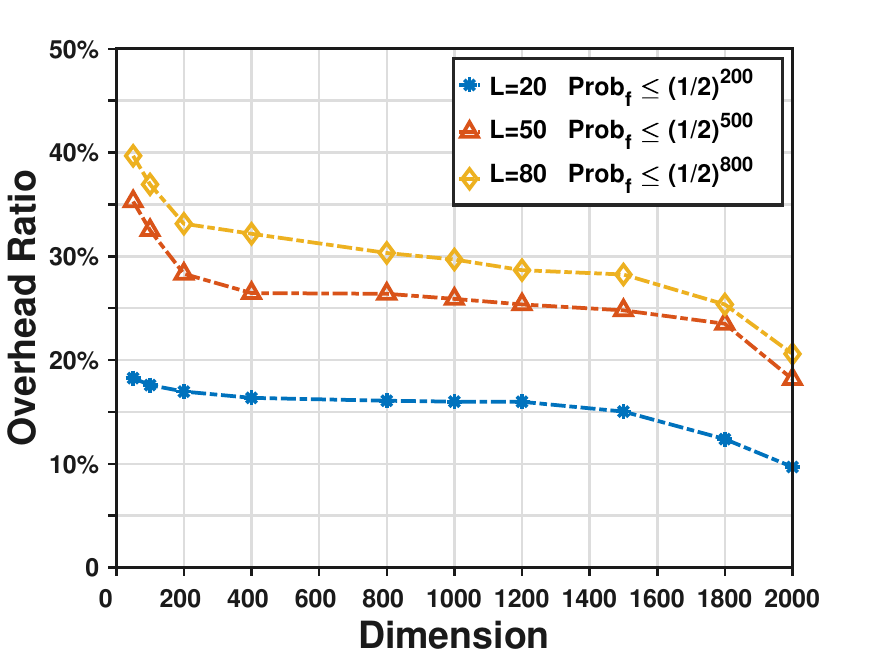}
   \label{fig: Overhead Verify}}
   \hspace{-0.1cm}
   \subfigure[EVA-S3PC for Big Matrix]{\includegraphics[width=0.49\linewidth]{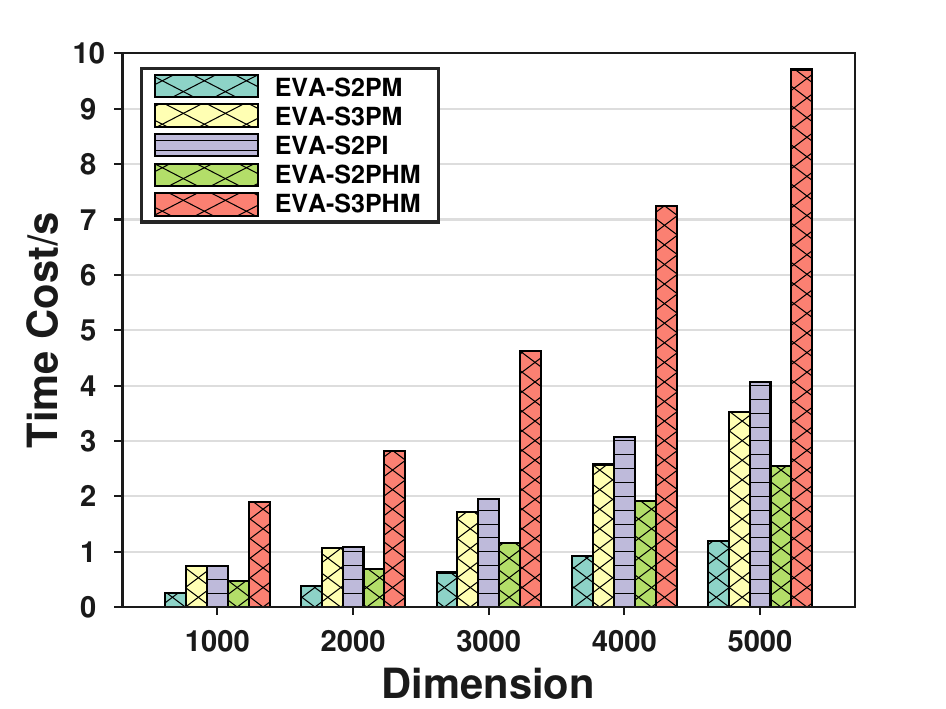}
   \label{fig: Big Matrix Time}}
    \hspace{-0.3cm}
    \caption{(a) shows the verification overhead ratio with round L=20,50,80, and dimension N=10$\sim$2000 .(b) shows running time cost of sub-protocols for big matrix in dimension N = 1000,2000,3000,4000,5000, which reported in seconds.}  
    \label{fig: Overhead and Big Matrix}
    \Description[<short description>]{<long description>}
\end{figure}
\subsubsection{Precision across Dynamic Ranges}
\noindent Accurate computation is crucial for the practical application of secure computing frameworks. In this section, we first test the maximum relative error (MRE) of the five basic protocols in EVA-S3PC under different matrix dimensions and precision settings and summarized the result in \autoref{tab:Precision of EVA-S3PC framework} and \autoref{tab:Precision of Different S3PC Framework}. We then compare EVA-S3PC against SS, GC, OT, HE, and S2PM based schemes, which is summarized in \autoref{tab:SMPC-significants}.\\
\noindent\textbf{Precision of EVA-S3PC.}
\autoref{tab:Precision of EVA-S3PC framework}(a)$\sim$(e) shows the trend of MRE for each protocol in EVA-S3PC framework when the dynamic range $\delta$ of the input matrices varies. Some interesting patterns are observed: (1) With the same dimension, MRE increases as the dynamic range of input data increases; (2) For S2PM, S3PM, S2PHM, and S3PHM, larger matrix dimension (from $N=10$ to $N=50$) leads to lower MRE; (3) S2PI exhibits the opposite trend higher MRE under larger dimension.

\noindent To explain phenomena (1) and (2), we refer to the error decomposition of matrix computation \cite{chung2006computer}: $\sigma_{total}=\sigma_{roundoff}+\sigma_{arithmetic}$. For phenomenon (1), as the dynamic range of matrices $M_A$ and $M_B$ increases to the range of $[E-\delta, E+\delta]$, the multiplication $M_A \times M_B = \sum_{k=1}^{N} a_{ik} \times b_{kj}$ is more likely to include numbers with larger variation, leading to more round-off and accumulation errors with increasing MRE. Conversely, for smaller dynamic ranges (e.g., $\delta=0$), the numbers are less varying in scale, reducing such errors and with smaller MRE.
For phenomenon (2), since S2PM, S3PM, S2PHM, and S3PHM are based on bilinear matrix multiplication over $\mathbb{R}$, research \cite{el2002inversion,fasi2023matrix}  shows that average errors in such operations converge to a limit near floating-point precision with unit in the last place (ULP) of $2^{-52}$ as dimensions increase. The accumulated MRE (AMRE) of a $N$-dimensional matrix filled with float-point (FLP) numbers is given by $AMRE(N, FLP(r,p,round)) = N^2 \times ARRE(FLP(r,p,round)) = N^2 \times \frac{(r^2-1)}{4r \ln r} \times ULP \approx 7.33 \times 10^{-14}$ (with $N=50$ and $r=2$), which is consistent to the MRE scale in \autoref{tab:Precision of EVA-S3PC framework}(a,b,d,e). This suggests that as matrix dimensions increase, positive and negative errors tend to cancel out, resulting in a regression to the mean relative error.
For phenomenon (3), since matrix inversion is a nonlinear operation, the MRE in S2PI does not follow the previous error convergence pattern but instead depends on the condition number of the input matrix. As matrix dimensions increase, control of the condition number becomes more difficult, leading to a higher likelihood of ill-conditioned matrices and thus larger inversion errors.

\begin{figure}[htbp]
    \center
    \subfigure[MRE of EVA-S2PM]{
    \begin{minipage}[c]{3.98cm} 
    \centering
    \includegraphics[width=1.0\linewidth]{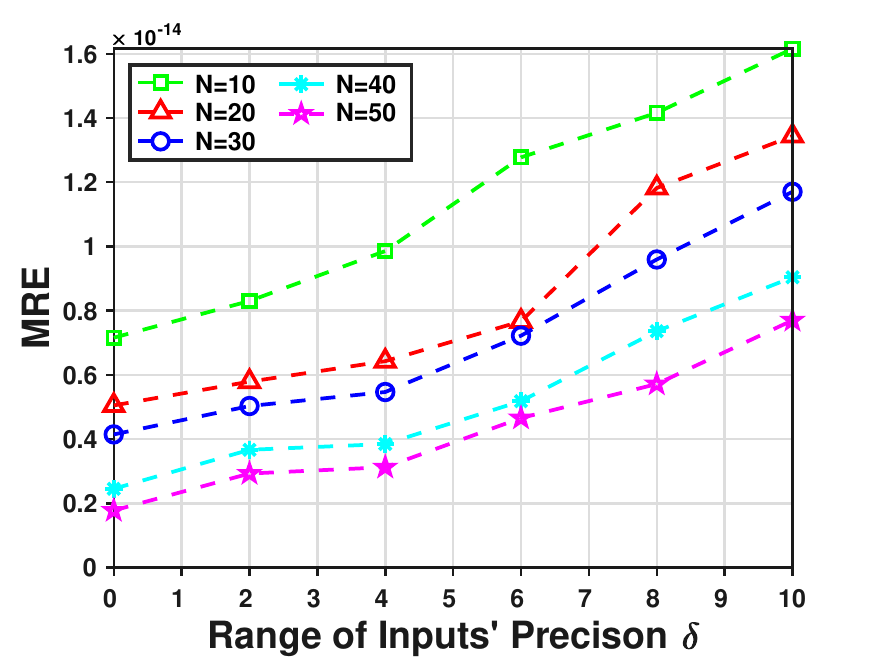}
    \end{minipage}}
    \subfigure[MRE of EVA-S3PM]{
    \begin{minipage}[c]{3.98cm} 
    \centering
    \includegraphics[width=1.0\linewidth]{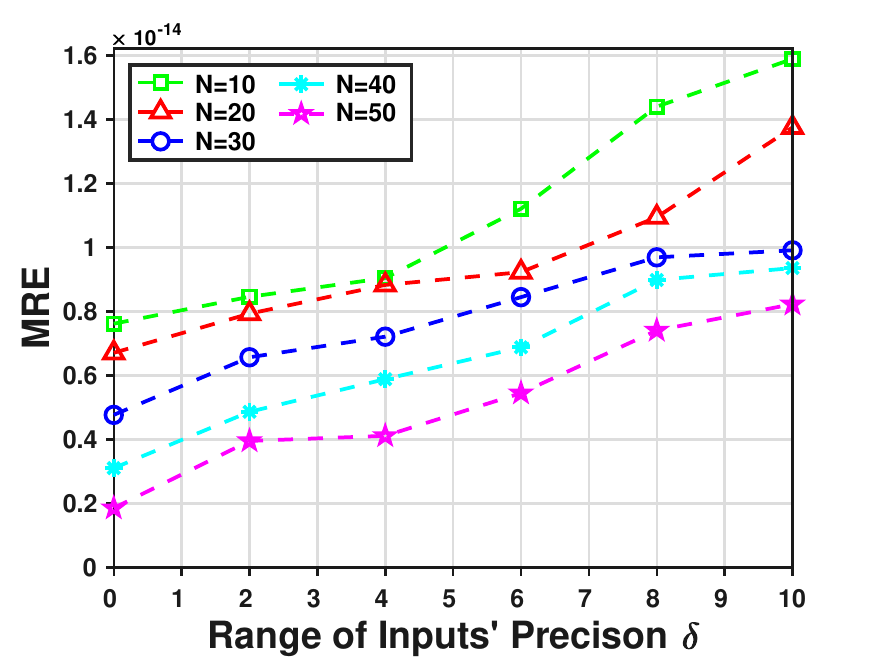}
    \end{minipage}}       
    \center
    \subfigure[MRE of EVA-S2PI]{
    \begin{minipage}[c]{3.98cm} 
    \centering
    \includegraphics[width=1.0\linewidth]{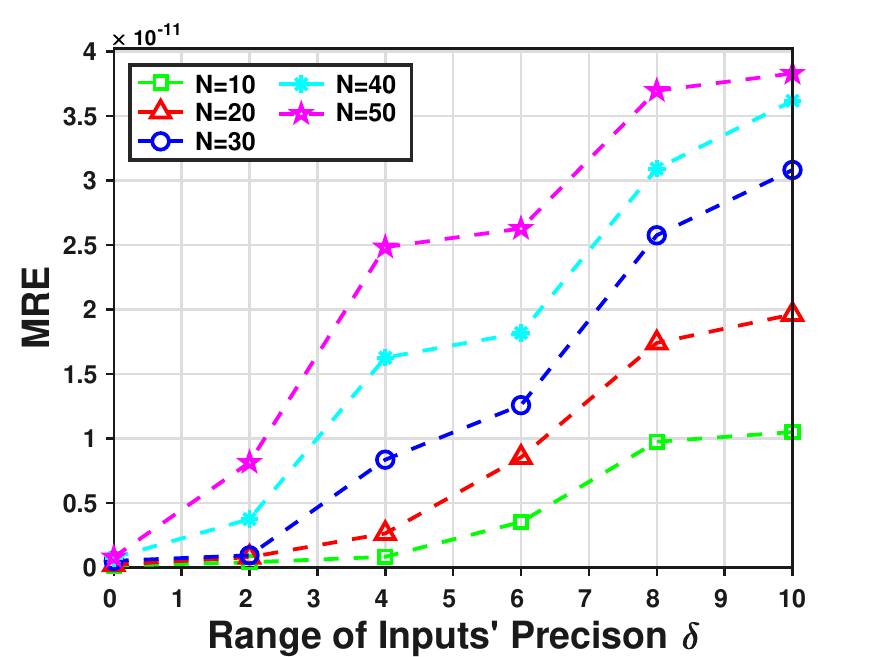}
    \end{minipage}}
    \subfigure[MRE of EVA-S2PHM]{
    \begin{minipage}[c]{3.98cm} 
    \centering
    \includegraphics[width=1.0\linewidth]{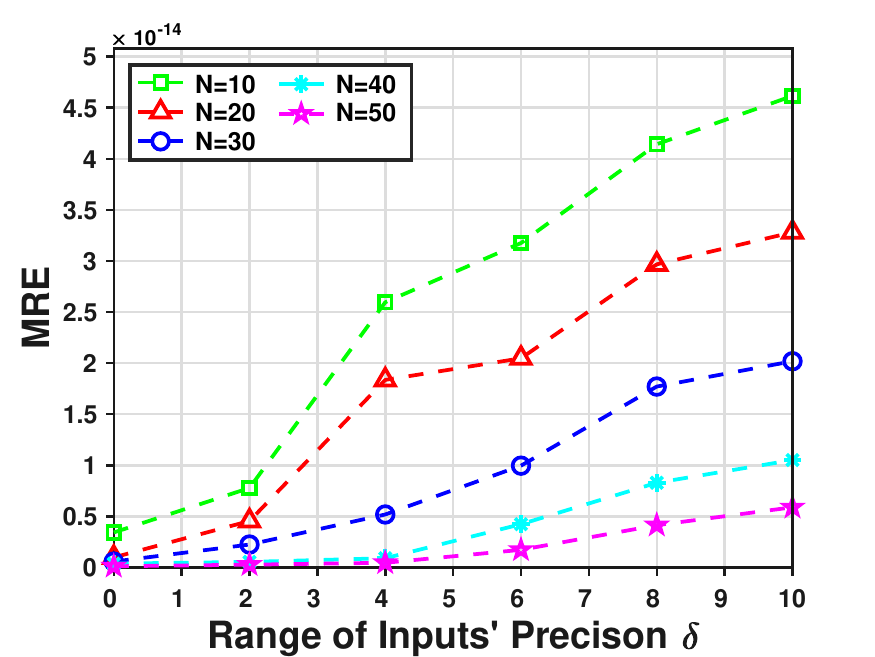}
    \end{minipage}}
    \center
    \subfigure[MRE of EVA-S3PHM]{
    \begin{minipage}[c]{3.98cm} 
    \centering
    \includegraphics[width=1.0\linewidth]{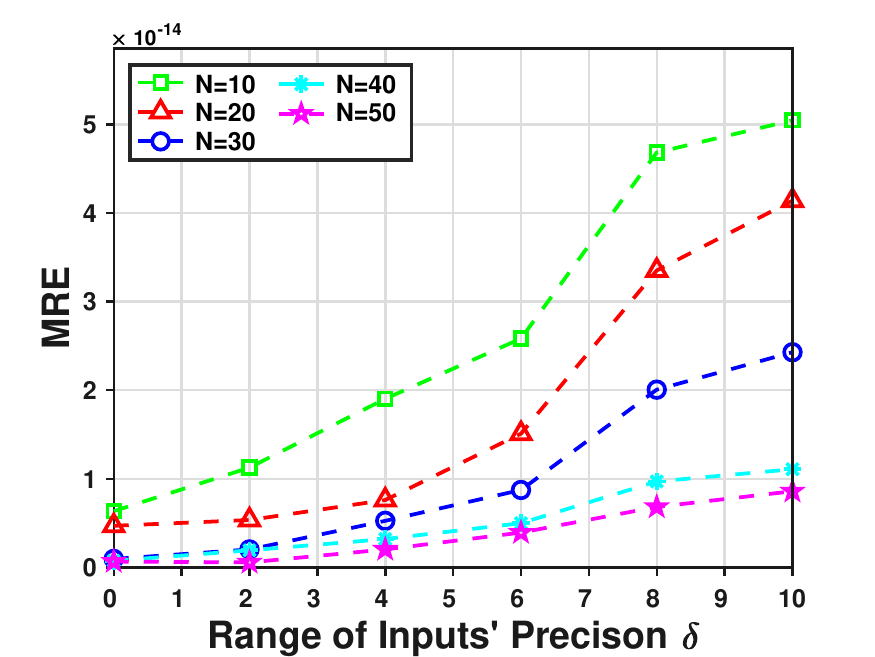}
    \end{minipage}}
    \caption{MRE of EVA-S3PC basic protocols. In (a)-(e), 1000 sets of random input matrices in the dynamic range of $\{\delta \mapsto X_{ij} \in [E-\delta, E+\delta]\}$ were sampled and the MREs of each protocol are plotted at different matrix dimensions of $N=10, 20, 30, 40, 50$.}
    \label{tab:Precision of EVA-S3PC framework}%
    \Description[<short description>]{<long description>}
\end{figure}

\begin{figure}[htbp]
    \center
    \subfigure[NMRE of S2PM]{
    \begin{minipage}[c]{3.98cm} 
    \centering
    \includegraphics[width=1.0\linewidth]{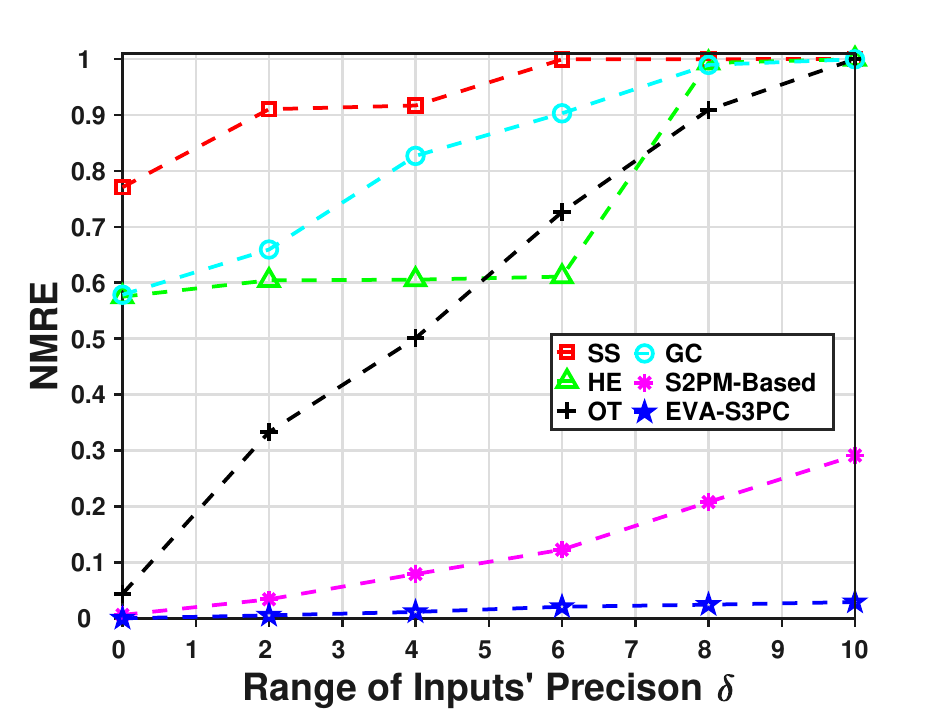}
    \end{minipage}}
    \subfigure[NMRE of S3PM]{
    \begin{minipage}[c]{3.98cm} 
    \centering
    \includegraphics[width=1.0\linewidth]{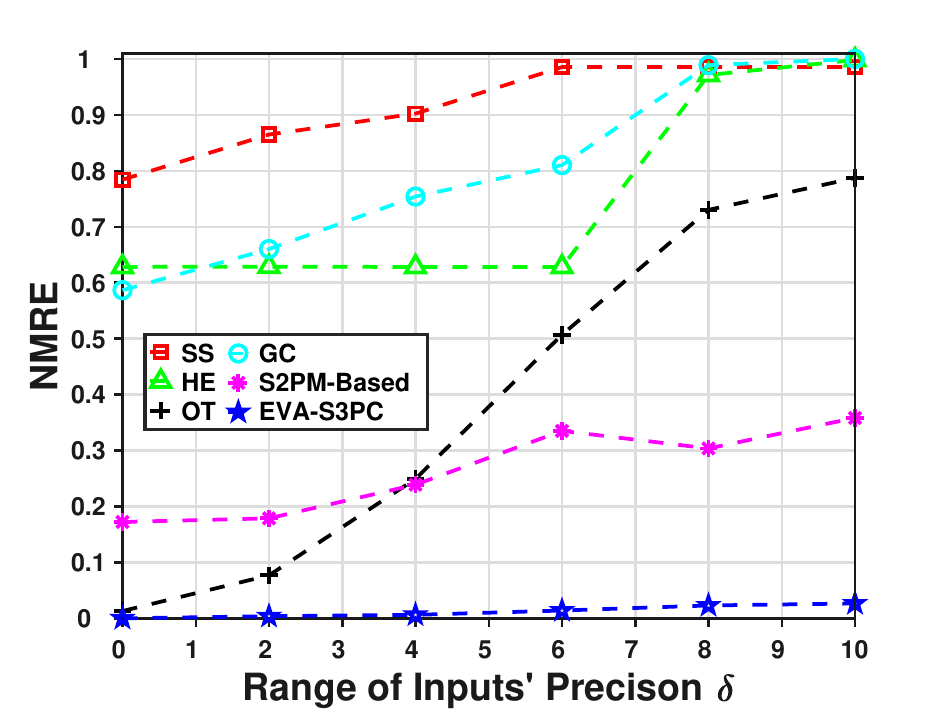}
    \end{minipage}}   
    \center
    \subfigure[NMRE of S2PI]{
    \begin{minipage}[c]{3.98cm} 
    \centering
    \includegraphics[width=1.0\linewidth]{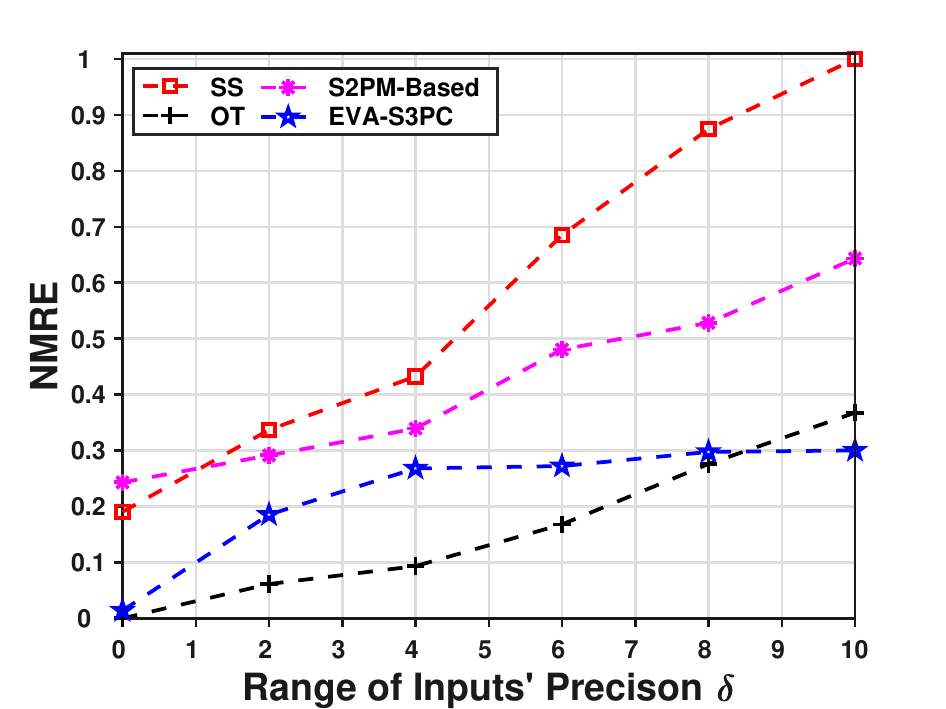}
    \end{minipage}}
    \subfigure[NMRE of S2PHM]{
    \begin{minipage}[c]{3.98cm} 
    \centering
    \includegraphics[width=1.0\linewidth]{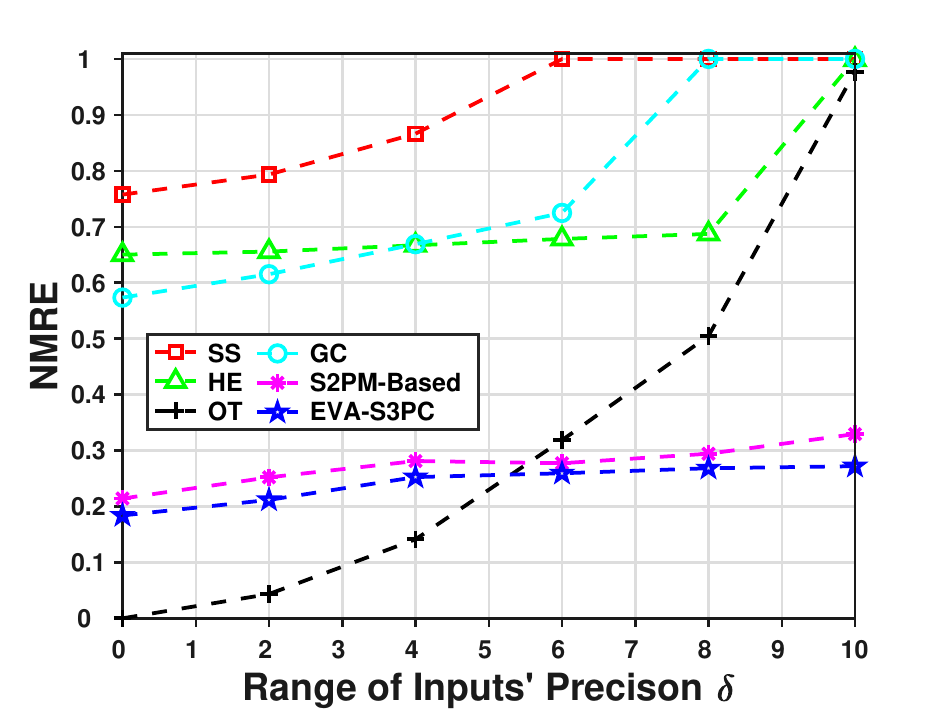}
    \end{minipage}}
    \center
    \subfigure[NMRE of S3PHM]{
    \begin{minipage}[c]{3.98cm} 
    \centering
    \includegraphics[width=1.0\linewidth]{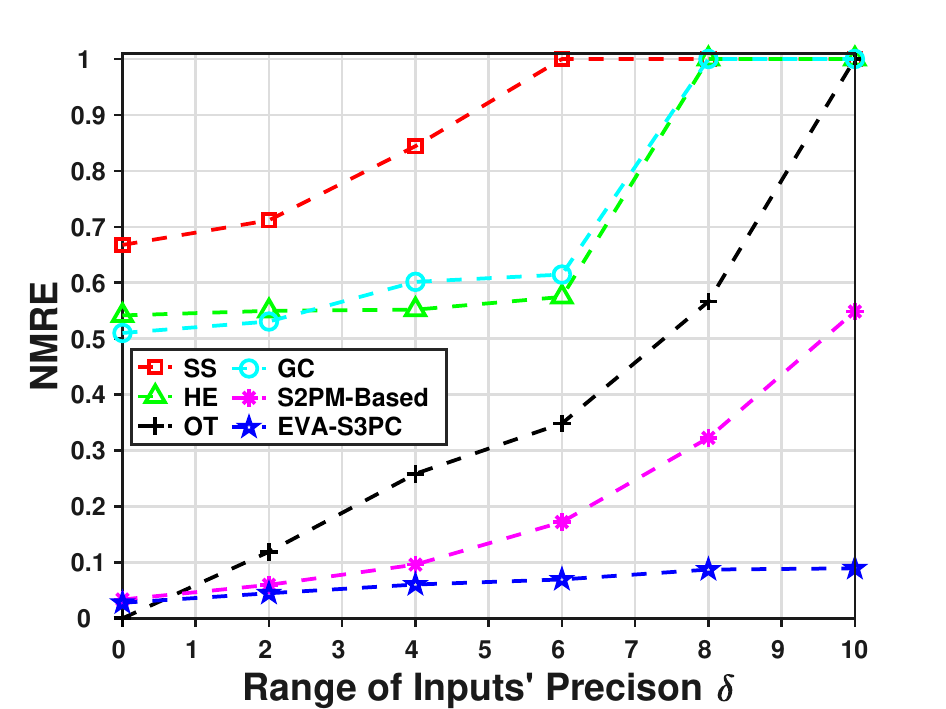}
    \end{minipage}}
    \caption{ Normalized MRE (NMRE) comparison of different frameworks. In (a)-(e), 1000 sets of random input matrices in the dynamic range of $\{\delta \mapsto X_{ij} \in [E-\delta, E+\delta]\}$ were sampled. S2PI is tested with $N=50$, while other protocols are tested with $N=10$. NMRE=1 indicates definite failure of computation. The maximum $\delta$ without failure for each framework shows the usable boundary $\delta_{max}$ without incurring error during computation. The S2PI tests using HE and GC are excluded as they cannot effectively carry out matrix inversion.}
    \label{tab:Precision of Different S3PC Framework}        \Description[<short description>]{<long description>}
\end{figure}

\noindent\textbf{Comparison to Prior Work.}
\noindent Because the difference in MREs among different methods are gigantic, they are normalized (NMRE) by taking the log and are then linearly rescaled to the range between 0 and 1 for ease of comparison. Because the expected MRE for a correct calculation is below $1E-3$, NMRE of 1 indicates the definite failure of computation due to the overstretched dynamic range. \autoref{tab:Precision of Different S3PC Framework}(a)$\sim$(e) presents the NMRE at different dynamic range, and the boundary with NMRE below 1 indicates the usable scope of different secure computing frameworks. The results show that EVA-S3PC outperforms other frameworks in all protocols, thanks to the real-number data disguising method, which avoids ciphertext-based computations, and eliminates errors from fixed-point truncation and ciphertext noise accumulation.

\noindent HE, operating in a finite field $\mathbb{Z_{2^l}}(l = 32)$, has promising precision when the dynamic range $\delta \leq 6$. However, when input variation exceed this threshold, error arises rapidly as that of S3PHM, in which consecutive multiplications accumulate too much noise to produce accurate result.
OT also shows acceptable precision due to its use of random masks for privacy without requiring numerical encoding or splitting. In contrast, GC suffers significant precision loss when converting arithmetic operations to Boolean circuits, especially in complex non-linear operations like S2PI, where multiple encryption/decryption steps across gates introduce cumulative errors.
SS demonstrates worse precision than other frameworks due to its fixed-point representation, which limits accuracy. Repeated matrix multiplications in shared and recovered data introduce rounding errors, particularly in non-linear computations, leading to substantial calculation bias.

\begin{table}[htbp]
  \centering  
  \caption{Comparison of usable dynamic range of Float64 calculation and maximum significant decimal digits. Each protocol was tested 1000 times using six frameworks under the same conditions to assess precision loss. Errors exceeding $1E-3$ were considered incorrect, and error rates were recorded. $\textbf{\ding{56}}$ indicates an error rate above 10\% so the framework is unusable, while $\textbf{\ding{52}}$ indicates a 0\% error rate.}
    \resizebox{\linewidth}{!}{
    \begin{tabular}{cccccccc}
    \toprule 
    \multicolumn{1}{c}{\multirow{2}[2]{*}{\textbf{Framework}}} & \multicolumn{5}{c}{\textbf{Error Rate of Basic Security Protocols}} & \multicolumn{1}{c}{\multirow{2}[2]{*}{\makecell[c]{\textbf{Range of} \\ \textbf{Precision}}}} & \multicolumn{1}{c}{\multirow{2}[2]{*}{\makecell[c]{\textbf{Maximum} \\ \textbf{Significant}}}} \\
\cmidrule{2-6}          & \textbf{S2PM} & \textbf{S3PM} & \textbf{S2PI} & \textbf{S2PHM} & \textbf{S3PHM} &       &  \\
    \midrule
    \textbf{SS} & \textbf{\ding{52}} & \color{black}3.54\% & \color{black}5.46\% & \color{black}2.72\% & \color{black}3.26\% & E-04\textasciitilde E+04 & 4\textasciitilde 5 \\
    \textbf{GC} & \textbf{\ding{52}} & \color{black}1.07\% & \textbf{\ding{56}} & \color{black}1.66\% & \color{black}1.97\% & E-06\textasciitilde E+06 & 7 \textasciitilde 8 \\
    \textbf{OT} & \textbf{\ding{52}} & \textbf{\ding{52}} & \textbf{\ding{52}} & \textbf{\ding{52}} & \textbf{\ding{52}} & E-10\textasciitilde E+10 & 10\textasciitilde 11 \\
    \textbf{HE} & \textbf{\ding{52}} & \textbf{\ding{52}} & \textbf{\ding{56}} & \textbf{\ding{52}} & \textbf{\ding{52}} & E-06\textasciitilde E+06 & 7\textasciitilde 8 
    \\
    \rowcolor{gray!15}
    \textbf{S2PM-Based} & \textbf{\ding{52}} & \textbf{\ding{52}} & \textbf{\ding{52}} & \textbf{\ding{52}} & \textbf{\ding{52}} & E-10\textasciitilde E+10 & 10\textasciitilde 11 \\        \rowcolor{gray!30}
    \textbf{EVA-S3PC} & \textbf{\ding{52}} & \textbf{\ding{52}} & \textbf{\ding{52}} & \textbf{\ding{52}} & \textbf{\ding{52}} & E-10\textasciitilde E+10 & 11\textasciitilde 12 \\
    \bottomrule
    \end{tabular}}
  \label{tab:SMPC-significants}
\end{table}%
\noindent According to numerical analysis studies \cite{sauer2011numerical}, the MRE and significant digits satisfy the relation $s = 1 - \log_{10}(2(a_1+1) \times MRE)$. Using the measured MREs in the experiments, we can estimate the maximum significant digits supported by each framework in Float64 computations. In \autoref{tab:Precision of Different S3PC Framework}(a)$\sim$(e), the usable dynamic range can be derived from the maximum $\delta$ with $NMRE<1$. We conducted extensive error tests for each framework under the same settings as in $\mathsection{\ref{settings}}$ and summarized the precision  in \autoref{tab:SMPC-significants}. EVA-S3PC consistently shows good precision for all protocols with a Float64 usable dynamic range of [E-10, E+10] and can produce 11 to 12 significant digits.

\subsection{Results for S3PLR}
\noindent The sub-protocols of basic operators can be assembled to form models with more complexity. We use a vertically partitioned 3-party linear regression model trained with least square method as an example and implement the model with various frameworks. The efficiency and accuracy are compared for both training and inference.

\subsubsection{Datasets and Benchmarks.}
\noindent To evaluate the performance of EVA-S3PC framework in linear regression, we use two standard benchmarking datasets from Scikit-learn: the \textbf{Boston dataset} with 506 samples (404 for training, 102 for testing), 13 features and 1 label, and the \textbf{Diabetes dataset} with 404 samples (353 for training, 89 for testing), 10 features and 1 label. In the 3-party regression, based on the design of S3PLR in section $\mathsection{\ref{S3PLR}}$, the label is private and only accessible to Carol, and the features are evenly split between Alice and Bob. The least square linear regression is repeated 100 times and the efficiency and accuracy are benchmarked using SecretFlow\cite{ma2023secretflow}, CryptGPU, LibOTe, FATE, S2PM-based methods, and EVA-S3PC on both datasets.
\subsubsection{Metrics of Comparison}
\noindent As a common preprocessing of data, we normalized the features to have standard Normal distribution. Efficiency metrics include time of training and inference, as well as communication overhead in both size and round of secure protocols. Accuracy comparison metrics include mean absolute error (MAE), mean square error (MSE) and root mean square error (RMSE) between prediction and label, L2-Norm relative error (LNRE) between securely trained model parameters and the ones learned from the plain text data using Scikit-learn, R-Square, and R-Square relative error (RRS) between privacy preserving models and plain text models. Particularly, LNRE measures the relative error between securely trained parameters $\beta$ and the ones learned from plaintext using Scikit-learn $\widehat{\beta}$: $LNRE = \frac{||\beta - \widehat{\beta}||_2}{||\beta||_2}$. RRS quantifies the relative difference of the R-Square values between secure and plaintext models: $RRS = \frac{|R^2 - \widehat{R}^2|}{\widehat{R}^2} \times 100\%$.
\begin{table}[htp]
  \centering
  \caption{Training and inference time comparison of various protocols in the LAN setting. The time obtained from testing on plaintext using Scikit-learn serves as the baseline for the training and inference phases. All communication is reported in MB, and time costs are reported in ms.}
    \resizebox{\linewidth}{!}{
    \begin{tabular}{ccccccccc}
    \toprule  
    \multicolumn{1}{c}{\multirow{2}[2]{*}{\textbf{Dataset}}} & \multicolumn{1}{c}{\multirow{2}[2]{*}{\textbf{Data State}}} & \multicolumn{1}{c}{\multirow{2}[2]{*}{\textbf{ Framework}}} & \multicolumn{2}{c}{\textbf{ Training Overhead }} & \multicolumn{1}{c}{\multirow{2}[1]{*}{\makecell[c]{\textbf{Training} \\ \textbf{Time/ms}}}} & \multicolumn{2}{c}{\textbf{ Inference Overhead }} &  \multicolumn{1}{c}{\multirow{2}[1]{*}{\makecell[c]{\textbf{ Inference} \\ \textbf{Time/ms}}}} \\
\cmidrule{4-5}\cmidrule{7-8}          &       &       & \multicolumn{1}{c}{\textbf{Com. /MB}} & \multicolumn{1}{c}{\textbf{Rounds}} &       & \multicolumn{1}{c}{\textbf{Com. /MB}} & \multicolumn{1}{c}{\textbf{Rounds}} &  \\
    \midrule
    \multirow{7}[4]{*}{\textbf{Diabetes}} & \multirow{6}[2]{*}{\textbf{Cipher}} & \textbf{SecretFlow} & 1.44  & 70    & 4741.83  & 1.12  & 54    & 1188.36  \\
          &       & \textbf{CryptGPU} & 2.64  & 133   & 53.47  & 0.07  & 9     & 23.70  \\
          &       & \textbf{Fate} & 5303.49  & 24    & 36920.94  & 300.52  & 6     & 2092.13  \\
          &       & \textbf{LibOTe} & 16.85  & 3424  & 33399.78  & 3.11  & 1800  & 18606.95  \\
          &       & \cellcolor{gray!30}\textbf{S2PM-Based} & \cellcolor{gray!30}1.24  & \cellcolor{gray!30}85    &\cellcolor{gray!30} 10.49  & \cellcolor{gray!30}0.09  & \cellcolor{gray!30}27    & \cellcolor{gray!30}0.37  \\
         &       & \cellcolor{gray!30}\textbf{EVA-S3PC} & \cellcolor{gray!30}\textbf{0.68 } & \cellcolor{gray!30}\textbf{60 } & \cellcolor{gray!30}\textbf{4.17 } & \cellcolor{gray!30}\textbf{0.09 } & \cellcolor{gray!30}\textbf{25 } & \cellcolor{gray!30}\textbf{0.32 } \\
\cmidrule{2-9}          & \textbf{Plain} & \textbf{Scikit-learn} & \textbf{/}     & \textbf{/}     & \textbf{0.58 } & \textbf{/}     & \textbf{/}     & \textbf{0.05 } \\
    \midrule
    \multirow{7}[4]{*}{\textbf{Boston}} & \multirow{6}[2]{*}{\textbf{Cipher}} & \textbf{SecretFlow} & 1.49  & 70    & 3265.70  & 1.03  & 54    & 1235.68  \\
          &       & \textbf{CryptGPU} & 3.70  & 133   & 50.66  & 0.08  & 9     & 15.77  \\
          &       & \textbf{Fate} & 8911.40  & 24    & 62037.90  & 285.24  & 6     & 1985.76  \\
          &       & \textbf{LibOTe} & 24.48  & 3662  & 34774.30  & 4.38  & 1800  & 20649.50  \\
          &       & \cellcolor{gray!30}\textbf{S2PM-Based} & \cellcolor{gray!30}1.79  & \cellcolor{gray!30}85    & \cellcolor{gray!30}13.28  & \cellcolor{gray!30}0.12  & \cellcolor{gray!30}27    & \cellcolor{gray!30}0.64  \\
           &       & \cellcolor{gray!30}\textbf{EVA-S3PC} & \cellcolor{gray!30}\textbf{0.99 } & \cellcolor{gray!30}\textbf{60 } & \cellcolor{gray!30}\textbf{8.18 } & \cellcolor{gray!30}\textbf{0.12 } & \cellcolor{gray!30}\textbf{25 } & \cellcolor{gray!30}\textbf{0.34 } \\
\cmidrule{2-9}          & \textbf{Plain} & \textbf{Scikit-learn} & \textbf{/}     & \textbf{/}     & \textbf{0.36 } & \textbf{/}     & \textbf{/}     & \textbf{0.05 } \\
    \bottomrule
    \end{tabular}}%
  \label{tab:S3PLR-Efficient}%
\end{table}%

\subsubsection{Efficiency and Accuracy of EVA-S3PC}
\noindent \autoref{tab:S3PLR-Efficient} compares EVA-S3PC framework to other PPML schemes. For the training on Diabetes dataset, S3PLR by EVA-S3PC reduces communication by $25.7\%\sim54.8\%$ compared to the fastest state of art methods (SecretFlow, CryptGPU, S2PM-based method), and speeds up the training time by $2.5\times\sim12.8\times$. For the inference, our communication cost is similar to that of CryptGPU. The performance on the Boston dataset demonstrates the same trend.
\autoref{tab:S3PLR-Accuracy} summaries the accuracy metrics for various frameworks. EVA-S3PC shows the closest parameters and R-Square to those derived from plain text, indicated by the smallest LNRE and RRS.
\noindent Particularly, SecretFlow (based on GC and SS) has low communication overhead ($1.03\sim1.49 MB$) and high accuracy but suffers from large training and inference times due to the need to encrypt and decrypt Boolean circuits for matrix operations.
CryptGPU (based on ASS and BSS) benefits from GPU acceleration, offering high computational efficiency, but its use of fixed-point representation ($Z_{2^{16}}$) reduces precision, which is reflected by its lower model accuracy (RRS at 7.32\%).
FATE introduces HE for secure linear operations, leading to the lowest number of communication rounds but significantly higher communication size ($5.3GB\sim8.9GB$ during training). Also the noise accumulates between the calls of HE, which affects accuracy.
LibOTe, using Diffie-Hellman key exchange, ensures privacy and high precision (RRS = $2.06E-5$), but its nearly 3000 interaction rounds cause excessive communication overhead and reduces practicality.
Our S2PM-based S3PLR and EVA-S3PC implementation both use data disguising techniques and thus have similar communication complexity to plaintext method with constant communication rounds. However, S2PM-based S3PLR relies on repeated calls of S2PM for 3-party regression, leading to lower precision and efficiency compared to EVA-S3PC.
\section{Conclusion and future work}\label{Conclusion}
\noindent To address the need for sensitive data protection in distributed SMPC scenarios, we leverage data disguising techniques to introduce a suite of semi-honest, secure 2-party and 3-party atomic operators for large-scale matrix operations, including linear (S2PM, S3PM) and complex non-linear operations (S2PI, S2PHM, S3PHM). Through an asynchronous computation workflow and precision-aligned random disguising, our EVA-S3PC framework achieves substantial improvements in communication overhead and Float64 precision over existing methods. Our Monte Carlo-based anomaly detection module offers robust detection at $L=20$ with minimal overhead, achieving a probability threshold of $Prob_f \leq 6.22 \times 10^{-61}$, even for minor errors.The EVA-S3PC framework achieves near-baseline accuracy in secure three-party regression, with significant improvements in communication complexity and efficiency over mainstream frameworks in both training and inference. Performance tests reveal the framework’s strong potential for privacy-preserving multi-party computing applications in distributed settings.
Currently, this study provides exploratory results in a single-machine simulated LAN  environment to validate EVA-S3PC’s performance in communication, precision, and efficiency. \\
Future work will evaluate scalability and practicality in complex multi-machine setups under real LAN/WAN conditions. Although our protocols focus on 2-party and 3-party settings, we aim to extend these to n-party protocols with optimized time complexity. Lastly, integrating lightweight cryptographic primitives will address potential collusion risks with the third-party CS node, ensuring secure, efficient, and scalable performance for real-world applications, including healthcare and finance.
\begin{table}[!hp]
  \centering
  \caption{Summary of benchmarks involving inference accuracy of various frameworks using Boston and Diabetes  datasets in the LAN setting. The RRS and LNRE refer to the relative error of R-Square and L-2 Norm of Scikit-learn computing on plain.}
    \resizebox{\linewidth}{!}{
    \begin{tabular}{ccccccccc}
    \toprule
    \multicolumn{1}{c}{\multirow{2}[2]{*}{\textbf{Dataset}}} & \multicolumn{1}{c}{\multirow{2}[2]{*}{\textbf{Data State}}} & \multicolumn{1}{c}{\multirow{2}[2]{*}{\textbf{ Framework}}} & \multicolumn{6}{c}{\textbf{Accuracy Evaluation}} \\
\cmidrule{4-9}          &       &       & \multicolumn{1}{c}{\textbf{MAE}} & \multicolumn{1}{c}{\textbf{MSE}} & \multicolumn{1}{c}{\textbf{RMSE}} & \multicolumn{1}{c}{\textbf{LNRE}} & \multicolumn{1}{c}{\textbf{R-Square}} & \multicolumn{1}{c}{\textbf{RRS}} \\
    \midrule
    \multirow{7}[4]{*}{\textbf{Diabetes}} & \multirow{6}[2]{*}{\textbf{Cipher}} & \textbf{SecretFlow} & 44.285  & 2948.274  & 54.298  & 0.322  & 0.468  & 1.16E-04 \\
          &       & \textbf{CryptGPU} & 45.389  & 3137.410  & 56.013  & 0.332  & 0.433  & 7.32E-02 \\
          &       & \textbf{Fate} & 46.127  & 2971.215  & 54.509  & 0.323  & 0.463  & 8.98E-03 \\
          &       & \textbf{LibOTe} & 44.285  & 2948.274  & 54.298  & 0.322  & 0.468  & 2.06E-05 \\
          &       & \cellcolor{gray!30}\textbf{S2PM-Based} & \cellcolor{gray!30} 44.333  & \cellcolor{gray!30} 2953.713  & \cellcolor{gray!30} 54.348  & \cellcolor{gray!30} 0.322  & \cellcolor{gray!30} 0.467  & \cellcolor{gray!30} 2.22E-03 \\
          &       & \cellcolor{gray!30} \textbf{EVA-S3PC} & \cellcolor{gray!30} \textbf{44.285 } & \cellcolor{gray!30} \textbf{2948.274 } & \cellcolor{gray!30} \textbf{54.298 } & \cellcolor{gray!30} \textbf{0.322 } & \cellcolor{gray!30} \textbf{0.468 } & \cellcolor{gray!30} \textbf{1.88E-06} \\
\cmidrule{2-9}          & \textbf{Plain} & \textbf{Scikit-learn} & \textbf{44.285 } & \textbf{2947.974 } & \textbf{54.295 } & \textbf{0.322 } & \textbf{0.468 } & \textbf{/} \\
    \midrule
    \multirow{7}[4]{*}{\textbf{Boston}} & \multirow{6}[2]{*}{\textbf{Cipher}} & \textbf{SecretFlow} & 3.534  & 24.152  & 4.914  & 0.195  & 0.764  & 3.62E-04 \\
          &       & \textbf{CryptGPU} & 3.577  & 25.343  & 5.034  & 0.200  & 0.753  & 1.52E-02 \\
          &       & \textbf{Fate} & 3.579  & 25.433  & 5.043  & 0.201  & 0.752  & 1.63E-02 \\
          &       & \textbf{LibOTe} & 3.534  & 24.151  & 4.914  & 0.195  & 0.765  & 1.31E-05 \\
          &       & \cellcolor{gray!30} \textbf{S2PM-Based} & \cellcolor{gray!30} 3.544  & \cellcolor{gray!30} 24.241  & \cellcolor{gray!30} 4.924  & \cellcolor{gray!30} 0.196  & \cellcolor{gray!30} 0.764  & \cellcolor{gray!30} 5.23E-04 \\ 
          &       & \cellcolor{gray!30} \textbf{EVA-S3PC} & \cellcolor{gray!30} \textbf{3.534 } & \cellcolor{gray!30} \textbf{24.151 } & \cellcolor{gray!30} \textbf{4.914 } & \cellcolor{gray!30} \textbf{0.195 } & \cellcolor{gray!30} \textbf{0.765 } & \cellcolor{gray!30} \textbf{2.05E-06} \\ 
\cmidrule{2-9}          & \textbf{Plain} & \textbf{Scikit-learn} & \textbf{3.534 } & \textbf{24.151 } & \textbf{4.914 } & \textbf{0.195 } & \textbf{0.765 } & \textbf{/} \\
    \bottomrule
    \end{tabular}}%
  \label{tab:S3PLR-Accuracy}%
\end{table}%

\section{Acknowledgments}
This work was supported by the National Science and Technology Major Project of the Ministry of Science and Technology of China Grant No.2022XAGG0148.

\section{Appendix}
\textbf{Data Disguising Technique}
A secure computing protocol is generally composed of a sequence of steps in which one step takes the input from the previous one and produce its output that will be fed to the next step. Each step may happen on different participant, and it is essential to make sure that no participant is able to infer the original data from the collection of its input. Take a simple step of $a \times b$ between Alice and Bob as an example. If this intermediate result is kept by Alice or Bob, it can immediately infer the other party's data. Therefore, we adopt a data disguising technique to safeguard the intermediate output at each step to avoid such side effect.

We use a two-party scenario as an example in \autoref{fig:DataDisguising} and it is easy to generalize it to more parties. In a sequence of computations steps indexed by $k$, $S_k= F_k(A_k,B_k)$, where $A_k$ and $B_k$ are the input from Alice and Bob and $S_k$ is the intermediate output that can not be revealed to either party. Instead of directly output $S_k$, the data disguising technique produces two random variables $A_{k+1}$ and $B_{k+1}$ for Alice and Bob respectively, such that $A_{k+1} + B_{k+1} = S_k$. $A_{k+1}$ and $B_{k+1}$ are taken as input for Alice and Bob respectively in the next step $F_{k+1}(A_{k+1},B_{k+1})$, which further produces $A_{k+2}+B_{k+2}=S_{k+1}$, and so on. We express such computation as:
\begin{align}
A_k:B_k \rightarrow [A_{k+1}:B_{k+1} | A_{k+1} + B_{k+1} = F_k(A_k,B_k) ].
\end{align}
Because Alice and Bob do not share the random matrices $A_{k+1}$ and $B_{k+1}$ to each other, it is guaranteed that the intermediate result $S_k$ is secure against both of them. Figure 3 demonstrates the sequence of computation steps under the data disguising technique. As a matter of fact, in the following introduction of various protocols, we assume that the intermediate outputs is never explicitly computed. Instead, they are disguised as random matrices distributed across parties. The only exception is the final step of the protocol in which the request client sums up the disguising matrices and obtain the final result.
\begin{figure}[ht]
  \centering
  \includegraphics[width=1.0\hsize]{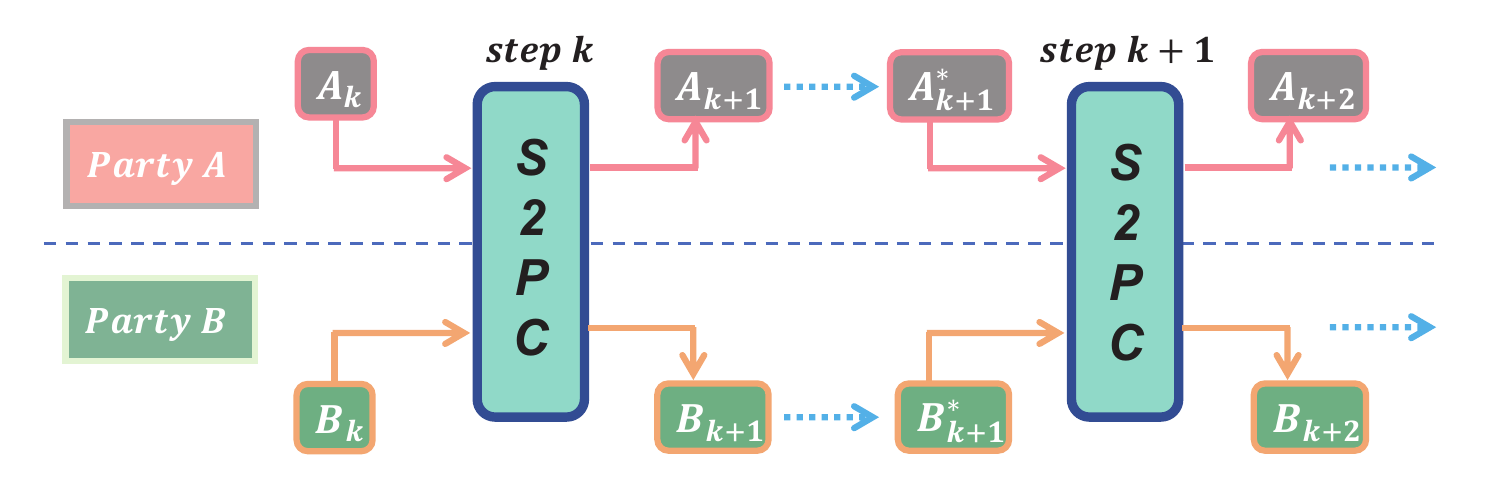}
  \vspace{-0.6cm}
  \caption{Data Disguising Strategy. The intermediate output from each step is disguised as two random matrices taken as the input in the next step.}
  \label{fig:DataDisguising}
  \Description[<short description>]{<long description>}
\end{figure}
\balance

\bibliographystyle{unsrt}
\bibliography{reference}

\begin{thebibliography}{10}

\bibitem{6424959}
Guilherme Galante and Luis Carlos~E. de~Bona.
\newblock A survey on cloud computing elasticity.
\newblock In {\em 2012 IEEE Fifth International Conference on Utility and Cloud Computing}, pages 263--270, 2012.

\bibitem{10.1145/3158363}
Zihao Shan, Kui Ren, Marina Blanton, and Cong Wang.
\newblock Practical secure computation outsourcing: A survey.
\newblock {\em ACM Comput. Surv.}, 51(2), February 2018.

\bibitem{sekar2023deep}
Jeyasri Sekar.
\newblock Deep learning as a service (dlaas) in cloud computing: Performance and scalability analysis.
\newblock {\em Journal of Emerging Technologies and Innovative Research}, 10:I541--I551, 2023.

\bibitem{voigt2017eu}
Paul Voigt and Axel Von~dem Bussche.
\newblock The eu general data protection regulation (gdpr).
\newblock {\em A Practical Guide, 1st Ed., Cham: Springer International Publishing}, 10(3152676):10--5555, 2017.

\bibitem{zhou2024secure}
Ian Zhou, Farzad Tofigh, Massimo Piccardi, Mehran Abolhasan, Daniel Franklin, and Justin Lipman.
\newblock Secure multi-party computation for machine learning: A survey.
\newblock {\em IEEE Access}, 2024.

\bibitem{mohanta2020multi}
Bhabendu~Kumar Mohanta, Debasish Jena, and Srichandan Sobhanayak.
\newblock Multi-party computation review for secure data processing in iot-fog computing environment.
\newblock {\em International Journal of Security and Networks}, 15(3):164--174, 2020.

\bibitem{munjal2023systematic}
Kundan Munjal and Rekha Bhatia.
\newblock A systematic review of homomorphic encryption and its contributions in healthcare industry.
\newblock {\em Complex \& Intelligent Systems}, 9(4):3759--3786, 2023.

\bibitem{zhao2022survey}
Ying Zhao and Jinjun Chen.
\newblock A survey on differential privacy for unstructured data content.
\newblock {\em ACM Computing Surveys (CSUR)}, 54(10s):1--28, 2022.

\bibitem{abspoel2021secure}
Mark Abspoel, Daniel Escudero, and Nikolaj Volgushev.
\newblock Secure training of decision trees with continuous attributes.
\newblock {\em Proceedings on Privacy Enhancing Technologies}, 2021.

\bibitem{haque2020privacy}
Rakib~Ul Haque, ASM~Touhidul Hasan, Qingshan Jiang, and Qiang Qu.
\newblock Privacy-preserving k-nearest neighbors training over blockchain-based encrypted health data.
\newblock {\em Electronics}, 9(12):2096, 2020.

\bibitem{aono2017input}
Yoshinori Aono, Takuya Hayashi, Le~Trieu Phong, and Lihua Wang.
\newblock Input and output privacy-preserving linear regression.
\newblock {\em IEICE TRANSACTIONS on Information and Systems}, 100(10):2339--2347, 2017.

\bibitem{mohassel2017secureml}
Payman Mohassel and Yupeng Zhang.
\newblock Secureml: A system for scalable privacy-preserving machine learning.
\newblock In {\em 2017 IEEE symposium on security and privacy (SP)}, pages 19--38. IEEE, 2017.

\bibitem{knott2021crypten}
Brian Knott, Shobha Venkataraman, Awni Hannun, Shubho Sengupta, Mark Ibrahim, and Laurens van~der Maaten.
\newblock Crypten: Secure multi-party computation meets machine learning.
\newblock {\em Advances in Neural Information Processing Systems}, 34:4961--4973, 2021.

\bibitem{mishra2020delphi}
Pratyush Mishra, Ryan Lehmkuhl, Akshayaram Srinivasan, Wenting Zheng, and Raluca~Ada Popa.
\newblock Delphi: A cryptographic inference system for neural networks.
\newblock In {\em Proceedings of the 2020 Workshop on Privacy-Preserving Machine Learning in Practice}, pages 27--30, 2020.

\bibitem{braun2022motion}
Lennart Braun, Daniel Demmler, Thomas Schneider, and Oleksandr Tkachenko.
\newblock Motion--a framework for mixed-protocol multi-party computation.
\newblock {\em ACM Transactions on Privacy and Security}, 25(2):1--35, 2022.

\bibitem{liu2024pencil}
Xuanqi Liu, Zhuotao Liu, Qi~Li, Ke~Xu, and Mingwei Xu.
\newblock Pencil: Private and extensible collaborative learning without the non-colluding assumption.
\newblock {\em arXiv preprint arXiv:2403.11166}, 2024.

\bibitem{libOTe}
Lance~Roy Peter~Rindal.
\newblock {libOTe: an efficient, portable, and easy to use Oblivious Transfer Library}.
\newblock \url{https://github.com/osu-crypto/libOTe}.

\bibitem{riazi2018chameleon}
M~Sadegh Riazi, Christian Weinert, Oleksandr Tkachenko, Ebrahim~M Songhori, Thomas Schneider, and Farinaz Koushanfar.
\newblock Chameleon: A hybrid secure computation framework for machine learning applications.
\newblock In {\em Proceedings of the 2018 on Asia conference on computer and communications security}, pages 707--721, 2018.

\bibitem{du2002practical}
Wenliang Du and Zhijun Zhan.
\newblock A practical approach to solve secure multi-party computation problems.
\newblock In {\em Proceedings of the 2002 workshop on New security paradigms}, pages 127--135, 2002.

\bibitem{mohassel2018aby3}
Payman Mohassel and Peter Rindal.
\newblock Aby3: A mixed protocol framework for machine learning.
\newblock In {\em Proceedings of the 2018 ACM SIGSAC conference on computer and communications security}, pages 35--52, 2018.

\bibitem{kumar2017privacy}
Malay Kumar, Jasraj Meena, and Manu Vardhan.
\newblock Privacy preserving, verifiable and efficient outsourcing algorithm for matrix multiplication to a malicious cloud server.
\newblock {\em Cogent Engineering}, 4(1):1295783, 2017.

\bibitem{van2023privacy}
Florian Van~Daalen, Lianne Ippel, Andre Dekker, and Inigo Bermejo.
\newblock Privacy preserving $ n $ n-party scalar product protocol.
\newblock {\em IEEE Transactions on Parallel and Distributed Systems}, 34(4):1060--1066, 2023.

\bibitem{keller2020mp}
Marcel Keller.
\newblock Mp-spdz: A versatile framework for multi-party computation.
\newblock In {\em Proceedings of the 2020 ACM SIGSAC conference on computer and communications security}, pages 1575--1590, 2020.

\bibitem{benaissa2021tenseal}
Ayoub Benaissa, Bilal Retiat, Bogdan Cebere, and Alaa~Eddine Belfedhal.
\newblock Tenseal: A library for encrypted tensor operations using homomorphic encryption.
\newblock {\em arXiv preprint arXiv:2104.03152}, 2021.

\bibitem{FATE}
{FederatedAI}.
\newblock Fate: An industrial grade federated learning framework.
\newblock \url{https://github.com/FederatedAI/FATE}.

\bibitem{ma2023secretflow}
Junming Ma, Yancheng Zheng, Jun Feng, Derun Zhao, Haoqi Wu, Wenjing Fang, Jin Tan, Chaofan Yu, Benyu Zhang, and Lei Wang.
\newblock $\{$SecretFlow-SPU$\}$: A performant and $\{$User-Friendly$\}$ framework for $\{$Privacy-Preserving$\}$ machine learning.
\newblock In {\em 2023 USENIX Annual Technical Conference (USENIX ATC 23)}, pages 17--33, 2023.

\bibitem{wagh2019securenn}
Sameer Wagh, Divya Gupta, and Nishanth Chandran.
\newblock Securenn: 3-party secure computation for neural network training.
\newblock {\em Proceedings on Privacy Enhancing Technologies}, 2019.

\bibitem{rabin2005exchange}
Michael~O Rabin.
\newblock How to exchange secrets with oblivious transfer.
\newblock {\em Cryptology ePrint Archive}, 2005.

\bibitem{yadav2022survey}
Vijay~Kumar Yadav, Nitish Andola, Shekhar Verma, and S~Venkatesan.
\newblock A survey of oblivious transfer protocol.
\newblock {\em ACM Computing Surveys (CSUR)}, 54(10s):1--37, 2022.

\bibitem{beaver1997commodity}
Donald Beaver.
\newblock Commodity-based cryptography.
\newblock In {\em Proceedings of the twenty-ninth annual ACM symposium on Theory of computing}, pages 446--455, 1997.

\bibitem{beaver1998server}
Donald Beaver.
\newblock Server-assisted cryptography.
\newblock In {\em Proceedings of the 1998 workshop on New security paradigms}, pages 92--106, 1998.

\bibitem{du2001privacy}
Wenliang Du and Mikhail~J Atallah.
\newblock Privacy-preserving cooperative statistical analysis.
\newblock In {\em Seventeenth Annual Computer Security Applications Conference}, pages 102--110. IEEE, 2001.

\bibitem{atallah2001secure}
Mikhail~J Atallah and Wenliang Du.
\newblock Secure multi-party computational geometry.
\newblock In {\em Workshop on Algorithms and Data Structures}, pages 165--179. Springer, 2001.

\bibitem{bogdanov2008sharemind}
Dan Bogdanov, Sven Laur, and Jan Willemson.
\newblock Sharemind: A framework for fast privacy-preserving computations.
\newblock In {\em Computer Security-ESORICS 2008: 13th European Symposium on Research in Computer Security, M{\'a}laga, Spain, October 6-8, 2008. Proceedings 13}, pages 192--206. Springer, 2008.

\bibitem{tan2021cryptgpu}
Sijun Tan, Brian Knott, Yuan Tian, and David~J Wu.
\newblock Cryptgpu: Fast privacy-preserving machine learning on the gpu.
\newblock In {\em 2021 IEEE Symposium on Security and Privacy (SP)}, pages 1021--1038. IEEE, 2021.

\bibitem{miller2021simple}
Jason~Mathew Miller, Logan~H Harbour, Robert~W Carlsen, Andrew~E Slaughter, Brandon~Samuel Biggs~Jr, and Cody~J Permann.
\newblock Simple, secure, internet delivery of moose-based applications.
\newblock Technical report, Idaho National Lab.(INL), Idaho Falls, ID (United States), 2021.

\bibitem{furukawa2017high}
Jun Furukawa, Yehuda Lindell, Ariel Nof, and Or~Weinstein.
\newblock High-throughput secure three-party computation for malicious adversaries and an honest majority.
\newblock In {\em Annual international conference on the theory and applications of cryptographic techniques}, pages 225--255. Springer, 2017.

\bibitem{guo2020efficient}
Chun Guo, Jonathan Katz, Xiao Wang, and Yu~Yu.
\newblock Efficient and secure multiparty computation from fixed-key block ciphers.
\newblock In {\em 2020 IEEE Symposium on Security and Privacy (SP)}, pages 825--841. IEEE, 2020.

\bibitem{10.1145/3195970.3196023}
Bita~Darvish Rouhani, M.~Sadegh Riazi, and Farinaz Koushanfar.
\newblock Deepsecure: scalable provably-secure deep learning.
\newblock In {\em Proceedings of the 55th Annual Design Automation Conference}, DAC '18, New York, NY, USA, 2018. Association for Computing Machinery.

\bibitem{kolesnikov2008improved}
Vladimir Kolesnikov and Thomas Schneider.
\newblock Improved garbled circuit: Free xor gates and applications.
\newblock In {\em Automata, Languages and Programming: 35th International Colloquium, ICALP 2008, Reykjavik, Iceland, July 7-11, 2008, Proceedings, Part II 35}, pages 486--498. Springer, 2008.

\bibitem{keller2018overdrive}
Marcel Keller, Valerio Pastro, and Dragos Rotaru.
\newblock Overdrive: Making spdz great again.
\newblock In {\em Annual International Conference on the Theory and Applications of Cryptographic Techniques}, pages 158--189. Springer, 2018.

\bibitem{baum2019using}
Carsten Baum, Daniele Cozzo, and Nigel~P Smart.
\newblock Using topgear in overdrive: a more efficient zkpok for spdz.
\newblock In {\em International Conference on Selected Areas in Cryptography}, pages 274--302. Springer, 2019.

\bibitem{7958569}
Payman Mohassel and Yupeng Zhang.
\newblock Secureml: A system for scalable privacy-preserving machine learning.
\newblock In {\em 2017 IEEE Symposium on Security and Privacy (SP)}, pages 19--38, 2017.

\bibitem{atallah2010securely}
Mikhail~J Atallah and Keith~B Frikken.
\newblock Securely outsourcing linear algebra computations.
\newblock In {\em Proceedings of the 5th ACM Symposium on Information, Computer and Communications Security}, pages 48--59, 2010.

\bibitem{vaidya2002privacy}
Jaideep Vaidya and Chris Clifton.
\newblock Privacy preserving association rule mining in vertically partitioned data.
\newblock In {\em Proceedings of the eighth ACM SIGKDD international conference on Knowledge discovery and data mining}, pages 639--644, 2002.

\bibitem{nikolaenko2013privacy}
Valeria Nikolaenko, Udi Weinsberg, Stratis Ioannidis, Marc Joye, Dan Boneh, and Nina Taft.
\newblock Privacy-preserving ridge regression on hundreds of millions of records.
\newblock In {\em 2013 IEEE symposium on security and privacy}, pages 334--348. IEEE, 2013.

\bibitem{gascon2016privacy}
Adri{\`a} Gasc{\'o}n, Phillipp Schoppmann, Borja Balle, Mariana Raykova, Jack Doerner, Samee Zahur, and David Evans.
\newblock Privacy-preserving distributed linear regression on high-dimensional data.
\newblock {\em Cryptology ePrint Archive}, 2016.

\bibitem{giacomelli2018privacy}
Irene Giacomelli, Somesh Jha, Marc Joye, C~David Page, and Kyonghwan Yoon.
\newblock Privacy-preserving ridge regression with only linearly-homomorphic encryption.
\newblock In {\em Applied Cryptography and Network Security: 16th International Conference, ACNS 2018, Leuven, Belgium, July 2-4, 2018, Proceedings 16}, pages 243--261. Springer, 2018.

\bibitem{gilad2019secure}
Ran Gilad-Bachrach, Kim Laine, Kristin Lauter, Peter Rindal, and Mike Rosulek.
\newblock Secure data exchange: A marketplace in the cloud.
\newblock In {\em Proceedings of the 2019 ACM SIGSAC Conference on Cloud Computing Security Workshop}, pages 117--128, 2019.

\bibitem{rathee2020cryptflow2}
Deevashwer Rathee, Mayank Rathee, Nishant Kumar, Nishanth Chandran, Divya Gupta, Aseem Rastogi, and Rahul Sharma.
\newblock Cryptflow2: Practical 2-party secure inference.
\newblock In {\em Proceedings of the 2020 ACM SIGSAC Conference on Computer and Communications Security}, pages 325--342, 2020.

\bibitem{evans2018pragmatic}
David Evans, Vladimir Kolesnikov, Mike Rosulek, et~al.
\newblock A pragmatic introduction to secure multi-party computation.
\newblock {\em Foundations and Trends{\textregistered} in Privacy and Security}, 2(2-3):70--246, 2018.

\bibitem{goldreich2004foundations}
Oded Goldreich.
\newblock {\em Foundations of Cryptography, Volume 2}.
\newblock Cambridge university press Cambridge, 2004.

\bibitem{lindell2017simulate}
Yehuda Lindell.
\newblock How to simulate it--a tutorial on the simulation proof technique.
\newblock {\em Tutorials on the Foundations of Cryptography: Dedicated to Oded Goldreich}, pages 277--346, 2017.

\bibitem{du2004privacy}
Wenliang Du, Yunghsiang~S Han, and Shigang Chen.
\newblock Privacy-preserving multivariate statistical analysis: Linear regression and classification.
\newblock In {\em Proceedings of the 2004 SIAM international conference on data mining}, pages 222--233. SIAM, 2004.

\bibitem{motwani1995randomized}
Rajeev Motwani.
\newblock {\em Randomized Algorithms}.
\newblock Cambridge University Press, 1995.

\bibitem{freivalds1979fast}
R{\=u}si{\c{n}}{\v{s}} Freivalds.
\newblock Fast probabilistic algorithms.
\newblock In {\em International Symposium on Mathematical Foundations of Computer Science}, pages 57--69. Springer, 1979.

\bibitem{bellman1997introduction}
Richard Bellman.
\newblock {\em Introduction to matrix analysis}.
\newblock SIAM, 1997.

\bibitem{10.5555/874063.875553}
R.~Canetti.
\newblock Universally composable security: A new paradigm for cryptographic protocols.
\newblock In {\em Proceedings of the 42nd IEEE Symposium on Foundations of Computer Science}, FOCS '01, page 136, USA, 2001. IEEE Computer Society.

\bibitem{chung2006computer}
C~Chung-Kuan.
\newblock Computer arithmetic algorithms and hardware design.
\newblock {\em Lecture notes, University of California, San Diego, La Jolla, CA}, 2006.

\bibitem{el2002inversion}
Laurent El~Ghaoui.
\newblock Inversion error, condition number, and approximate inverses of uncertain matrices.
\newblock {\em Linear algebra and its applications}, 343:171--193, 2002.

\bibitem{fasi2023matrix}
Massimiliano Fasi, Nicholas~J Higham, Florent Lopez, Theo Mary, and Mantas Mikaitis.
\newblock Matrix multiplication in multiword arithmetic: Error analysis and application to gpu tensor cores.
\newblock {\em SIAM Journal on Scientific Computing}, 45(1):C1--C19, 2023.

\bibitem{sauer2011numerical}
Timothy Sauer.
\newblock {\em Numerical analysis}.
\newblock Addison-Wesley Publishing Company, 2011.

\end{thebibliography}
\end{sloppypar}
\end{document}